\documentclass[final]{siamart251104}

\usepackage[utf8]{inputenc}
\usepackage[english]{babel}

\usepackage{algorithmicx}
\usepackage{algpseudocode}

\usepackage[section]{placeins}

\newsiamremark{remark}{Remark}

\usepackage[sort]{cite}
\usepackage{makecell}
\usepackage[normalem]{ulem}
\usepackage{tikz}
\usetikzlibrary{quotes}

\tikzset{
znode/.style={
  circle,
  inner sep=0pt,
  text width=6mm,
  align=center,
  draw=black,
  fill=white,
  scale=0.8
  }
}

\tikzset{
xnode/.style={
  circle,
  inner sep=0pt,
  text width=6mm,
  align=center,
  draw=black,
  fill=gray!40,
  scale=0.8
  }
}

\tikzset{
hnode/.style={
  rectangle,
  draw=black,
  fill=white
  }
}

\usepackage{amssymb,bm}
\usepackage{physics}

\newcommand*{\ceil}[1]{{\left\lceil{#1}\right\rceil}}
\newcommand*{\floor}[1]{{\left\lfloor{#1}\right\rfloor}}
\newcommand*{\mat}[1]{{#1}}
\newcommand*{\B}[1]{{#1}}

\newcommand*{\C}{{\mathbb{C}}}
\newcommand*{\R}{{\mathbb{R}}}
\newcommand*{\Q}{{\mathbb{Q}}}
\newcommand*{\Z}{{\mathbb{Z}}}
\newcommand*{\F}{{\mathbb{F}}}

\newcommand*{\J}{{\mathbb{J}}}
\newcommand*{\W}{{\mathbb{W}}}

\newcommand*{\N}{{\mathbb{N}}}

\newcommand*{\otherwise}{{\text{otherwise}}}
\newcommand*{\and}{{\text{and}}}
\newcommand*{\where}{{\text{where}}}
\newcommand*{\actson}[2]{\ensuremath{{#1}^{\langle #2 \rangle}}}

\makeatletter
\def\set{\@ifstar\@setpred\@setenum}
\newcommand*{\@setpred}[2]{\left\{{#1}\;\middle|\;{#2}\right\}}
\newcommand*{\@setenum}[1]{\left\{{#1}\right\}}
\makeatother

\newcommand*{\binaryset}{\set{0,1}}

\DeclareMathOperator{\spn}{span}
\DeclareMathOperator{\nnz}{nnz}

\newcommand*{\rxp}{\hat{\mat{X}}}
\newcommand*{\rxm}{\check{\mat{X}}}
\newcommand*{\rzp}{\hat{\mat{Z}}}
\newcommand*{\rzm}{\check{\mat{Z}}}

\newcommand{\symblock}[3]{
\ifnum#2>#3
    {{#1}_{{#3}{#2}}^T}
\else
    {{#1}_{{#2}{#3}}}
\fi
}

\crefformat{footnote}{#2\footnotemark[#1]#3}

\title{Simulating Clifford Circuits with Gaussian Elimination}

\author{
Yuchen Pang\thanks{Siebel School of Computing and Data Science, University of Illinois Urbana-Champaign, Urbana, IL 61801, USA; Computer Network Information Center, Chinese Academy of Sciences, Beijing 100083, China (\email{ycpang@cnic.cn})} \and
Edgar Solomonik\thanks{Siebel School of Computing and Data Science, University of Illinois Urbana-Champaign, Urbana, IL 61801, USA (\email{solomon2@illinois.edu})}
}

\begin{document}
\maketitle

\begin{abstract}
Quantum circuits are considered more powerful than classical circuits and require exponential resources to simulate classically. Clifford circuits are a special class of quantum circuits that can be simulated in polynomial time but still show important quantum effects such as entanglement. In this work, we present an algorithm that simulates Clifford circuits by performing Gaussian elimination on a modified adjacency matrix derived from the circuit structure. Our work builds on a ZX-calculus tensor network representation of Clifford circuits that reduces to quantum graph states. We give a concise formula of amplitudes of graph states based on the LDL decomposition of matrices over GF(2), and use it to get efficient algorithms for strong and weak simulation of Clifford circuits using a tree-decomposition-based sparse LDL algorithm.
The complexity of our algorithm matches the state of the art for weak graph state simulation and improves the state of the art for strong graph state simulation by taking advantage of Strassen-like fast matrix multiplication. Our algorithm is also efficient when computing many amplitudes or samples of a Clifford circuit.
Further, our amplitudes formula provides a new characterization of locally Clifford equivalent graph states as well as an efficient protocol to learn graph states with low-rank adjacency matrices.

\end{abstract}

\begin{keywords}
Quantum computation, Clifford circuit, ZX-calculus, Tree decomposition, Gauss--Jordan elimination, Locally Clifford equivalent graph states, Graph state learning
\end{keywords}

\begin{MSCcodes}
68Q12, 81P68, 15A09, 05C50
\end{MSCcodes}

\tableofcontents

\section{Introduction}
Circuits consisting of Clifford gates (e.g. Hadamard gate, CZ gate, and the phase gate $\mat{R}_Z(\pi/2)$) are of wide interest in quantum computing.
This gate set would be universal for quantum computing provided the inclusion of T gates (i.e. the phase gate $\mat{R}_Z(\pi/4)$).
Despite this restriction, circuits consisting of only Clifford gates (Clifford circuits) suffice to express important quantum circuits, such as stabilizer codes for error correction, and graph states, which have been shown to provide an asymptotic depth advantage over any classical circuit~\cite{bravyi2018quantum} (the first quantum advantage result).
In addition, Clifford circuits are known to admit polynomial-time simulation algorithms.
This celebrated result is known as the Gottesman-Knill theorem~\cite{gottesman1998heisenberg,aaronson2004improved}.
A mathematical intuition for this result is that any $n$-qubit Clifford state can be associated with an $n$-dimensional quadratic form~\cite{dehaene2003clifford}.

The problem of Clifford circuit simulation can be defined as follows.
\begin{definition}[Clifford circuit simulation]\label{def:circ-sim}
Let $\mat{U}$ be an $n$-qubit Clifford circuit. The \textit{strong simulation} of $\mat{U}$ takes any $\mat{x}\in\set{0,1}^n$ as input, and outputs $\mel{\mat{x}}{\mat{U}}{0^n}$. The \textit{weak simulation} of $\mat{U}$ samples $\mat{x}\in\set{0,1}^n$ according to the distribution $P(\mat{x})=\abs{\mel{\mat{x}}{\mat{U}}{0^n}}^2$.
\end{definition}
Here we adopt the bra-ket notation from physics, where $\bra{\mat{x}}$ and $\ket{0^n}$ can be viewed as row and column vectors of size $2^n$ respectively and $\mat{U}$ can be viewed as a $2^n$-by-$2^n$ matrix. Then $\mel{\mat{x}}{\mat{U}}{0^n}$ represents the inner product between $\bra{\mat{x}}$ and $\mat{U}\ket{0^n}$. However, as mentioned above, such matrices and vectors in Clifford circuits admit polynomial-sized representations due to the Gottesman-Knill theorem.

The Clifford circuit simulation problem is closely related to basic linear algebra problems such as matrix-matrix multiplication and linear system solving.
State of the art simulation algorithms for Clifford circuits~\cite{qassim2021improved,gosset2024fast,gidney2021stim} are based on an algebraic formalism known as stabilizer tableaux~\cite{bravyi2016improved}.
A tableau encodes an $n$-qubit Clifford unitary using a pair of matrices and vectors of size $2n$. It has been shown that all Clifford circuit operations such as gate applications and measurements could all be implemented as standard matrix operations on tableaux, where the most costly part is matrix-matrix products~\cite{gosset2024fast}.
From the perspective of complexity, it is known that simulating Clifford circuits is complete for the complexity class $\oplus L$~\cite{aaronson2004improved}, which is the complexity class of solving linear systems over $\F_2$~\cite{damm1990problems}.
An algorithm to solve symmetric linear systems over $\F_2$ with zero diagonals via graph state simulation is proposed by Gosset et al.~\cite[Theorem 26]{gosset2024fast}.
Our work can be viewed as the other direction of this reduction. The simulation algorithms below use the sparse LDL algorithm of~\cite{solomonik2025fast}, a companion preprint by one of the authors whose interface we state explicitly in \cref{def:sparse-ldl-setting,lem:gj-complexity}; every time-complexity claim rests on it. In particular, we provide a complete reduction from Clifford circuit simulation problem to the LDL decomposition of a symmetric matrix over $\F_2$ and give time complexities of circuit simulation based on the sparse LDL algorithm of~\cite{solomonik2025fast}. Our algorithm performs better than previous algorithms when calculating many amplitudes/samples by utilizing the Strassen-like fast matrix multiplication.

To represent the Clifford circuit as a matrix so that an LDL decomposition of it gives the simulation result,
we leverage a modern formalism for quantum circuits known as ZX-calculus.
In particular, we use a canonical form for circuits derived in this formalism known as a graph-like ZX-diagram~\cite{duncan2020graph}, and then an $n$-qubit circuit consisting of $m$ gates could be converted to a \textit{phased graph state} with $N=O(m+n)$ vertices. This reduction is described in \cref{sec:circ-as-pgs}.

The core of this paper is the analysis of phased graph states and the corresponding simulation algorithm. Phased graph states are defined as generalized graph states characterized by matrices in the set
\[
\J_n=\set*{\mat{A}\in\Z_4^{n\times n}}{\mat{A}=\mat{A}^T,\,O(\mat{A})\in\F_2^{n\times n}},
\]
where $O(\mat{A})$ denotes the off-diagonals of $\mat{A}$ and informally we treat $\Z_4=\set{0,1,2,3}$ and $\F_2=\set{0,1}$. For $\mat{A}\in\J_n$, the phased graph state $\ket{\mat{A}}$ is given by
\[
\ket{\mat{A}}=\bigotimes_{j=1}^n\rzm^{a_{jj}}\ket{G},
\]
where $\rzm=\mat{R}_z(-\pi/2)$ is a phase gate, $a_{jj}$ is the $j$'th diagonal element of $\mat{A}$,
and $\ket{G}$ is the standard graph state whose adjacency matrix is $O(\mat{A})$.

In \cref{sec:pgs-ge}, we define a generalized Gauss--Jordan elimination process on $\mat{A}\in\J_n$ and show that this process corresponds to a series of vertex and edge complementations~\cite{danielsen2005self,van2004graphical,duncan2013pivoting} of $\ket{G}$. Then we provide a simple and explicit formula for the amplitudes of Hadamard-transformed phased graph states as following.
\begin{theorem}[Informal form of \cref{thm:amp-formula}]\label{thm:amp-formula-intro}
Let $\mat{A}\in\J_n$, and write $k=\rank(\omega_1(\mat{A}))$. Up to a global phase independent of $\mat{x}$,
\begin{equation}
\bra{\mat{x}}\mat{H}^{\otimes n}\ket{\mat{A}}
=\begin{cases}
2^{-k/2}i^{-(\mat{x}+\mat{v})^T\mat{B}(\mat{x}+\mat{v})} & \mat{x}\oplus\mat{w}\in\spn(\omega_1(\mat{A})),\\
0 & \text{otherwise},
\end{cases}
\end{equation}
where $\omega_1(\mat{A})=\mat{A}\bmod2$, and $\mat{B}\in\J_n$, $\mat{v},\mat{w}\in\F_2^n$ can be calculated along with an LDL factorization of $\omega_1(\mat{A})$. The matrix $\mat{B}$ is the output of the generalized Gauss--Jordan process (\cref{alg:gj}); the vector $\mat{v}$ in this full-vector formula extends its recorded pivot bits by ones on the uneliminated coordinates. After a symmetric permutation placing a nonsingular principal $k\times k$ block $A_{11}$ first, $B$ satisfies
\begin{align}
\omega_1(B)&=\begin{bmatrix}A_{11}^{\#}&B_{21}^T\\ B_{21}&0\end{bmatrix},
\qquad B_{21}=\omega_1(A_{21}A_{11}^{\#}),
\end{align}
where $A_{11}^{\#}$ is the inverse over $\F_2$; the vectors and matrix are permuted back for inputs in the original ordering.
\end{theorem}
Prior works have observed that these amplitudes can be computed by Gaussian elimination~\cite{aaronson2004improved,guan2019stabilizer}, and the amplitude formula for these states has also been derived based on the Tutte polynomial~\cite{mann2021simulating}.
However, we are not aware of a prior algebraic expression explicitly based on Gaussian elimination.
Our formula not only allows us to derive simulation algorithms using standard linear algebra routines, but also has other theoretical implications as discussed in \cref{sec:apps-other}.

In \cref{sec:pgs-sim}, we give an algorithm to solve the phased graph state simulation problem based on \cref{thm:amp-formula-intro} and the tree-decomposition-based sparse LDL algorithm of~\cite{solomonik2025fast}. We define the problem formally in \cref{def:pgs-sim}; informally, strong simulation evaluates the amplitude $\mel{\mat{x}}{\mat{H}^{\otimes n}}{\mat{A}}$ for a given $\mat{x}\in\F_2^n$, and weak simulation samples $\mat{x}$ from the corresponding distribution over a coset $\mathcal{X}_{S,\mat{y}}$ with $S\subseteq[n]$ and $\mat{y}\in\F_2^n$.
By expressing the dominant computations in the simulation task as matrix multiplications, our algorithm exploits the sparsity structure of the graph and the efficiency of fast matrix multiplication, achieving the following time complexity.
\begin{theorem}[Summary of \cref{thm:weak-sim-pgs},
\cref{thm:strong-sim-pgs} and \cref{thm:strong-sim-pgs-fixed-bits}]
Consider any $y\in\F_2^n$, $S\subseteq[n]$, $A\in\J_n$ and a given
tree decomposition of the graph corresponding to $A$ with width
$\tau\ge1$. Let $\ell=n-\abs{S}$.
Then strong simulation of $A$ with $k\in\N$ inputs from
$\mathcal{X}_{S,y}$ or weak simulation of $A$ for $k$ samples from
$\mathcal{X}_{S,y}$ takes
\[
O\bigl(n\tau^{\omega-1}
+\min(kn\tau^{\omega-2},\ell n\tau^{\omega-2}+k\ell^{\omega-1})\bigr)
\]
time\footnote{$\omega$ denotes the matrix multiplication exponent.}.
In the case of weak simulation, the algorithm reports if the conditioning event
has probability zero.
\end{theorem}

In \cref{sec:apps-sim}, we apply our algorithm for phased graph state simulation to other simulation problems such as graph state simulation, Clifford circuit simulation, and Clifford+T circuit simulation. These include the following results.
In \cref{sec:gs-sim}, we solve the graph state simulation problem~\cite{gosset2024fast} in the following complexity.
\begin{theorem}[Restatement of \cref{thm:gs-sim-complexity}]\label{thm:gs-sim-complexity-intro}
Let $G$ be a graph with $n$ vertices and a tree decomposition of $G$ with width $\tau$ is given.
Then the strong/weak simulation of $\ket{G}$ for $k\in\N$ amplitudes/samples could be solved in $O(n\tau^{\omega-1}+kn\tau^{\omega-2})$ time.
\end{theorem}
When $k=1$, this matches the state of the art for weak simulation~\cite{gosset2024fast} and improves the previous $O(n\tau^{2})$ complexity for strong simulation (with phase)~\cite{kerzner2021clifford}.
Our algorithm can also efficiently compute many amplitudes/samples of a fixed graph state, which is not explicitly considered in previous literature.
In \cref{sec:clifford-sim}, we apply our algorithm to Clifford circuit simulation and get the following complexity based on treewidth of the circuit.
\begin{theorem}[Restatement of \cref{thm:clifford-complexity}]\label{thm:clifford-complexity-intro}
Let $\mat{U}$ be a Clifford circuit with $n$ qubits and $m$ gates, each of which acts on at most 2 qubits, for $m=\Omega(n)$, and a tree decomposition of its graph-like diagram (\cref{lem:circ-as-pgs}) of width $\tau$ is given. Then strong/weak simulation of $\mat{U}$ for $k\in\N$ amplitudes/samples can be solved in
$O(m\tau^{\omega-1}+\min(km\tau^{\omega-2}, mn\tau^{\omega-2}+kn^{\omega-1}))$ time.
\end{theorem}
Since every $n$-qubit circuit has a simple tree decomposition of width $O(n)$, our algorithm can simulate any Clifford circuit in $O(mn^{\omega-1}+kn^{\omega-1})$ time. We compare our algorithm with previous Clifford circuit simulation algorithms in \cref{tab:complexity-cmp} and show that our bound is asymptotically smaller than the tableau and Pauli-frame bounds when $m\gg n^{\omega-1}$ and $k\gg m$.

In \cref{sec:clifford-t-sim}, we consider the strong simulation of general Clifford+T circuits and show that the LDL factorization can be used in conjunction with stabilizer decompositions of T gates, obtaining an exact algorithm. We compare its complexity with the planar-circuit methods of~\cite{gosset2024fast} and their phase-sensitive extension~\cite{kerzner2021clifford}, obtaining improved bounds in the parameter regimes analyzed there.

In \cref{sec:apps-other}, we present other applications of the amplitudes formula in \cref{thm:amp-formula-intro}. In \cref{sec:lc-eq}, we give a characterization of locally Clifford (LC) equivalent graph states based on the generalized Gauss--Jordan process as presented in the following theorem.
\begin{theorem}[Restatement of \cref{thm:lc-eq}]\label{thm:lc-eq-intro}
Two $n$-qubit graph states with adjacency matrices $\bar A$ and $\bar B$ are locally Clifford (LC) equivalent if and only if there exist a permutation matrix $P$ (that is, a common ordering of the vertices of the two graphs), vectors $u,v\in\F_2^n$, and some $k\in\{0,\ldots,n\}$ such that, with $A=P^T\bar AP$ and $B=P^T\bar BP$,
\[
A+D(u)=\omega_1(\Gamma_k(B+D(v))),
\]
where $\Gamma_k$ is the generalized Gauss--Jordan process of \cref{alg:gj} and $k$ is admissible (the leading $k\times k$ block of $B+D(v)$ is nonsingular and admits an unpivoted LDL). Equivalently, $A+D(u)$ is the principal pivot transform of $B+D(v)$ on a nonsingular leading $k\times k$ block, whose first bits are given by the block formula in \cref{eq:gjl1}.
\end{theorem}
Based on \cref{thm:lc-eq-intro}, we give an alternative proof for a linear upper bound for each orbit of LC equivalent graph states~\cite{adcock2020mapping,claudet2025local}.
In \cref{sec:gs-learn}, we provide an efficient protocol to learn graph states with low-rank adjacency matrices.

\subsection{Related works}
Our amplitude formula in \cref{thm:amp-formula-intro} is closely related to the quadratic-form-expansion (QFE) method for Clifford circuit simulation~\cite{de2022fast}, which is covered in more detail in \cref{sec:other-sim-methods}. Here, we highlight several key differences between our algorithm and the QFE method.
Firstly, the QFE method represents a stabilizer state with a quadratic form in
\cref{eq:qfe-amplitudes} and simulates the circuit by applying gates and measurements in order. In comparison, our method views the whole circuit as a graph and exploits its overall structure. Secondly, the QFE method designs dedicated algorithms that transforms the quadratic form to implement gate applications and measurements on arbitrary stabilizer states. In contrast, our method focuses on a subset of stabilizer states that admit richer structures, namely phased graph states. This allows us to relate vertex and edge complementations to standard Gaussian elimination routines and utilize fast algorithms based on Strassen's fast matrix multiplication. This leads to the difference in time complexities as shown in \cref{tab:complexity-cmp}.

Also, the phased graph state simulation problem we solve can be viewed as a variant of the graph state simulation problem studied in \cite{gosset2024fast} and \cite{kerzner2021clifford}. Our algorithm matches the complexity of theirs for weak simulation and improves the time complexity for strong simulation as discussed in \cref{sec:gs-sim}.
In addition, Gosset et al.~\cite{gosset2024fast,kerzner2021clifford} also consider the Clifford circuit simulation problem by converting the circuit to a graph state in a similar way as ours. In particular, they consider the connectivity graph of qubits and bound the treewidth of depth-$d$ circuits using the treewidth of the qubit-connectivity graph. In this work, we focus on the graph structure of circuits and give a specific circuit-to-graph reduction that allows the reduction from circuit simulation to LDL decomposition. This leads to algorithms for different scenarios with different complexities, which are further discussed in \cref{sec:gs-sim} and \cref{sec:clifford-sim}.

\section{Preliminaries and Notations}\label{sec:preliminaries}

\subsection{Matrix-related notations}
We use uppercase letters such as $\mat{A},\mat{B},\mat{U}$ to represent matrices, operators, and quantum circuits.
Lowercase letters like $\mat{x},\mat{y}$ represent vectors.
Lowercase letters with subscripts are used to denote elements in the corresponding matrix or vector. For example, $v_i$ denotes the $i$'th element in $\mat{v}$, and $a_{ij}$ denotes the element in the $i$'th row and $j$'th column in $\mat{A}$.

To specify the range of indices, we define $[n]=\set*{i\in\Z}{1\le i\le n}$.
Given $S\subseteq[n]$ and $\mat{y}\in\F_2^n$, we define $\mathcal{X}_{S,\mat{y}}=\set*{\mat{x}\in\F_2^n}{\forall i\in S,\,x_i=y_i}$ to represent vectors that agree with $\mat{y}$ over the indices in $S$. In addition, we use $\F^S$ to denote the set of vectors over $\F$ indexed by elements in $S$. We write $\mat{x}_S$ for the restriction of a vector $\mat{x}$ to the indices in $S$.

The identity matrix is represented by $\mat{I}$ when its size is clear in the context. Otherwise, we use $\mat{I}_n$ to denote the identity matrix of size $n$. The elementary vector $\mat{e}_i$ indicates the $i$'th vector in standard basis, i.e. the $i$'th row/column vector of $\mat{I}$. The length of $\mat{e}_i$ should always be clear in the context.

The vector given by the diagonal of $\mat A$ is represented by $d(\mat A)$, while $D(\mat{A})$ refers to the matrix that only contains the diagonal part of $\mat{A}$. Additionally, we use $O(\mat{A})$ to denote the matrix that contains only the off-diagonal part of $\mat{A}$. For a vector in $v\in \Z^{n}$, let $D(v)$ be the diagonal matrix with diagonal equal to $v$.

The transpose of a matrix $\mat{A}$ is denoted as $\mat{A}^T$ and the conjugate transpose of a complex matrix $\mat{A}$ is denoted as $\mat{A}^\dagger$.

To describe a local operator $\mat{U}: V\rightarrow V$ acting on the $j$'th component of a tensor product space $V^{\otimes n}$, we define $\actson{\mat{U}}{j}=\mat{I}\otimes\dots\otimes\mat{U}\otimes\dots\otimes\mat{I}$, where $\mat{U}$ is the $j$'th operator and all other $n-1$ operators are identities. Similarly, we use $\actson{\mat{W}}{i,j}$ to denote an operator $\mat{\mat{W}}: V\otimes V\rightarrow V\otimes V$ that acts on the $i$'th and $j$'th component of $V^{\otimes n}$.
Also, for $\mat{v}\in\Z^n$, we define $\mat{U}^\mat{v}=\bigotimes_{j=1}^n\mat{U}^{v_j}$.

The matrix multiplication exponent is denoted by $\omega$. That is, we assume the cost of multiplying two matrices of size $n\times n$ is $O(n^\omega)$, where $2<\omega<2.372$~\cite{strassen1969gaussian,williams2024new}.

\subsection{Matrices over $\Z_4$ and $\F_2$}

Let $\Z_4$ be the ring of integers modulo 4 with elements $\set{0,1,2,3}$.
Let $\omega_i(x)=\floor{x/2^{i-1}}\bmod 2$ for any integer scalar, vector, or matrix $x$ over $\Z_4$, which calculates the $i$'th bit of $x$ in binary representation.

Let $\F_2=\set{0,1}\subset\Z_4$ and define $a\oplus b=\omega_1(a+b)$. Then $(\F_2,\oplus,\cdot)$ is the finite field of size 2.
For $\mat{A}\in\F_2^{n\times n}$, we use $\mat{A}^\#$ to denote the inverse of $\mat{A}$ over $\F_2$ (by $A^{-1}$ we will have in mind the inverse of $A\in\Z^{n\times n}$ over $\Q$ if it exists).

For matrix $\mat{A}\in \F_2^{n\times n}$, we use $\rank(\mat{A})$ and  $\spn(\mat{A})$ to refer to the rank and column span of $\mat{A}$ over $\F_2$.

The LDL decomposition over $\F_2$ is defined as following to characterize the Gaussian elimination process on symmetric matrices over $\F_2$.
\begin{definition}[LDL decomposition over $\F_2$]\label{def:gf2-ldl}
Let $\mat{A}\in\F_2^{n\times n}$ be symmetric. Then the triple $(\mat{P},\mat{L},\mat{D})$ is an  LDL decomposition of $\mat{A}$ over $\F_2$ if
\begin{enumerate}
\item $\mat{A}=\omega_1(\mat{P}\mat{L}\mat{D}\mat{L}^T\mat{P}^T)$;
\item $\mat{P}\in\F_2^{n\times n}$ is a permutation matrix;
\item $\mat{L}\in\F_2^{n\times n}$ is unit-diagonal and lower-triangular;
\item $\mat{D}\in\F_2^{n\times n}$ is block-diagonal with blocks 0, 1, or $\mqty[\pmat{x}]$;
\item For each $i$ such that $\begin{bmatrix}
d_{ii} & d_{i(i+1)} \\
d_{(i+1)i} & d_{(i+1)(i+1)}
\end{bmatrix}=\begin{bmatrix}
0 & 1\\
1 & 0
\end{bmatrix}$, $l_{(i+1)i}=0$.
\end{enumerate}
\end{definition}
An LDL decomposition of a symmetric matrix always exists and can be constructed via standard Gaussian elimination.
The use of such a decomposition in applications dates back to work on integer factoring in the 1970s~\cite{lamacchia1991solving,morrison1975method}.

When $A$ is not full rank, we sometimes use the reduced LDL decomposition of it.
\begin{definition}[Reduced LDL decomposition over $\F_2$]\label{def:gf2-reduced-ldl}
Let $\mat{A}\in\F_2^{n\times n}$ be symmetric with rank $r$. Then $(\mat{P},\mat{L},\mat{D})$ is a
reduced LDL decomposition of $\mat{A}$ over $\F_2$, where $\mat{P}$, $\mat{L}$, $\mat{D}$ are defined the same as in \cref{def:gf2-ldl}, except that
\begin{itemize}
\item $\mat{L}\in\F_2^{n\times r}$ is unit-diagonal and lower-trapezoidal;
\item $\mat{D}\in\F_2^{r\times r}$ is block-diagonal with blocks 1, or $\mqty[\pmat{x}]$.
\end{itemize}
\end{definition}

\subsection{Tree decomposition and treewidth}
The complexity results in this work rely on Gaussian elimination algorithms based on tree decompositions~\cite{solomonik2025fast}, which is introduced below.
\begin{definition}[Tree decomposition]
A tree decomposition of a graph $G=(V,E)$ is a tree $\mathcal{T}=(V_{\mathcal{T}},E_{\mathcal{T}})$ such that each node $B\in V_{\mathcal{T}}$ is a bag, $B_t\subseteq V$, and the following conditions are satisfied,
\begin{itemize}
\item For each $v\in V$, there exists $B\in V_{\mathcal{T}}$ such that $v\in B$,
\item For each $(u,v)\in E$, there exists $B\in V_{\mathcal{T}}$ such that $\set{u,v}\subseteq B$;
\item For each $v\in V$, $\set*{B\in V_{\mathcal{T}}}{v\in B}$ forms a connected subtree of $\mathcal{T}$.
\end{itemize}
Also, the width of $\mathcal{T}$ is defined as $(\max_{B\in V_{\mathcal{T}}}\abs{B}-1)$.
\end{definition}
The treewidth of a graph could then be defined as follows.
\begin{definition}[Treewidth]
The treewidth of a graph $G$, or $tw(G)$, is the minimum width over all tree decompositions of $G$.
\end{definition}

\subsection{Quantum circuits}
In this section, we introduce the necessary definitions and concepts about quantum circuits that are used in this paper. For a more detailed introduction to quantum computation and related topics, we refer our readers to \cite{nielsen2010quantum}.

Qubits are quantum analog of bits in classical computation. Each qubit could be defined as a normalized vector in $\C^{2}$ up to a complex phase factor, while an $n$-qubit quantum state is defined as a normalized vector in $n$-fold tensor product of $\C^{2}$ up to a phase factor. By convention, we will use $\ket{\psi}$ to represent such vectors, where $\psi$ is a label depending on the context, and use $\bra{\psi}$ to represent the conjugate transpose of $\ket{\psi}$. For example, some commonly used quantum states are denoted as
\begin{align*}
\ket{0}=\mqty[1\\0],\quad\ket{1}=\mqty[0\\1],\quad
\ket{+}=\frac{1}{\sqrt2}\mqty[1\\1],\quad\ket{-}=\frac{1}{\sqrt2}\mqty[1\\-1].
\end{align*}
The states $\ket{0}$ and $\ket{1}$ are called the \textit{computational basis} for a qubit. For multiple qubits, the computational basis states are obtained via tensor products, which we express with the notation $\ket{\mat{b}}=\ket{b_1}\otimes\dots\otimes\ket{b_n}$ for $\mat{b}\in\binaryset^n$.

A quantum circuit (without measurements) describes quantum computation with a series of quantum gates, where each gate is a unitary acting on a subset of qubits.
Therefore, a quantum gate on $n$ qubits can be represented as a $2^n$-by-$2^n$ matrix. For example, some quantum gates used in this paper are
\begin{gather*}
    \mat{X}=\mqty[\paulixmatrix],\quad
    \mat{Y}=\mqty[\pauliymatrix],\quad
    \mat{Z}=\mqty[\paulizmatrix],\quad
    \mat{H}=\frac{1}{\sqrt2}\mqty[1&1\\1&-1],\\
    \mat{R}_z(\theta)=\mqty[1&0\\0&e^{i\theta}],\quad
    \mat{R}_x(\theta)=\mat{H}\mat{R}_z(\theta)\mat{H},\quad
    \mat{U}_{CZ}=\left[\begin{smallmatrix}
    1 & 0 & 0 & 0\\
    0 & 1 & 0 & 0\\
    0 & 0 & 1 & 0\\
    0 & 0 & 0 & -1
    \end{smallmatrix}\right],
\end{gather*}
where $\mat{X}$, $\mat{Y}$ and $\mat{Z}$ are known as Pauli matrices and $\mat{H}$ is known as the Hadamard gate.
Note that similar to quantum states, quantum gates that differ by a phase factor are equivalent, i.e. $\mat{U}$ and $e^{i\theta}\mat{U}$ for $\theta\in\R$ are equivalent.

Also, we define the following abbreviation,
\begin{gather*}
\rzp=\mat{R}_z(\pi/2),\;
\rzm=\mat{R}_z(-\pi/2),\;
\rxp=\mat{R}_x(\pi/2),\;
\rxm=\mat{R}_x(-\pi/2).
\end{gather*}
Note that $\rzp=\rzm^{-1}$ and $\rzp^2=\rzm^2=\mat{Z}$, and similar relations hold for $\rxp$, $\rxm$ and $\mat{X}$. Throughout, $\hat{X},\hat{Z}$ denote the gates $\rxp=R_x(\pi/2)$ and $\rzp=R_z(\pi/2)$, and $\check X=R_x(-\pi/2)$, $\check Z=R_z(-\pi/2)$ their inverses, while the undecorated $\mat X,\mat Z$ denote the Pauli matrices. All scalar prefactors in the identities of \cref{sec:pgs-ge} are exact for the convention $R_z(\theta)=\begin{bmatrix}1&0\\0&e^{i\theta}\end{bmatrix}$ (so that $\rzp=\begin{bmatrix}1&0\\0&i\end{bmatrix}$), with $\check X=\mat H\check Z\mat H$ (so that $\check X^2=\mat X$; this fixes the global phase, which differs from that of $\exp(-i\theta X/2)$).

Measuring an $n$-qubit quantum state $\ket{\psi}$ in computational basis corresponds to sample a binary string $\mat{b}\in\binaryset^n$ from the distribution $P(\mat{b})=\abs{\braket{\mat{b}}{\psi}}^2$.
Typically, quantum simulation algorithms work on emulating measurements (a.k.a. weak simulation), or computing one or more amplitudes (or just magnitudes of amplitudes) of a state $\ket\psi$ obtained by executing a quantum circuit (a.k.a. strong simulation). The latter could be used as a subroutine to solve the former problem~\cite{bravyi2022simulate}.

\subsection{Graph states}\label{subsec:graph-state}
Let $G=(V,E)$ be a simple undirected graph with vertices $V(G)=V$ and edges $E(G)=E$.
We use $ij$ to denote the edge $\{i,j\}\in E$. For $v\in V$, we use $N(v)$ to represent its neighbors, i.e. $N(v)=\{u\mid uv\in E\}$.
Additionally, we use $G\setminus S$ for some $S\subset V$ to represent the subgraph induced by $V\setminus S$.
The adjacency matrix of $G$ is defined as a symmetric matrix $\mat{A}\in\F_2^{n\times n}$ with zero diagonals and $a_{ij}=1$ iff $ij\in E$.
Moreover, to simplify the notations of matrices and tensor product spaces related to $G$, we assume $V$ are ordered and indexed by $[n]=\{1,\dots,n\}$.
Then graph states can be defined as follows.

\begin{definition}[Graph state]\label{def:graph-state}
Let $G=(V,E)$ be a simple undirected graph. Then the corresponding graph state is defined by
$
\ket{G}=\prod_{ij\in E}\actson{\mat{U}_{CZ}}{i,j}\ket{+}^{\otimes \abs{V}}.
$
\end{definition}

Next, we introduce two operations on graphs, namely \textit{vertex complementation} and \textit{edge complementation}, that can be implemented on graph states with only local gates. They are also known as local complementation and edge-local complementation respectively in literature~\cite{bouchet1988graphic}.
Let $\triangle$ denote symmetric difference, $K(A)=\set*{uv}{u,v\in A}$, and $K(A,B)=\set*{uv}{u\in A,v\in B}$ for some vertices sets $A$ and $B$.

\begin{theorem}[Vertex complementation~\cite{danielsen2005self,van2004graphical}]\label{thm:vertex-complementation}
Let $G=(V,E)$ be a simple undirected graph and $i\in V$ and $\tau_i(G)=(V,E\triangle K(N(i)))$ denote the local complementation of $G$ at vertex $i$. Then we have
\begin{equation*}
\ket{G}=\actson{\rxm}{i}\prod_{j\in N(i)}\actson{\rzp}{j}\ket{\tau_i(G)}.
\end{equation*}
\end{theorem}

\begin{theorem}[Edge complementation, Proposition 3.1 in~\cite{duncan2013pivoting}]\label{thm:edge-complementation}
Let $G=(V,E)$ be a simple undirected graph and $ij\in E$ and $\epsilon_{ij}(G)=\tau_i\circ\tau_j\circ\tau_i(G)$ denote the edge complementation of $G$ along $ij$. Then we have
\begin{equation*}
\ket{G}=\actson{\mat{H}}{i}\actson{\mat{H}}{j}\prod_{k\in N(i)\cap N(j)}\actson{\mat{Z}}{k}\ket{\epsilon_{ij}(G)}.
\end{equation*}
\end{theorem}
We note that $\epsilon_{ij}(G)$ could be alternatively defined as $(V,F)$, where $F$ is obtained by: (1) replacing $E$ with $E\triangle(K(A,B)\cup K(A,C)\cup K(B,C))$, where $A=N(i)\setminus N(j)$, $B=N(j)\setminus N(i)$ and $C=N(i)\cap N(j)$; (2) and exchanging the neighbors of $i$ and $j$.

Also, it is widely known that the amplitudes of a graph state in computational basis has a simple formula.
\begin{lemma}[\cite{claudet2025local}]\label{lem:graph-state-amplitudes}
Let $\ket{G}$ be an $n$-qubit graph state with adjacency matrix $\mat{A}$ and $\mat{x}\in\F_2^n$. Then
\[ \braket{\mat{x}}{G} = 2^{-n/2}(-1)^{\frac{1}{2} \B x^T\B A\B x}.\]
\end{lemma}
Note that $\mat{x}^T\mat{A}\mat{x}$ is always even since $\mat{A}$ is an adjacency matrix.

\subsection{Clifford circuits and stabilizer states}
Clifford circuits are a subset of quantum circuits that can be represented by the gate set $\{\mat{H}, \rzp, \mat{U}_{CZ}\}$. It is known that such circuits can be simulated in polynomial time on classical computers by the Gottesman-Knill theorem ~\cite{aaronson2004improved,dehaene2003clifford,gosset2024fast}.

The states generated by Clifford circuits are called \textit{stabilizer states}, since such state can be identified with
a group of $2^n$ Pauli operators that stabilize it\footnote{It means that an operator $\mat{A}$ stabilizes $\ket{\psi}$ if $\mat{A}\ket{\psi}=\ket{\psi}$.}~\cite{gottesman1996class,calderbank1997quantum,gottesman1998heisenberg}. These Pauli operators are known as stabilizers and each stabilizer is a signed $n$-fold tensor product of $\set{\mat{I},\mat{X},\mat{Y},\mat{Z}}$.  The group operation is defined to be multiplication because if $\mat{A}$ and $\mat{B}$ are both stabilizers of a state, then $\mat{A}\mat{B}$ is also its stabilizer. Since this group of $2^n$ stabilizers can be represented by a subset of size $n$ that generates it, each stabilizer state can also be identified by this generating set. In other words, each stabilizer state can be represented by $n$ signed Pauli operators $\set{\mat{S}_1,\dots,\mat{S}_n}$ that are independent and commuting. Here, the independency means that $\mat{S}_i\mat{S}_j\neq\pm\mat{S}_k$ for any different $i,j,k\in[n]$.

Furthermore, let $\mat{Q}$ be a Clifford circuit applied on a state $\ket{\psi}$ stabilized by $\set{\mat{S}_1,\dots,\mat{S}_n}$. Then we have
\[
\mat{Q}\ket{\psi}=
(\mat{Q}\mat{S}_j\mat{Q}^{-1})\mat{Q}\ket{\psi},\quad\forall j=1,\dots,n.
\]
Since $\mat{Q}\mat{S}_j\mat{Q}^{-1}$ is still a signed Pauli operator, $\mat{Q}$ can be viewed as a mapping from one group of $2^n$ stabilizers to another group of $2^n$ stabilizers. Such a mapping could be fully described by specifying the images of $2n$ independent Pauli operators.

The widely-used \textit{tableau method}~\cite{aaronson2004improved} utilizes this fact to represent Clifford circuits efficiently.
The tableau method describes $\mat{Q}$ by its images of $\mat{X}^{\mat{e}_j}$ and $\mat{Z}^{\mat{e}_j}$ for $j=1,\dots,n$. These images are encoded in the following matrix,
\[
\begin{bmatrix}
\mat{\bar{X}} & \mat{\bar{Z}} & \mat{s}\\
\mat{\tilde{X}} & \mat{\tilde{Z}} & \mat{r}
\end{bmatrix},
\]
where $\mat{\bar{X}},\mat{\bar{Z}},\mat{\tilde{X}},\mat{\tilde{Z}}\in\binaryset^{n\times n},\mat{r},\mat{s}\in\binaryset^n$, and
\begin{align*}
\mat{Q}\mat{X}^{\mat{e}_j}\mat{Q}^{-1}
&=(-1)^{s_j}\bigotimes_{k=1}^n \mat{P}^{(\bar{z}_{jk},\bar{x}_{jk})},\\
\mat{Q}\mat{Z}^{\mat{e}_j}\mat{Q}^{-1}
&=(-1)^{r_j}\bigotimes_{k=1}^n \mat{P}^{(\tilde{z}_{jk},\tilde{x}_{jk})},
\end{align*}
for $\mat{P}^{(0,0)}=\mat{I},
\mat{P}^{(1,0)}=\mat{Z},
\mat{P}^{(0,1)}=\mat{X},
\mat{P}^{(1,1)}=\mat{Y}$.

A stabilizer state $\ket\psi$ can also be represented with a tableau by representing some $\mat{Q}$ such that $\ket\psi=\mat{Q}\ket{0^n}$ (note that the last $n$ rows of the tableau of $\mat{Q}$ generate the stabilizers of $\ket\psi$). It has been shown that in this representation, circuit composition, amplitude evaluation and measurements can all be reduced to matrix multiplications and cost $O(n^\omega)$ time~\cite{aaronson2004improved,gosset2024fast}. Moreover, it is worth mentioning that applying basic Clifford gates such as $\{\mat{H}, \rzp, \mat{U}_{CZ}\}$ to a stabilizer state in tableau representation only costs $O(n)$ time~\cite{aaronson2004improved}. We also list these gate application algorithms in \cref{alg:tableau-operations}.

\subsubsection{Other Clifford circuit simulation methods}
\label{sec:other-sim-methods}
Besides tableaux, graph states and quadratic form expansions are also used to simulate Clifford circuits in classical polynomial time. We discuss two such methods below.

The graph-state-based method~\cite{anders2006fast} represents a stabilizer state as a graph state with local Clifford operators applied on each vertex. Such representation always exists~\cite{van2004graphical}\footnote{We also provide an explicit algorithm that transforms a stabilizer state into such a form in \cref{app:stab-to-graph}.}. Then gate application and measurements can both be implemented by changing structure of the graph and updating the local operators due to \cref{thm:vertex-complementation}. Let $n$ be number of qubits and $m$ be number of gates. Then the time cost of finding the representation of the final state is $O(ndm)$, where $d\le n$ is the maximal degree of the graph during simulation, and each one-qubit measurement costs $O(d)$ time.

The quadratic-form-expansion methods~\cite{de2022fast} utilize the fact that an $n$-qubit stabilizer state can always be written in the following form~\cite{van2010classical},
\begin{equation}\label{eq:qfe-amplitudes}
    \ket{\psi}=2^{-r/2}\sum_{\mat{x}\in\binaryset^r}i^{q(\mat{x})}\ket{(\mat{A}\mat{x}+\mat{b})\bmod2},\qquad q(\mat{x})=\sum_{i=1}^r d_ix_i+2\sum_{1\le i<j\le r}q_{ij}x_ix_j,
\end{equation}
where $r\le n$, $\mat{A}\in\binaryset^{n\times r}$, $\mat{b}\in\binaryset^{n}$, and the phase is a degree-two polynomial over the integers, specified by $d\in\Z_4^r$ and $q_{ij}\in\binaryset$; equivalently $q(\mat{x})=\mat{x}^T\widetilde{\mat{Q}}\mat{x}$ for a symmetric $\widetilde{\mat{Q}}\in\J_r$ with entries in $\Z_4$, in the same convention as the phased graph states $\ket{A}$ of \cref{thm:amp-formula} (the cross terms $2q_{ij}x_ix_j$ are what make the phase pattern non-separable: the graph state on two connected vertices has phases $(1,1,1,-1)$, which no linear $\F_2$ exponent can produce, since $x^T\mat{Q}x=\sum_iQ_{ii}x_i$ over $\F_2$). Applying Clifford gates to this representation can be described by matrix operations~\cite{de2022fast}. Then simulating an $n$-qubit circuit with $m$ gates would cost $O(mn^2)$ time, after which sampling or amplitude computing only takes $O(n^2)$ time.

\subsection{Tensor networks}
Tensor networks are graphical representations of tensor contractions. The tensor networks used in this paper can be defined as $(G, O)$. $G$ is a multigraph such that $V(G)$ is a set of contracted tensors and $E(G)$ is the multiset of contracted indices. $O$ is a multiset representing uncontracted indices, which is also known as open edges. For example, the following contraction,
\begin{equation}\label{eq:tn-example}
    \mat{T}_{ijk}=\sum_{\alpha\beta\gamma}\mat{A}_{ij\alpha\gamma}\mat{B}_{\alpha\beta}\mat{C}_{\beta\gamma k}
\end{equation}
is represented by the graph $G=(\set{\mat{A},\mat{B},\mat{C}},\set{\alpha,\beta,\gamma})$ and $O=\set{i,j,k}$, where
\begin{align*}
    i=j=\set{\mat{A}},\quad
    k=\set{\mat{C}},\quad
    \alpha=\set{\mat{A},\mat{B}},\quad
    \beta=\set{\mat{B},\mat{C}},\quad
    \gamma=\set{\mat{A},\mat{C}}.
\end{align*}
The diagrams of such graphs are known as tensor diagrams. For example, \cref{fig:tn-example} shows the tensor diagram of \cref{eq:tn-example}. When $O=\emptyset$, we call the tensor network as \textit{closed}.
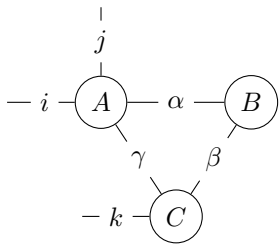
\begin{figure}[h]
\centering
\begin{tikzpicture}
\node[circle, draw, fill=white] (A) at (-1,0) {$A$};
\node[circle, draw, fill=white] (B) at (1,0) {$B$};
\node[circle, draw, fill=white] (C) at (0,-1.5) {$C$};
\node (alpha) at (0, 0) {$\alpha$};
\node (beta) at (0.5,-0.75) {$\beta$};
\node (gamma) at (-0.5,-0.75) {$\gamma$};
\node (i) at (-1.75,0) {$i$};
\node (j) at (-1, 0.8) {$j$};
\node (k) at (-0.8,-1.5) {$k$};
\draw (A) -- (i) (i) -- (-2.25,0);
\draw (A) -- (j) (j) -- (-1,1.25);
\draw (C) -- (k) (k) -- (-1.25,-1.5);
\draw (A) -- (alpha) (alpha) -- (B);
\draw (B) -- (beta) (beta) -- (C);
\draw (C) -- (gamma) (gamma) -- (A);
\end{tikzpicture}
\caption{A tensor diagram representing \cref{eq:tn-example}}
\label{fig:tn-example}
\end{figure}

Tensor networks are a powerful tool in studying multibody quantum systems where high-dimensional tensors are approximated by compact tensor network representation. Also, they play an important role in quantum circuit simulation because quantum circuits can be viewed as tensor networks~\cite{markov2008simulating}. In this paper, we will focus on the latter application.

\subsection{ZX-diagrams}\label{subsec:zx}
ZX-diagrams are special tensor networks that only contain two types of tensors known as \textit{Z-spiders} and \textit{X-spiders}, respectively~\cite{coecke2008interacting},
\begin{align*}
    \mat{Z}(\alpha)_{i_1i_2\dots i_n}&=\begin{cases}
    1 & i_1i_2\dots i_n=00\dots0 \\
    \exp(i\alpha) & i_1i_2\dots i_n=11\dots1 \\
    0 & \otherwise,
    \end{cases}\\
    \mat{X}(\alpha)&=\prod_{j=1}^n \actson{\mat{H}}{j}\mat{Z}(\alpha),
\end{align*}
where indices $i_j\in\binaryset$ for $j=1\dots n$ and $\actson{\mat{H}}{j}$ is the Hadamard matrix applied on the $j$'th index. In this paper, we will use white and gray nodes to represent Z- and X-spiders respectively in the diagrams, and when $\alpha=0$, we omit the 0 inside the node:
\begin{align*}
\mat{Z}(\alpha)=\vcenter{\hbox{\begin{tikzpicture}
    \node[znode] (v) at (0,0) {$\alpha$};
    \node (t) at (0.6,0.1) {$\vdots$};
    \coordinate (v0) at (0.75,0.5);
    \coordinate (v1) at (0.75,-0.5);
    \draw (v) to[out=60,in=180] (v0) (v) to[out=-60,in=180] (v1);
\end{tikzpicture}}}\quad,\qquad
\mat{X}(\alpha)=\vcenter{\hbox{\begin{tikzpicture}
    \node[xnode] (v) at (0,0) {$\alpha$};
    \node (t) at (0.6,0.1) {$\vdots$};
    \coordinate (v0) at (0.75,0.5);
    \coordinate (v1) at (0.75,-0.5);
    \draw (v) to[out=60,in=180] (v0) (v) to[out=-60,in=180] (v1);
\end{tikzpicture}}}\quad,\qquad
\vcenter{\hbox{\begin{tikzpicture}
    \node[znode] (v) at (0,0) {$0$};
    \node (t) at (0.6,0.1) {$\vdots$};
    \coordinate (v0) at (0.75,0.5);
    \coordinate (v1) at (0.75,-0.5);
    \draw (v) to[out=60,in=180] (v0) (v) to[out=-60,in=180] (v1);
\end{tikzpicture}}}
=\vcenter{\hbox{\begin{tikzpicture}
    \node[znode] (v) at (0,0) {\phantom{$0$}};
    \node (t) at (0.6,0.1) {$\vdots$};
    \coordinate (v0) at (0.75,0.5);
    \coordinate (v1) at (0.75,-0.5);
    \draw (v) to[out=60,in=180] (v0) (v) to[out=-60,in=180] (v1);
\end{tikzpicture}}}\quad.
\end{align*}
ZX-diagrams can be used to represent quantum states and quantum gates. For example, we have the following up to a multiplicative constant factor.
\begin{equation}\label{eq:zx-examples}
\begin{aligned}
&
\ket{0}=\vcenter{\hbox{\begin{tikzpicture}
    \node[xnode] (v) at (0,0) {$0$};
    \coordinate (v0) at (0.5,0);
    \draw (v) -- (v0);
\end{tikzpicture}}}\quad,\quad
\ket{1}=\vcenter{\hbox{\begin{tikzpicture}
    \node[xnode] (v) at (0,0) {$\pi$};
    \coordinate (v0) at (0.5,0);
    \draw (v) -- (v0);
\end{tikzpicture}}}\quad,\quad
\ket{+}=\vcenter{\hbox{\begin{tikzpicture}
    \node[znode] (v) at (0,0) {$0$};
    \coordinate (v0) at (0.5,0);
    \draw (v) -- (v0);
\end{tikzpicture}}}\quad,\quad
\ket{-}=\vcenter{\hbox{\begin{tikzpicture}
    \node[znode] (v) at (0,0) {$\pi$};
    \coordinate (v0) at (0.5,0);
    \draw (v) -- (v0);
\end{tikzpicture}}}\quad,\\&
\mat{Z}=\vcenter{\hbox{\begin{tikzpicture}
    \node[znode] (v) at (0,0) {$\pi$};
    \coordinate (v0) at (-0.5,0);
    \coordinate (v1) at (0.5,0);
    \draw (v) -- (v0) (v) -- (v1);
\end{tikzpicture}}}\quad,\quad
\mat{X}=\vcenter{\hbox{\begin{tikzpicture}
    \node[xnode] (v) at (0,0) {$\pi$};
    \coordinate (v0) at (-0.5,0);
    \coordinate (v1) at (0.5,0);
    \draw (v) -- (v0) (v) -- (v1);
\end{tikzpicture}}}\quad,\quad
\mat{S}=\vcenter{\hbox{\begin{tikzpicture}
    \node[znode] (v) at (0,0) {$\frac{\pi}{2}$};
    \coordinate (v0) at (-0.6,0);
    \coordinate (v1) at (0.6,0);
    \draw (v) -- (v0) (v) -- (v1);
\end{tikzpicture}}}\quad,\quad
\mat{T}=\vcenter{\hbox{\begin{tikzpicture}
    \node[znode] (v) at (0,0) {$\frac{\pi}{4}$};
    \coordinate (v0) at (-0.6,0);
    \coordinate (v1) at (0.6,0);
    \draw (v) -- (v0) (v) -- (v1);
\end{tikzpicture}}}\quad,\\&
\mat{U}_{CZ}=\vcenter{\hbox{\begin{tikzpicture}
    \node[znode] (u) at (0,0.5) {};
    \node[hnode] (h) at (0,0) {};
    \node[znode] (v) at (0,-0.5) {};
    \coordinate (u0) at (-0.5,0.5);
    \coordinate (u1) at (0.5,0.5);
    \coordinate (v0) at (-0.5,-0.5);
    \coordinate (v1) at (0.5,-0.5);
    \draw (u) -- (h) (h) -- (v);
    \draw (u) -- (u0) (u) -- (u1);
    \draw (v) -- (v0) (v) -- (v1);
\end{tikzpicture}}}\quad\where\quad
\vcenter{\hbox{\begin{tikzpicture}
    \node[hnode] (v) at (0,0) {};
    \coordinate (v0) at (-0.5,0);
    \coordinate (v1) at (0.5,0);
    \draw (v) -- (v0) (v) -- (v1);
\end{tikzpicture}}}=\mat{H}=
\vcenter{\hbox{\begin{tikzpicture}
    \node[znode] (v1) at (0,0) {$\frac{\pi}{2}$};
    \node[xnode] (v2) at (1,0) {$\frac{\pi}{2}$};
    \node[znode] (v3) at (2,0) {$\frac{\pi}{2}$};
    \coordinate (v0) at (-0.6,0);
    \coordinate (v4) at (2.6,0);
    \draw (v0) -- (v1) (v1) -- (v2) (v2) -- (v3) (v3) -- (v4);
\end{tikzpicture}}}\quad.
\end{aligned}
\end{equation}
Here, we use a square node to represent $H$ gate to simplify the diagram even though it is representable by only Z- and X-spiders. Also, we note that the constant factor is often omitted in ZX-diagrams, because it can be easily calculated along with the diagram itself.
In this paper, we also omit these constant factors for simplicity as they would not affect our results.

It is worth noting that ZX-diagrams, along with a set of rules that transform them, are complete to describe linear maps between qubits~\cite{hadzihasanovic2018two}. One of the rules we use in this paper is known as \textit{merge} rule, which says that two connected spiders of the same type can be merged into one with their phases being summed up. A diagram that depicts merge rule is shown below.
\begin{equation}\label{eq:merge-rule}
\vcenter{\hbox{\begin{tikzpicture}
    \node[znode] (u) at (0,0) {$\alpha$};
    \node[znode] (v) at (1,0) {$\beta$};
    \node (u2) at (-0.6,0.1) {$\vdots$};
    \node (v2) at (1.6,0.1) {$\vdots$};
    \node (mid) at (0.5,0.1) {$\vdots$};
    \coordinate (u0) at (-0.75,0.5);
    \coordinate (u1) at (-0.75,-0.5);
    \coordinate (v0) at (1.75,0.5);
    \coordinate (v1) at (1.75,-0.5);
    \draw (u) to[out=45,in=135] (v);
    \draw (u) to[out=-45,in=-135] (v);
    \draw (u) to[out=120,in=0] (u0) (u) to[out=-120,in=0] (u1);
    \draw (v) to[out=60,in=180] (v0) (v) to[out=-60,in=180] (v1);
\end{tikzpicture}}}
=\vcenter{\hbox{\begin{tikzpicture}
    \node[znode] (v) at (0,0) {$\scriptstyle\alpha+\beta$};
    \node (u2) at (-0.6,0.1) {$\vdots$};
    \node (v2) at (0.6,0.1) {$\vdots$};
    \coordinate (u0) at (-0.75,0.5);
    \coordinate (u1) at (-0.75,-0.5);
    \coordinate (v0) at (0.75,0.5);
    \coordinate (v1) at (0.75,-0.5);
    \draw (v) to[out=120,in=0] (u0) (v) to[out=-120,in=0] (u1);
    \draw (v) to[out=60,in=180] (v0) (v) to[out=-60,in=180] (v1);
\end{tikzpicture}}}
\end{equation}
For a complete list of rules, we refer our readers to~\cite{hadzihasanovic2018two}.

\subsection{Graph states as ZX-diagrams}
\label{subsubsect:zx-graph-state}
ZX-diagrams can represent graph states in a convenient way. Let $G=(V,E)$ be a graph, then $\ket{G}$ is represented by a ZX-diagram with $\abs{V}$ Z-spiders of phase 0. Each Z-spider corresponds to a vertex in $G$, which has an open edge and is connected to its neighbors via a Hadamard node. For example, let $K_3$ be the complete graph with 3 vertices. The ZX-diagram of $\ket{K_3}$ is shown below.
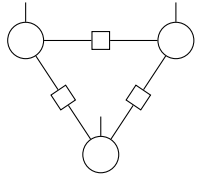
\begin{align*}
\vcenter{\hbox{\begin{tikzpicture}
    \node[znode] (u) at (1,0) {};
    \node[znode] (v) at (-1,0) {};
    \node[znode] (w) at (0,-1.5) {};
    \node[hnode] (e1) at (0,0) {};
    \node[hnode,rotate=56.3] (e2) at (0.5,-0.75) {};
    \node[hnode,rotate=123.7] (e3) at (-0.5,-0.75) {};
    \coordinate (u0) at (1,0.5);
    \coordinate (v0) at (-1,0.5);
    \coordinate (w0) at (0,-1);
    \draw (u) -- (e1) (e1) -- (v);
    \draw (u) -- (e2) (e2) -- (w);
    \draw (v) -- (e3) (e3) -- (w);
    \draw (u) -- (u0) (v) -- (v0) (w) -- (w0);
\end{tikzpicture}}}
\end{align*}
The correctness of this representation could be easily seen by rewriting~\cref{def:graph-state} using ZX-diagrams.

\subsection{Graph-like ZX-diagrams}
Due to the natural relation between graph states and ZX-diagrams, \textit{graph-like ZX-diagrams} have been proposed as a canonical form to help analyze and simplify them~\cite{duncan2020graph}.
\begin{definition}[Graph-like ZX-diagram]\label{def:graph-like-zx}
A ZX-diagram is called graph-like if
\begin{enumerate}
    \item it only contains Z-spiders and Hadamard gates;
    \item Z-spiders are only connected via Hadamard gates and Hadamard gates are only used to connect Z-spiders;
    \item Z-spiders and all Hadamard gates connecting them form a simple graph;
    \item every open edge is incident to a Z-spider and every Z-spider is incident to at most one open edge.
\end{enumerate}
\end{definition}
By definition, each graph-like ZX-diagram admits an underlying graph formed by Z-spiders as the \textit{vertices} and Hadamard gates as the \textit{edges}. Note that this graph is
different from the graph by viewing the diagram as a tensor network (where Hadamard gates are treated as vertices). However, when considering the treewidth of a graph-like ZX-diagram, both definitions of the underlying graph lead to equivalent tree decompositions and give the same definition of treewidth. Therefore, we will not distinguish them when the context is clear.

Any ZX-diagram is equivalent to a graph-like ZX-diagram~\cite{duncan2020graph}. Moreover, it is implicitly shown by the proof of \cite[ Lemma 3.2]{duncan2020graph} that the complexity of this reduction is linear and the treewidth of the diagram is not increased. These are summarized in the following lemma as a modified version of \cite[Lemma 3.2]{duncan2020graph}.
\begin{lemma}
\label{lem:zx-to-graph}
Let $\mathcal{D}$ be any ZX-diagram with $n$ spiders, $m$ contracted indices and $k$ uncontracted indices. Then $\mathcal{D}$ can be converted to an equivalent graph-like ZX-diagram $\mathcal{D}'$ with at most $(n+2k)$ vertices and at most $(m+2k)$ edges in $O(m+n+k)$ time.
Moreover, if a tree decomposition with width $\tau$ of $\mathcal{D}$ is given, one can obtain a tree decomposition of $\mathcal{D}'$ with width at most $\tau$ in the same time complexity.
\end{lemma}
\begin{proof}
We restate the algorithm given by the proof of~\cite[Lemma 3.2]{duncan2020graph} to analyze its time and space complexity. The algorithm takes the following steps to convert $\mathcal{D}$ to $\mathcal{D}'$.
\begin{enumerate}
\item Replace all X-spiders with Z-spiders and Hadamards by definition.
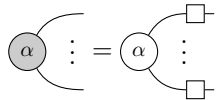
\begin{align*}
\vcenter{\hbox{\begin{tikzpicture}
    \node[xnode] (v) at (0,0) {$\alpha$};
    \node (t) at (0.6,0.1) {$\vdots$};
    \coordinate (v0) at (0.75,0.5);
    \coordinate (v1) at (0.75,-0.5);
    \draw (v) to[out=60,in=180] (v0) (v) to[out=-60,in=180] (v1);
\end{tikzpicture}}}=
\vcenter{\hbox{\begin{tikzpicture}
    \node[znode] (v) at (0,0) {$\alpha$};
    \node (t) at (0.6,0.1) {$\vdots$};
    \node[hnode] (h0) at (0.75,0.5) {};
    \node[hnode] (h1) at (0.75,-0.5) {};
    \coordinate (v0) at (1,0.5);
    \coordinate (v1) at (1,-0.5);
    \draw (v) to[out=60,in=180] (h0) (v) to[out=-60,in=180] (h1);
    \draw (h0) -- (v0) (h1) -- (v1);
\end{tikzpicture}}}
\end{align*}

\item Remove adjacent Hadamards since they cancel each other.
\begin{align*}
\vcenter{\hbox{\begin{tikzpicture}
    \node[hnode] (u) at (-0.25,0) {};
    \node[hnode] (v) at (0.25,0) {};
    \coordinate (u0) at (-0.75,0);
    \coordinate (v0) at (0.75,0);
    \draw (v) -- (v0) (u) -- (u0) (u) -- (v);
\end{tikzpicture}}}=
\vcenter{\hbox{\begin{tikzpicture}
    \coordinate (v0) at (-0.5,0);
    \coordinate (v1) at (0.5,0);
    \draw (v0) -- (v1);
\end{tikzpicture}}}
\end{align*}

\item Remove edges without Hadamards by merging adjacent Z-spiders.
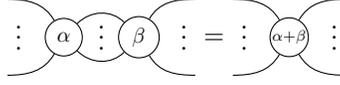
\begin{align*}
\vcenter{\hbox{\begin{tikzpicture}
    \node[znode] (u) at (0,0) {$\alpha$};
    \node[znode] (v) at (1,0) {$\beta$};
    \node (u2) at (-0.6,0.1) {$\vdots$};
    \node (v2) at (1.6,0.1) {$\vdots$};
    \node (mid) at (0.5,0.1) {$\vdots$};
    \coordinate (u0) at (-0.75,0.5);
    \coordinate (u1) at (-0.75,-0.5);
    \coordinate (v0) at (1.75,0.5);
    \coordinate (v1) at (1.75,-0.5);
    \draw (u) to[out=45,in=135] (v);
    \draw (u) to[out=-45,in=-135] (v);
    \draw (u) to[out=120,in=0] (u0) (u) to[out=-120,in=0] (u1);
    \draw (v) to[out=60,in=180] (v0) (v) to[out=-60,in=180] (v1);
\end{tikzpicture}}}
=\vcenter{\hbox{\begin{tikzpicture}
    \node[znode] (v) at (0,0) {$\scriptstyle\alpha+\beta$};
    \node (u2) at (-0.6,0.1) {$\vdots$};
    \node (v2) at (0.6,0.1) {$\vdots$};
    \coordinate (u0) at (-0.75,0.5);
    \coordinate (u1) at (-0.75,-0.5);
    \coordinate (v0) at (0.75,0.5);
    \coordinate (v1) at (0.75,-0.5);
    \draw (v) to[out=120,in=0] (u0) (v) to[out=-120,in=0] (u1);
    \draw (v) to[out=60,in=180] (v0) (v) to[out=-60,in=180] (v1);
\end{tikzpicture}}}
\end{align*}

\item Remove parallel Hadamard edges and self loops by using the following rules:
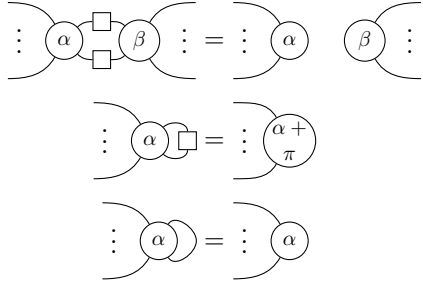
\begin{align*}
\vcenter{\hbox{\begin{tikzpicture}
    \node[znode] (u) at (0,0) {$\alpha$};
    \node[znode] (v) at (1,0) {$\beta$};
    \node[hnode] (h1) at (0.5,0.25) {};
    \node[hnode] (h2) at (0.5,-0.25) {};
    \node (u2) at (-0.6,0.1) {$\vdots$};
    \node (v2) at (1.6,0.1) {$\vdots$};
    \coordinate (u0) at (-0.75,0.5);
    \coordinate (u1) at (-0.75,-0.5);
    \coordinate (v0) at (1.75,0.5);
    \coordinate (v1) at (1.75,-0.5);
    \draw (u) to[out=45,in=180] (h1) (h1) to[out=0,in=135] (v);
    \draw (u) to[out=-45,in=180] (h2) (h2) to[out=0,in=-135] (v);
    \draw (u) to[out=120,in=0] (u0) (u) to[out=-120,in=0] (u1);
    \draw (v) to[out=60,in=180] (v0) (v) to[out=-60,in=180] (v1);
\end{tikzpicture}}}
&=\vcenter{\hbox{\begin{tikzpicture}
    \node[znode] (u) at (0,0) {$\alpha$};
    \node[znode] (v) at (1,0) {$\beta$};
    \node (u2) at (-0.6,0.1) {$\vdots$};
    \node (v2) at (1.6,0.1) {$\vdots$};
    \coordinate (u0) at (-0.75,0.5);
    \coordinate (u1) at (-0.75,-0.5);
    \coordinate (v0) at (1.75,0.5);
    \coordinate (v1) at (1.75,-0.5);
    \draw (u) to[out=120,in=0] (u0) (u) to[out=-120,in=0] (u1);
    \draw (v) to[out=60,in=180] (v0) (v) to[out=-60,in=180] (v1);
\end{tikzpicture}}}\\[1ex]
\vcenter{\hbox{\begin{tikzpicture}
    \node[znode] (u) at (0,0) {$\alpha$};
    \node[hnode] (h) at (0.5,0) {};
    \node (u2) at (-0.6,0.1) {$\vdots$};
    \coordinate (u0) at (-0.75,0.5);
    \coordinate (u1) at (-0.75,-0.5);
    \draw (u) to[out=45,in=90] (h) (h) to[out=-90,in=-45] (u);
    \draw (u) to[out=120,in=0] (u0) (u) to[out=-120,in=0] (u1);
\end{tikzpicture}}}
&=\vcenter{\hbox{\begin{tikzpicture}
    \node[znode] (u) at (0,0) {$\alpha+\pi$};
    \node (u2) at (-0.6,0.1) {$\vdots$};
    \coordinate (u0) at (-0.75,0.5);
    \coordinate (u1) at (-0.75,-0.5);
    \draw (u) to[out=120,in=0] (u0) (u) to[out=-120,in=0] (u1);
\end{tikzpicture}}}\\[1ex]
\vcenter{\hbox{\begin{tikzpicture}
    \node[znode] (u) at (0,0) {$\alpha$};
    \node (u2) at (-0.6,0.1) {$\vdots$};
    \coordinate (v) at (0.5,0);
    \coordinate (u1) at (-0.75,-0.5);
    \draw (u) to[out=60,in=90] (v) (v) to[out=-90,in=-60] (u);
    \draw (u) to[out=120,in=0] (u0) (u) to[out=-120,in=0] (u1);
\end{tikzpicture}}}
&=\vcenter{\hbox{\begin{tikzpicture}
    \node[znode] (u) at (0,0) {$\alpha$};
    \node (u2) at (-0.6,0.1) {$\vdots$};
    \coordinate (u0) at (-0.75,0.5);
    \coordinate (u1) at (-0.75,-0.5);
    \draw (u) to[out=120,in=0] (u0) (u) to[out=-120,in=0] (u1);
\end{tikzpicture}}}
\end{align*}

\item For open edges that are not connected to Z-spiders, add dummy spiders to connect them.
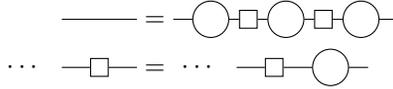
\begin{align*}
\vcenter{\hbox{\begin{tikzpicture}
    \coordinate (v0) at (-0.5,0);
    \coordinate (v1) at (0.5,0);
    \draw (v0) -- (v1);
\end{tikzpicture}}}&=
\vcenter{\hbox{\begin{tikzpicture}
    \node[znode] (u) at (0,0) {};
    \node[znode] (v) at (1,0) {};
    \node[znode] (w) at (2,0) {};
    \node[hnode] (h1) at (0.5,0) {};
    \node[hnode] (h2) at (1.5,0) {};
    \coordinate (u0) at (-0.5,0);
    \coordinate (w0) at (2.5,0);
    \draw (u) -- (h1) (h1) -- (v) (v) -- (h2) (h2) -- (w);
    \draw (u) -- (u0) (w) -- (w0);
\end{tikzpicture}}}\\
\vcenter{\hbox{\begin{tikzpicture}
    \node[hnode] (h1) at (0,0) {};
    \node (text) at (-1,0) {$\cdots$};
    \coordinate (v0) at (-0.5,0);
    \coordinate (v1) at (0.5,0);
    \draw (v0) -- (h1) (h1) -- (v1);
\end{tikzpicture}}}&=
\vcenter{\hbox{\begin{tikzpicture}
    \node[hnode] (h1) at (0,0) {};
    \node[znode] (u) at (0.75,0) {};
    \node (text) at (-1,0) {$\cdots$};
    \coordinate (v0) at (-0.5,0);
    \coordinate (v1) at (1.25,0);
    \draw (v0) -- (h1) (h1) -- (u) (u) -- (v1);
\end{tikzpicture}}}
\end{align*}

\item For multiple open edges connected to one Z-spider, adding dummy spiders to split them.
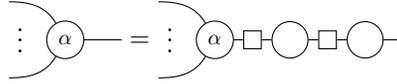
\begin{align*}
\vcenter{\hbox{\begin{tikzpicture}
    \node[znode] (u) at (0,0) {$\alpha$};
    \node (u2) at (-0.6,0.1) {$\vdots$};
    \coordinate (v) at (0.75,0);
    \coordinate (u0) at (-0.75,0.5);
    \coordinate (u1) at (-0.75,-0.5);
    \draw (u) -- (v);
    \draw (u) to[out=120,in=0] (u0) (u) to[out=-120,in=0] (u1);
\end{tikzpicture}}}
&=\vcenter{\hbox{\begin{tikzpicture}
    \node[znode] (u) at (0,0) {$\alpha$};
    \node[hnode] (v1) at (0.5,0) {};
    \node[znode] (v2) at (1,0) {};
    \node[hnode] (v3) at (1.5,0) {};
    \node[znode] (v4) at (2,0) {};
    \node (u2) at (-0.6,0.1) {$\vdots$};
    \coordinate (v5) at (2.5,0);
    \coordinate (u0) at (-0.75,0.5);
    \coordinate (u1) at (-0.75,-0.5);
    \draw (u) -- (v1) (v1) -- (v2) (v2) -- (v3) (v3) -- (v4) (v4) -- (v5);
    \draw (u) to[out=120,in=0] (u0) (u) to[out=-120,in=0] (u1);
\end{tikzpicture}}}
\end{align*}
\end{enumerate}
Note that each step can be completed by traversing the graph once, so the running time is $O(n+m+k)$.
Also, note that steps 1 to 4 do not increase the number of vertices/spiders or edges, and steps 5 and 6 introduce at most $2$ vertices and $2$ edges per uncontracted index, counting each open end of an uncontracted edge as one index, so $\mathcal{D}'$ has at most $(n+2k)$ vertices and at most $(m+2k)$ edges.
Moreover, consider an arbitrary tree decomposition $\mathcal{T}$ of $\mathcal{D}$. All graph operations involved in steps 1 to 6 (including merging connected vertices, deleting vertices/edges, and adding chained dummy vertices) transform $\mathcal{T}$ in linear time, and none of these operations increases the width of $\mathcal{T}$.
\end{proof}

\section{Clifford circuit simulation as phased graph state simulation}
\label{sec:circ-as-pgs}
In this section, we define phased graph states, and present a reduction from Clifford circuit simulation to phased graph state simulation.
This allows us to transform the circuit using vertex and edge complementation (i.e. \cref{thm:vertex-complementation} and \cref{thm:edge-complementation}) and apply tools from linear algebra to solve the simulation problem in \cref{sec:pgs-sim}.

To begin with, we define \textit{phased graph states} as a generalization of graph states with arbitrary $\rzm$ gates applied on each vertex of the graph. Such states can be represented by symmetric matrices whose off-diagonals are over $\F_2$ and diagonals are over $\Z_4$. In particular, consider $\mat{A}\in\Z_4^{n\times n}$, we define $\Omega_{1,2}(\mat{A})=\omega_1(\mat{A})+2\omega_2(D(\mat{A}))$ and denote $\J_n=\set*{\mat{A}\in\Z_4^{n\times n}}{\Omega_{1,2}(\mat{A})=\mat{A},\,\mat{A}=\mat{A}^T}$.
Then phased graph states can be defined as follows.

\begin{definition}[Phased graph state]\label{def:phased-graph-state}
Let $\mat{A}\in\J_n$. Then the phased graph state $\ket{\mat{A}}$ is defined by
$
\ket{\mat{A}}=\rzm^{d(\mat{A})}\ket{G}
$,
where $G$ is the graph with adjacency matrix $O(\mat{A})$.
\end{definition}
Similar notations of phased graph states have been seen in the literature~\cite{zhao2016fast,spengler2013graph}. Also, the amplitudes of $\ket{\mat{A}}$ in computational basis could be obtained as a generalization of \cref{lem:graph-state-amplitudes}.
\begin{lemma}\label{lem:phased-graph-state-amplitudes}
Let $\mat{A}\in\J_n$ and $\mat{x}\in\F_2^n$. Then
\(
\braket{\mat{x}}{\mat{A}} = 2^{-n/2}(-i)^{\B x^T\B A\B x}.
\)
\end{lemma}
\begin{proof}
Note that $\rzm$ gates in $\ket{\mat{A}}$ only introduce a complex phase depending on $\mat{x}$ and $d(\mat{A})$. Then by \cref{def:phased-graph-state} and \cref{lem:graph-state-amplitudes},
\[
\bra{\mat{x}}\ket{\mat{A}}=\bra{\mat{x}}\rzm^{d(\mat{A})}\ket{O(\mat{A})}=(-i)^{\mat{x}^T d(\mat{A})}\bra{\mat{x}}\ket{O(\mat{A})}
=2^{-n/2}(-i)^{\B x^T\B A\B x}.
\]
\end{proof}

We are interested in solving the strong and weak simulation problem for phased graph states defined as follows.
\begin{definition}[Phased graph state simulation]\label{def:pgs-sim}
Let $\mat{A}\in\J_n$. The \textit{strong simulation} of $\ket{\mat{A}}$ takes any $\mat{x}\in\F_2^n$ as input and calculates
\[
\mel{\mat{x}}{\mat{H}^{\otimes n}}{\mat{A}}.
\]
The \textit{weak simulation} of $\ket{\mat{A}}$ takes $S\subseteq[n]$ and $\mat{y}\in\F_2^n$ as input, and samples $\mat{x}\in\mathcal{X}_{S,\mat{y}}$ according to the distribution
\[
\forall \mat{x} \in \mathcal{X}_{S,y},\;p(\mat{x})\propto\abs{\mel{\mat{x}}{\mat{H}^{\otimes n}}{\mat{A}}}^2.
\]
\end{definition} For $k\in\N$ the $k$-fold versions take $k$ inputs $x^{(1)},\dots,x^{(k)}$ and return the $k$ amplitudes, respectively $k$ independent samples; each sample is returned as its $\ell=n-\abs{S}$ free coordinates (the $S$-coordinates are $y_S$), and further coordinates are recovered when needed.
Note that we always consider the phased graph states being transformed by Hadamard gates because of two reasons: (1) this is a natural result of \cref{lem:circ-as-pgs} as presented below; (2) vertices that are not transformed by Hadamard gates can be efficiently eliminated from the graph as we will discuss in \cref{sec:gs-sim}.

Now we are ready to present the following lemma that reduces Clifford circuit simulation problem as defined in \cref{def:circ-sim} to phased graph state simulation as defined in \cref{def:pgs-sim}.

\begin{lemma}\label{lem:circ-as-pgs}
Let $\mat{U}$ be an $n$-qubit Clifford circuit represented by $m$ gates from the gate set $\{\mat{H},\rzp,\mat{U}_\text{CZ}\}$. Then one can compute a tuple $(\mat{A},S,\mat{y})$ in $O(m+n)$ time, where $\mat{A}\in\J_N$ represents a phased graph state for some $N=O(m+n)$, $S\subseteq[N]$ with $\abs{S}=N-n$, and $\mat{y}\in\F_2^S$,
such that
\begin{itemize}
\item
for every $\mat{x}\in\F_2^n$, the amplitude $\mel{\mat{x}}{\mat{U}}{0^n}$ equals $2^{(k-\nu)/2}\mel{\tilde{\mat{x}}}{\mat{H}^{\otimes N}}{\mat{A}}$, where $\tilde{\mat{x}}\in\F_2^N$ agrees with $\mat{x}$ on the $n$ output vertices (identified with $[n]$ in the natural way) and with $\mat{y}$ on $S$, $k=\rank(\omega_1(\mat{A}))$, and $\nu$ is the dimension of the affine support of $\mat{H}^{\otimes N}\ket{\mat{A}}$ on the fibre $\mathcal{X}_{S,\mat{y}}$;
\item
the weak simulation of $\mat{U}$ is equivalent to the weak simulation of $\ket{\mat{A}}$ with input $S$ and $\mat{y}$, i.e. to sampling $\tilde{\mat{x}}\in\mathcal{X}_{S,\mat{y}}$ with probability proportional to $\abs{\mel{\tilde{\mat{x}}}{\mat{H}^{\otimes N}}{\mat{A}}}^2$ and returning the bits on the output vertices;
\item
any tree decomposition of the graph-like diagram representing $\mat{U}$ with width $\tau$ can be converted to a tree decomposition of the graph corresponding to $\mat{A}$ with width at most $\max(\tau,1)$, since attaching the output legs of an edgeless diagram creates a matching.
\end{itemize}
\end{lemma}
\begin{proof}
This reduction could be obtained
by combining known techniques in ZX calculus.
Consider $\mat{U}\ket{0^n}$, whose ZX-diagram representation could be obtained by connecting small ZX-diagrams that represent each gate in $\mat{U}$ and $\ket{0}^{\otimes n}$ as given in \cref{eq:zx-examples}. This ZX-diagram contains $O(m+n)$ vertices, $O(m+n)$ edges, and $O(n)$ open edges. By \cref{lem:zx-to-graph}, this ZX-diagram could be converted to be graph-like in $O(m+n)$ time, and the resulting graph-like ZX-diagram contains $O(m+n)$ vertices and $O(m+n)$ edges.

Note that the graph-like ZX-diagram from the previous step has $n$ open edges, each of which corresponds to an output qubit of $\mat{U}\ket{0^n}$. For each output qubit $j$, we represent $x_j$ as a ZX-diagram that looks like
\begin{equation}\label{eq:output-qubit-zx}
\bra{x_j}=
\vcenter{\hbox{\begin{tikzpicture}
\node[xnode] (v) at (0,0) {$x_j\pi$};
\coordinate (v0) at (0.5,0);
\draw (v) -- (v0);
\end{tikzpicture}}}=
\vcenter{\hbox{\begin{tikzpicture}
\node[znode] (v) at (0,0) {$x_j\pi$};
\node[hnode] (h) at (0.6,0) {};
\coordinate (v0) at (1,0);
\draw (v) -- (h) -- (v0);
\end{tikzpicture}}}.
\end{equation}
Connecting \cref{eq:output-qubit-zx} to each open edge, we get a graph-like ZX-diagram representing $\mel{\mat{x}}{\mat{U}}{0^n}$ and the treewidth is not increased beyond $\max(\tau,1)$ (item 3 of \cref{lem:circ-as-pgs}: attaching one leaf per output leg gives treewidth at most $\max(\tau,1)$, since a leaf joins its neighbour's bags and the attached matching has treewidth $1$, and phase-splitting replaces each vertex by an edge-twin in the same graph). This ZX-diagram contains $N=O(m+n)$ Z-spiders and no open edges, and $n$ of the Z-spiders have phase $x_j\pi$. Moreover, the phases of Z-spiders are all multiples of $\pi/2$ because these phases are all sums of the phases in~\cref{eq:zx-examples} due to the merge rule~\cref{eq:merge-rule}.

Next, we will show that the ZX-diagram representing $\mel{\mat{x}}{\mat{U}}{0^n}$ from the previous step can be converted to the form of phased graph state simulation.

To begin with, we note that a graph-like ZX-diagram without open edges can be viewed as an inner product between a graph state and tensor product of some one-qubit states: for each Z-spider of phase $p_j$, one can apply the merge rule \cref{eq:merge-rule} reversely to split it into two connected Z-spiders, so that one of them has phase $p_j$ and the other inherits the edges associated with the original Z-spider. Then the former spiders represent a tensor product of one-qubit states, while the latter represents a graph state as presented in \cref{subsubsect:zx-graph-state}.
This procedure is illustrated in~\cref{fig:zx-to-graph}.
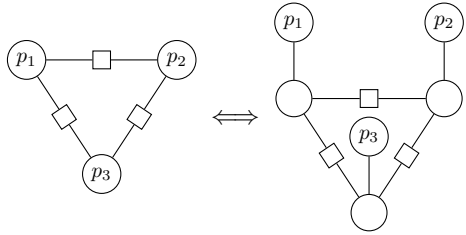
\begin{figure}[h]
\centering
$
 \vcenter{\hbox{\begin{tikzpicture}
    \node[znode] (u) at (1,0) {$p_2$};
    \node[znode] (v) at (-1,0) {$p_1$};
    \node[znode] (w) at (0,-1.5) {$p_3$};
    \node[hnode] (e1) at (0,0) {};
    \node[hnode,rotate=56.3] (e2) at (0.5,-0.75) {};
    \node[hnode,rotate=123.7] (e3) at (-0.5,-0.75) {};
    \draw (u) -- (e1) (e1) -- (v);
    \draw (u) -- (e2) (e2) -- (w);
    \draw (v) -- (e3) (e3) -- (w);
\end{tikzpicture}}}
\iff
\vcenter{\hbox{\begin{tikzpicture}
    \node[znode] (u) at (1,0) {};
    \node[znode] (v) at (-1,0) {};
    \node[znode] (w) at (0,-1.5) {};
    \node[hnode] (e1) at (0,0) {};
    \node[hnode,rotate=56.3] (e2) at (0.5,-0.75) {};
    \node[hnode,rotate=123.7] (e3) at (-0.5,-0.75) {};
    \node[znode] (u0) at (1,1) {$p_2$};
    \node[znode] (v0) at (-1,1) {$p_1$};
    \node[znode] (w0) at (0,-0.5) {$p_3$};
    \draw (u) -- (e1) (e1) -- (v);
    \draw (u) -- (e2) (e2) -- (w);
    \draw (v) -- (e3) (e3) -- (w);
    \draw (u) -- (u0) (v) -- (v0) (w) -- (w0);
\end{tikzpicture}}}
$
\caption{A graph-like ZX-diagram viewed as an inner product between graph states and a tensor product of one-qubit states. The Z-spiders with phases $p_j$ represent a tensor product of one-qubit states, while the remaining ZX-diagram represents a graph state.}
\label{fig:zx-to-graph}
\end{figure}

All rewritings in this subsection are equalities of tensors up to a scalar $\lambda$ that depends only on the number of spiders and edges of the diagram, and not on the amplitude index; ratios of amplitudes, and hence the probabilities and sampling distributions used below, are therefore unaffected. For the exact amplitudes, $\lambda=2^{(k-\nu)/2}$, where the affine-support dimension $\nu$ is computed by the additional sparse LU factorization in \cref{lem:sampling-complexity}, so the factor is computable within the same preprocessing bound. The identities used in this subsection omit only real powers of $\sqrt2$, since $\mat{H}$ is drawn as a square node (the explicit three-gate chain that also represents $\mat{H}$ differs from it by $e^{i\pi/4}$, and we do not use that form here). Hence $\lambda$ is positive real and determined exactly by the elimination. Applying this procedure to the ZX-diagram representing $\mel{\mat{x}}{\mat{U}}{0^n}$, we get a graph $G$ with $N=O(m+n)$ vertices and a tensor product of one-qubit states represented by Z-spiders,
\begin{equation}
\bigotimes_{j=1}^N \bra{\psi_j}=\bigotimes_{j=1}^N \frac{1}{\sqrt2}(\bra{0}+e^{ip_j}\bra{1}).
\end{equation}
Since all $p_j$ are multiples of $\pi/2$, let $\tilde{x}_j=\floor{p_j/\pi}\bmod2$ and $a_j=(2\tilde{x}_j-2p_j/\pi)\bmod4$, so that $2\tilde{x}_j-a_j\equiv 2p_j/\pi\pmod4$. Then we have
\begin{align}
\bra{\psi_j}
&=\bra{0}\mat{H}\rzp^{2\tilde{x}_j-a_j}
=\bra{0}\mat{X}^{\tilde{x}_j}\mat{H}\rzm^{a_j}
=\bra{\tilde{x}_j}\mat{H}\rzm^{a_j},
\end{align}
since $(-1)^{\tilde{x}_j}(-i)^{a_j}=i^{2\tilde{x}_j-a_j}=e^{ip_j}$; the exponent $a_j$ is taken modulo $4$, matching the $\Z_4$ diagonal of $\mat{A}$.
Let $\mat{A}\in\J_N$ such that $d(\mat{A})=\mat{a}$ and $O(\mat{A})$ is the adjacency matrix of $G$ and we have $N=O(m+n)$.
The graph-like ZX-diagram representing $\mel{\tilde{\mat{x}}}{\mat{U}}{0^n}$ is equal to
\begin{align}
\mel{\tilde{\mat{x}}}{\mat{H}^{\otimes N}\rzm^{\mat{a}}}{G}
=\mel{\tilde{\mat{x}}}{\mat{H}^{\otimes N}}{\mat{A}}.
\end{align}
Note that by construction, $n$ of the $\tilde{x}_j$'s are equal to the corresponding bits $x_j$ of the circuit output.
We define $S$ to be the set containing the remaining $N-n$ indices and define $\mat{y}\in\F_2^S$ such that $y_j=\tilde{x}_j$ for $j\in S$. Then we have $(\mat{A},S,\mat{y})$ as stated in the lemma, so that the strong and weak simulation of $\mat{U}$ are equivalent to the strong and weak simulation of $\ket{\mat{A}}$. Also, given any tree decomposition of $U$ with width $\tau$, a tree decomposition of $A$ of width at most $\max(\tau,1)$ could be obtained by item 3 of \cref{lem:circ-as-pgs}, as explained above.
\end{proof}

\section{Phased graph state and Gaussian elimination}\label{sec:pgs-ge}
In this section, we reveal the relation between phased graph states and Gaussian elimination. In \cref{sec:pgs-gj}, we define a generalized Gauss--Jordan elimination process for matrices in $\J_n$ and show that the result of this process can be obtained by LDL decomposition over $\F_2$. Also, we show that this process gives an amplitudes formula for Hadamard transformed phased graph states in \cref{sec:amp-formula}.

\subsection{A generalized Gauss--Jordan process}
\label{sec:pgs-gj}
In this section, we define a generalized Gauss--Jordan elimination process for matrices in $\J_n$. For simplicity, we consider symmetric matrices with unpivoted LDL decompositions so that Gaussian elimination can be applied sequentially on the row/columns.
Let $\W_n$ be the subset of matrices in $\J_n$ that admit an unpivoted LDL decomposition. More precisely,
\begin{align}
\W_n=\set*{\mat{A}\in\J_n}{\omega_1(\mat{A})\text{ admits an LDL decomposition }(\mat{P},\mat{L},\mat{D})\text{ s.t. } \mat{P}=\mat{I}}.
\end{align}
Note that any matrix in $\J_n$ can be transformed to a matrix in $\W_n$ by symmetric row/column permutation.
Indeed, move a maximal set of $k=\rank(\omega_1(A))$ linearly independent rows and the corresponding columns to the front. The leading block $\omega_1(A_{11})$ is nonsingular: a dependence among its columns extends to the remaining rows, since they are combinations of the first $k$ rows, and would contradict independence of the first $k$ columns. Its indices can then be ordered so that each pivot required by \cref{alg:gj} is available because the block that remains to be eliminated stays nonsingular; the block left after the $k$ pivots vanishes because $\rank(\omega_1(A))=k$. We use such an ordering throughout, so that the eliminated block is the leading block.

Consider a sequence of transformations (Gauss--Jordan process) based on eliminations of consecutive rows/columns or pairs of them for matrices in $\W_n$.
The process is defined in \cref{alg:gj}, with $a\leftarrow b$ denoting assignment. Its off-diagonal updates are vertex and edge complementations, while its first bits perform Gauss--Jordan elimination over $\F_2$. At each pivot we record its input second diagonal bits in $v$, and replace the pivot block by the negative of its first-bit block. The recorded bits are not subsequently changed; later complementations still update the diagonal of $B$.

\begin{algorithm}[t]
\caption{$[B,v]=$ Gauss--Jordan-$\W_n$($A,k$)}\label{alg:gj}
\begin{algorithmic}[1]
\Require $A\in\W_n$, with $A_{11}\in\W_k$ of rank $k$.
\Ensure $B=\Gamma_k(A)\in\J_n$ and $v\in\F_2^k$ satisfying \cref{lem:gj_cor,lem:gj-pgs}.
\State $i\leftarrow1$, $B\leftarrow A$
\While{$i\le k$}
\State $P\leftarrow\begin{bmatrix}0&I_{n-i+1}\\I_{i-1}&0\end{bmatrix}$, $\bar A\leftarrow PBP^T$
\If{$\omega_1(\bar a_{11})=1$}
\State $v_i\leftarrow\omega_2(\bar a_{11})$
\State $B\leftarrow P^T\Omega_{1,2}\left(\begin{bmatrix}-1&\bar a_{12}\\\bar a_{21}&\bar A_{22}-\bar a_{21}\bar a_{12}\end{bmatrix}\right)P$
\State $i\leftarrow i+1$
\Else
\State Partition $\bar A$ with leading block $\bar A_{11}\in\W_2$, and set $D_i\leftarrow\omega_1(\bar A_{11})=\begin{bmatrix}0&1\\1&0\end{bmatrix}$
\State $(v_i,v_{i+1})\leftarrow\omega_2(d(\bar A_{11}))^T$
\State $B\leftarrow P^T\Omega_{1,2}\left(\begin{bmatrix}-D_i&D_i\bar A_{12}\\\bar A_{21}D_i&\bar A_{22}-\bar A_{21}D_i\bar A_{12}\end{bmatrix}\right)P$
\State $i\leftarrow i+2$
\EndIf
\EndWhile
\end{algorithmic}
\end{algorithm}

Empty blocks and $k=0$ are allowed, with $\Gamma_0(A)=A$.

\begin{lemma}\label{lem:gj_cor}
Let $A\in\W_n$ and suppose its leading block $A_{11}\in\W_k$ is nonsingular over $\F_2$. Then \cref{alg:gj} terminates successfully and its output satisfies
\begin{align}
\omega_1(B)
&=\omega_1\left(\begin{bmatrix}A_{11}^{\#}&A_{11}^{\#}A_{12}\\A_{21}A_{11}^{\#}&A_{22}-A_{21}A_{11}^{\#}A_{12}\end{bmatrix}\right).
\label{eq:gjl1}
\end{align}
If $\omega_1(A_{11})=\omega_1(LDL^T)$ is the unpivoted LDL decomposition, then
\begin{align}
B_{11}&=\Omega_{1,2}(-L^{-T}DL^{-1}),\label{eq:gjl2}\\
B_{22}&=\Omega_{1,2}(A_{22}+A_{21}B_{11}A_{12}),\label{eq:gjl3}\\
v&=d\bigl(\omega_2(LDL^T)\oplus\omega_2(A_{11})\bigr).\label{eq:gjl4}
\end{align}
The products inside $\Omega_{1,2}$ and $\omega_2$ are evaluated over the integers. Since $L$ is unit lower triangular, its integral inverse may equivalently be replaced by its binary representative inside $\Omega_{1,2}$.
\end{lemma}
\begin{proof}
At the start of each iteration of \cref{alg:gj}, let $\ell=i-1$ be the number of rows and columns already eliminated.
The requirement that $A_{11}\in\W_k$ is full rank ensures that the first-bit pivot is either $[1]$ or $\begin{bmatrix}0&1\\1&0\end{bmatrix}$, and that the algorithm terminates with $i=k+1$.
Each of these two binary blocks equals its inverse over $\F_2$.
Further, the second diagonal bits do not affect the pivot choices or the first bits of subsequent updates.
Hence, modulo $2$, the process is equivalent to Gauss--Jordan matrix inversion over $\F_2$, proving \eqref{eq:gjl1}.

Further, to derive the form of $v$ in \eqref{eq:gjl4}, note that we may write
\[
LDL^T=\sum_{j=1}^t L_jD_jL_j^T,
\]
where $t$ is the number of diagonal blocks in $D$, and each $L_j$ consists of the corresponding one or two consecutive columns of $L$.
The $j$th forward update subtracts $L_jD_jL_j^T$ on the not-yet-eliminated rows and columns, with the diagonal retained modulo $4$.
Thus, for a coordinate $a$ in the $j$th pivot block, immediately before normalization,
\[
b_{aa}\equiv a_{aa}-\sum_{h<j}(L_hD_hL_h^T)_{aa}\pmod4.
\]
The diagonal block of $L$ at that pivot is the identity, so its own contribution is $d(D)_a$, and later blocks contribute zero.
Since $v_a=\omega_2(b_{aa})$ and $\omega_1(b_{aa})=d(D)_a$, it follows that
\[
2v\equiv d(A_{11}-LDL^T)\pmod4.
\]
The two matrices in this difference have the same first bits, so this is equivalent to \eqref{eq:gjl4}.

Next, consider an iteration with $\ell$ indices already eliminated and a pivot of size $s\in\{1,2\}$.
For this part of the proof, repartition $A$, $B$, $L$, and $D$ after the first $\ell$ and next $s$ indices, so that
\[
L=\begin{bmatrix}
L_{11}&0&0\\
L_{21}&L_{22}&0\\
L_{31}&L_{32}&L_{33}
\end{bmatrix},\qquad
B=\begin{bmatrix}
B_{11}&B_{12}&B_{13}\\
B_{21}&B_{22}&B_{23}\\
B_{31}&B_{32}&B_{33}
\end{bmatrix},
\]
where $L_{11}\in\F_2^{\ell\times\ell}$, $L_{22}=I_s$, and $\omega_1(B_{22})=D_{22}$.
By equivalence to Gauss--Jordan inversion over $\F_2$,
\begin{align*}
B_{21}
&=\omega_1(A_{21}A_{11}^{\#})\\
&=\omega_1(L_{21}D_{11}L_{11}^TL_{11}^{-T}D_{11}L_{11}^{-1})
 =\omega_1(L_{21}L_{11}^{-1}).
\end{align*}
Hence the update to the leading $\ell\times\ell$ block is
\begin{align*}
B_{11}&\leftarrow\Omega_{1,2}(B_{11}-B_{12}D_{22}B_{21})\\
&=\Omega_{1,2}(B_{11}-L_{11}^{-T}L_{21}^TD_{22}L_{21}L_{11}^{-1}).
\end{align*}
Here the binary representative of a multiplier may be replaced by any congruent integer matrix: for symmetric integer $Q$ and $R'\equiv R\pmod2$,
\[
\Omega_{1,2}(R'^TQR')=\Omega_{1,2}(R^TQR),
\]
since the diagonal changes by multiples of $4$ and the off-diagonals by multiples of $2$.
Replacing $Q$ by $\Omega_{1,2}(Q)$ likewise leaves this congruence unchanged.
Further,
\[
\begin{bmatrix}L_{11}&0\\L_{21}&I_s\end{bmatrix}^{-1}
=\begin{bmatrix}L_{11}^{-1}&0\\-L_{21}L_{11}^{-1}&I_s\end{bmatrix}.
\]
Starting from $B_{11}=\Omega_{1,2}(-L_{11}^{-T}D_{11}L_{11}^{-1})$, the displayed update gives the leading block of the negative congruence for this enlarged factor.
The new pivot block is $\Omega_{1,2}(-D_{22})$, and the new lower cross block is $\omega_1(D_{22}L_{21}L_{11}^{-1})$, as required by the same congruence.
Induction, starting with the empty leading block, proves \eqref{eq:gjl2}.

It remains to show that the trailing matrix updates give \eqref{eq:gjl3}.
Return to the partition of $A$ and $B$ in the lemma statement, and let $L_{21}=\omega_1(A_{21}L^{-T}D)$ denote the forward multiplier block for the remaining rows.
By the same trailing-update argument used for $v$, the congruence observation above, and $D^2=I$, we have
\begin{align*}
B_{22}&=\Omega_{1,2}(A_{22}-L_{21}DL_{21}^T)\\
&=\Omega_{1,2}(A_{22}-A_{21}L^{-T}DDD L^{-1}A_{12})\\
&=\Omega_{1,2}(A_{22}+A_{21}B_{11}A_{12}),
\end{align*}
where the last equality uses \eqref{eq:gjl2}.
\end{proof}

We now demonstrate directly how the two updates implement vertex and edge complementations of phased graph states, including their phases.
\begin{lemma}\label{lem:gj-pgs}
Under the assumptions of \cref{lem:gj_cor}, let $[B,v]=$ Gauss--Jordan-$\W_n(A,k)$ and $\bar v=\begin{bmatrix}v\\0\end{bmatrix}\in\F_2^n$. Then
\[
(H^{\otimes k}\otimes I^{\otimes(n-k)})\ket{A}
=\alpha X^{\bar v}\ket{B},
\]
where, if $n_1$ is the number of $1\times1$ pivots,
\begin{equation}\label{eq:alpha-pivots}
\alpha=\left(\frac{\sqrt2}{1+i}\right)^{n_1}=e^{-i\pi n_1/4}.
\end{equation}
\end{lemma}
\begin{proof}
Bring the current pivot first and write
\[
\bar A=\begin{bmatrix}D_i+2D(\nu)&R\\R^T&C\end{bmatrix},
\qquad \nu=\omega_2(d(\bar A_{11})),
\]
where $D_i=[1]$ or $D_i=\begin{bmatrix}0&1\\1&0\end{bmatrix}$. The update in \cref{alg:gj} is
\begin{equation}\label{eq:phase-pivot}
\bar B=\Omega_{1,2}\left(\begin{bmatrix}-D_i&D_iR\\R^TD_i&C-R^TD_iR\end{bmatrix}\right).
\end{equation}

For a single pivot, let $a_1$ be the first column of $O(\bar A)$. By \cref{thm:vertex-complementation},
\[
\ket{O(\bar A)}=\check X^{\langle1\rangle}\hat Z^{a_1}\ket{O(\bar B)}.
\]
As in the vertex-complementation argument, write
$H=\frac{1+i}{\sqrt2}\check X\check Z\check X$, so that
$H\check X=\frac{\sqrt2}{1+i}\hat X\hat Z$.
Since $\bar a_{11}=1+2\nu$ and $\ket{\bar A}=\check Z^{d(\bar A)}\ket{O(\bar A)}$, we obtain
\begin{align*}
H^{\langle1\rangle}\ket{\bar A}
&=(\check X^{\bar a_{11}})^{\langle1\rangle}
  \check Z^{d(\bar A)-\bar a_{11}e_1}
  H^{\langle1\rangle}\ket{O(\bar A)}\\
&=\frac{\sqrt2}{1+i}(X^\nu)^{\langle1\rangle}
  \check Z^{d(\bar A)-\bar a_{11}e_1}
  \hat Z^{\langle1\rangle}\hat Z^{a_1}\ket{O(\bar B)}\\
&=\frac{\sqrt2}{1+i}(X^\nu)^{\langle1\rangle}
  \check Z^{d(\bar A)-(\bar a_{11}+1)e_1-a_1}\ket{O(\bar B)}\\
&=\frac{\sqrt2}{1+i}(X^\nu)^{\langle1\rangle}\ket{\bar B}.
\end{align*}
Here the diagonal in \cref{eq:phase-pivot} satisfies
$d(\bar B)=d(\bar A)-(\bar a_{11}+1)e_1-a_1\pmod4$:
the pivot becomes $-1$, and each neighboring diagonal decreases by $1$.
The bit $\nu=\omega_2(\bar a_{11})$ is recorded in $v$ before this update.

For a two-dimensional pivot, let $a_1,a_2$ be the first two columns of $O(\bar A)$. By \cref{thm:edge-complementation},
\[
\ket{O(\bar A)}=H^{\langle1\rangle}H^{\langle2\rangle}
 Z^{a_1\ast a_2}\ket{O(\bar B)},
\]
where $\ast$ is the Hadamard (entrywise) product.
Since $\bar a_{11}=2\nu_1$ and $\bar a_{22}=2\nu_2$, we obtain
\begin{align*}
H^{\langle1\rangle}H^{\langle2\rangle}\ket{\bar A}
&=(X^{\nu_1})^{\langle1\rangle}(X^{\nu_2})^{\langle2\rangle}
 \check Z^{d(\bar A)-\bar a_{11}e_1-\bar a_{22}e_2}
 H^{\langle1\rangle}H^{\langle2\rangle}\ket{O(\bar A)}\\
&=(X^{\nu_1})^{\langle1\rangle}(X^{\nu_2})^{\langle2\rangle}
 \check Z^{d(\bar A)-\bar a_{11}e_1-\bar a_{22}e_2}
 Z^{a_1\ast a_2}\ket{O(\bar B)}\\
&=(X^{\nu_1})^{\langle1\rangle}(X^{\nu_2})^{\langle2\rangle}
 \check Z^{d(\bar A)-\bar a_{11}e_1-\bar a_{22}e_2-2a_1\ast a_2}\ket{O(\bar B)}\\
&=(X^{\nu_1})^{\langle1\rangle}(X^{\nu_2})^{\langle2\rangle}\ket{\bar B}.
\end{align*}
The last equality is the diagonal update in \cref{eq:phase-pivot}: the pivot diagonals become zero, and each common-neighbor diagonal decreases by $2$.
There is no scalar factor in this case.
These identities act on the entire current graph, including previously eliminated vertices. The accumulated $X$ gates are supported on those earlier vertices and commute with Hadamards on the new pivot. Induction therefore gives the stated identity, with one factor $\sqrt2/(1+i)$ per single pivot.
\end{proof}

\subsection{An amplitude formula for Hadamard-transformed phased graph states}
\label{sec:amp-formula}
A direct formula for the amplitudes of $\mat{H}^{\otimes n}\ket{\mat{A}}$ is given by \cref{thm:amp-formula}
if $\omega_1(\mat{A})$ is full rank over $\F_2$, via an $n$-step $\W_n$ Gauss--Jordan process.
We now consider the general case, with $\rank(\omega_1(\mat{A})) = k$.
To obtain this formula, we need to work out the action of the remaining $n-k$ Hadamard gates acting on the trailing $n-k$ qubits of $\ket{\mat{A}}$. The following well-known lemma describes the action of $n-k$ Hadamard gates on a general state.
\begin{lemma}\label{lem:h-second-half}

Let $\mat{x}_1\in\F_2^k$, $\mat{x}_2\in\F_2^{n-k}$, and $\ket{\psi}$ be an $n$-qubit quantum state. Then
\begin{equation}
\mel{\begin{bmatrix}
\mat{x}_1\\\mat{x}_2
\end{bmatrix}}{
\mat{I}^{\otimes k}\otimes\mat{H}^{\otimes (n-k)}}
{\psi}=2^{-(n-k)/2}
\sum_{\mat{y}_2\in\F_2^{(n-k)}}
(-1)^{\mat{y}_2^T\mat{x}_2}
\braket{\begin{bmatrix}
\mat{x}_1\\\mat{y}_2
\end{bmatrix}}{\psi}.
\end{equation}
\begin{proof}
For $x,y\in\F_2$, we have $\bra{x}\mat{H}\ket{y}=2^{-1/2}(-1)^{xy}$. Then the lemma follows.
\end{proof}
\end{lemma}
We leverage the fact that the leading $k\times k$ block of $\omega_1(\mat{A})$ has a block saddle point structure (the Schur complement of $\omega_1(\mat{A}_{11})\in\F_2^{k \times k}$ must be zero if the rank of $\omega_1(\mat{A})$ is $k$).
\begin{lemma}\label{lem:null}
Consider any $\mat{B}\in\J_n$ such that
\[\mat{B} = \begin{bmatrix} \mat{B}_{11} & \mat{B}_{12} \\ \mat{B}_{21} & \mat{B}_{22}\end{bmatrix}, \quad \mat{B}_{11}\in\J_k, \quad
\omega_1(\mat{B}_{22})=0.
\]
Then for all $\mat{x}_1\in\F_2^k$, $\mat{x}_2\in\F_2^{n-k}$,
\[
\mel{\begin{bmatrix}
\mat{x}_1\\\mat{x}_2
\end{bmatrix}}
{\mat{I}^{\otimes k}\otimes\mat{H}^{\otimes (n-k)}}
{\mat{B}}
=\begin{cases}
2^{-k/2}i^{-\mat{x}_1^T\mat{B}_{11}\mat{x}_1} & \mat{x}_2=\omega_1(\mat{B}_{21}\mat{x}_1+\mat{z})\\
0 & \text{otherwise.}
\end{cases}
\]
where $\mat{z}=\omega_2(d(\mat{B}_{22}))$.
\end{lemma}
\begin{proof}
By \cref{def:phased-graph-state} and the specified structure of $\mat{B}$, we have
\begin{equation}
\braket{\mat{x}}{\mat{B}}
=2^{-n/2}
i^{-\mat{x}_1^T\mat{B}_{11}\mat{x}_1}
(-1)^{
\mat{x}_2^T
(\mat{B}_{21}\mat{x}_1+\mat{z})}
\end{equation}
Then applying \cref{lem:h-second-half}, we get
\begin{align}
&\mel{\begin{bmatrix}
\mat{x}_1\\\mat{x}_2
\end{bmatrix}}
{\mat{I}^{\otimes k}\otimes\mat{H}^{\otimes (n-k)}}
{\mat{B}}\\
=&\;
2^{-(2n-k)/2}
i^{-\mat{x}_1^T\mat{B}_{11}\mat{x}_1}
\sum_{\mat{y}_2\in\F_2^{(n-k)}}
(-1)^{
\mat{y}_2^T
(\mat{B}_{21}\mat{x}_1+\mat{x}_2+\mat{z})}\\
=&\;
\begin{cases}
2^{-k/2}i^{-\mat{x}_1^T\mat{B}_{11}\mat{x}_1} & \mat{x}_2=\omega_1(\mat{B}_{21}\mat{x}_1+\mat{z})\\
0 & \text{otherwise.}
\end{cases}
\end{align}
This completes the proof.
\end{proof}

To obtain the full formula for $\bra{\mat{x}}\mat{H}^{\otimes n}\ket{\mat{A}}$, it suffices to combine \cref{lem:gj-pgs} and \cref{lem:null}.

\begin{theorem}\label{thm:amp-formula}
Let $A\in\W_n$, $k=\rank(\omega_1(A))$, and suppose $A_{11}\in\W_k$ is nonsingular. Let $[B,v]=$ Gauss--Jordan-$\W_n(A,k)$, partition $B$ conformally with $A_{11}$, and set
\[
z=\omega_2(d(B_{22})),\qquad
w=\begin{bmatrix}v\\z\end{bmatrix}.
\]
If $n_1$ is the number of $1\times1$ pivots, define
\begin{equation}\label{eq:gamma}
\gamma(A_{11})=e^{-i\pi n_1/4}.
\end{equation}
Then, for $x\in\F_2^n$ with leading $k$ entries $x_1$,
\begin{equation}\label{eq:amp-formula}
\bra{x}H^{\otimes n}\ket{A}
=\begin{cases}
2^{-k/2}\gamma(A_{11})i^{-(x_1+v)^TB_{11}(x_1+v)}
& x\oplus w\in\spn(\omega_1(A)),\\
0&\text{otherwise}.
\end{cases}
\end{equation}
Equivalently, with $\widehat v=\begin{bmatrix}v\\\mathbf1_{n-k}\end{bmatrix}$, the exponent on the support is $-(x+\widehat v)^TB(x+\widehat v)$, as in \cref{thm:amp-formula-intro}.
\end{theorem}

\begin{proof}
Let $\mat{x}_2$ denote the trailing $n-k$ elements of $\mat{x}$, and extend
$\mat{v}$ by zeros as $\bar v=(v^T,0^T)^T\in\F_2^n$.
By \cref{lem:gj-pgs},
\begin{align*}
\bra{\mat{x}}\mat{H}^{\otimes n}\ket{\mat{A}}
&=\bra{\mat{x}}(\mat{I}^{\otimes k}\otimes\mat{H}^{\otimes(n-k)})
  (\mat{H}^{\otimes k}\otimes\mat{I}^{\otimes(n-k)})\ket{\mat{A}}\\
&=\gamma(A_{11})\bra{\mat{x}}
  (\mat{I}^{\otimes k}\otimes\mat{H}^{\otimes(n-k)})
  \mat{X}^{\bar v}\ket{\mat{B}}\\
&=\gamma(A_{11})\bra{\mat{x}}\mat{X}^{\bar v}
  (\mat{I}^{\otimes k}\otimes\mat{H}^{\otimes(n-k)})\ket{\mat{B}}\\
&=\gamma(A_{11})\bra{\omega_1(\mat{x}+\bar v)}
  (\mat{I}^{\otimes k}\otimes\mat{H}^{\otimes(n-k)})\ket{\mat{B}}.
\end{align*}
Here $X^{\bar v}$ commutes with the remaining Hadamards because it acts only
on the first $k$ qubits. Since $\rank(\omega_1(A))=k$ and $\omega_1(A_{11})$
is nonsingular, its Schur complement is zero, so \cref{eq:gjl1} gives
$\omega_1(B_{22})=0$.
By \cref{lem:null}, when
$\mat{x}_2=\omega_1(\mat{B}_{21}(\mat{x}_1+\mat{v})+\mat{z})$, we have
\begin{align*}
\bra{\mat{x}}\mat{H}^{\otimes n}\ket{\mat{A}}
&=2^{-k/2}\gamma(A_{11})
 i^{-\omega_1(\mat{x}_1+\mat{v})^T\mat{B}_{11}\omega_1(\mat{x}_1+\mat{v})}\\
&=2^{-k/2}\gamma(A_{11})
 i^{-(\mat{x}_1+\mat{v})^T\mat{B}_{11}(\mat{x}_1+\mat{v})}.
\end{align*}
The last equality follows because replacing an integer vector by its binary
representative leaves a symmetric integer quadratic form unchanged modulo $4$.
Also, if $\mat{x}_2\ne\omega_1(\mat{B}_{21}(\mat{x}_1+\mat{v})+\mat{z})$,
then $\bra{\mat{x}}\mat{H}^{\otimes n}\ket{\mat{A}}=0$.

Finally, we show that this condition is equivalent to
$\mat{x}\oplus\mat{w}\in\spn(\omega_1(\mat{A}))$.
When the condition is satisfied, we have
\[
\begin{bmatrix}\mat{B}_{11}&\mat{B}_{12}\\\mat{B}_{21}&0\end{bmatrix}
\begin{bmatrix}\mat{x}_1+\mat{v}\\0\end{bmatrix}
\equiv
\begin{bmatrix}\mat{B}_{11}(\mat{x}_1+\mat{v})\\\mat{x}_2+\mat{z}\end{bmatrix}
\pmod2.
\]
By \cref{eq:gjl1}, equivalently the domain-range exchange property of the
principal pivot transform~\cite[Theorem 3.1]{tsatsomeros2000principal},
this implies
\[
\begin{bmatrix}\mat{A}_{11}&\mat{A}_{12}\\\mat{A}_{21}&\mat{A}_{22}\end{bmatrix}
\begin{bmatrix}\mat{B}_{11}(\mat{x}_1+\mat{v})\\0\end{bmatrix}
\equiv
\begin{bmatrix}\mat{x}_1+\mat{v}\\\mat{x}_2+\mat{z}\end{bmatrix}
\pmod2.
\]
Hence $\mat{x}\oplus\mat{w}$ is in the span of $\omega_1(\mat{A})$ over $\F_2$.
Conversely, the first $k$ columns of $\omega_1(A)$ span its column space.
If $\mat{x}\oplus\mat{w}$ is in that space, its representation using these
columns has coefficients $\omega_1(B_{11}(x_1+v))$, since
$\omega_1(B_{11})=A_{11}^{\#}$. Its trailing coordinates therefore give
$\mat{x}_2=\omega_1(\mat{B}_{21}(\mat{x}_1+\mat{v})+\mat{z})$.
This proves \cref{eq:amp-formula}.

For the full-vector form, set $a=x_1\oplus v$. Since $B_{22}=2D(z)$ and
$B_{21}a\oplus z=x_2$ on the support,
\begin{align*}
(x+\widehat v)^TB(x+\widehat v)
&\equiv a^TB_{11}a+2(x_2\oplus\mathbf1)^T(B_{21}a+z)\\
&\equiv a^TB_{11}a+2(x_2\oplus\mathbf1)^Tx_2\\
&\equiv a^TB_{11}a\pmod4.
\end{align*}
The last inner product vanishes because complementary bits never both equal $1$.
\end{proof}

The amplitude formula is a finite binary instance of Weil's theorem on Fourier transforms of quadratic phases~\cite[Theorem~2]{weil1964certains}: a nondegenerate quadratic phase transforms to another quadratic phase, up to a constant scalar. Here $H^{\otimes n}$ is the Fourier transform on $\F_2^n$. With $M=\omega_1(A)$ and $q(t)=t^TAt\pmod4$, its amplitudes are the quadratic Gauss sums
\[
F_A(x):=\bra{x}H^{\otimes n}\ket{A}
=2^{-n}\sum_{t\in\F_2^n}i^{-q(t)}(-1)^{x^Tt}.
\]
The identity $q(t\oplus z)\equiv q(t)+q(z)+2t^TMz\pmod4$ shows that, when $M$ is nonsingular (the nondegenerate case), translating the summation variable by $z=M^{-1}x$ gives $F_A(x)=F_A(0)i^{q(z)}$, with $\abs{F_A(0)}=2^{-n/2}$. Thus solving a binary linear system determines the output-dependent quadratic phase.

When $M$ has rank $k<n$, translation by $h\in\ker M$ instead forces $F_A(x)=0$ unless $q(h)\equiv2x^Th\pmod4$ for every such $h$. On $\ker M$, $q$ is additive and takes only the values $0$ and $2$, so these are linear equations over $\F_2$. Their solution set is the affine support $w\oplus\spn(M)$ in \cref{thm:amp-formula}. The remaining sum is a nondegenerate quadratic Gauss sum in $k$ variables, giving magnitude $2^{-k/2}$ on the support. Our derivation evaluates these sums by the one- and two-variable eliminations corresponding to vertex and edge complementations, obtaining the support shift, quadratic phase, and global scalar directly from the Gauss--Jordan process.

In summary, this section reduced the Hadamard transform of a phased graph state to the generalized Gauss--Jordan process and gave its amplitudes in closed form (\cref{thm:amp-formula}); Section \ref{sec:pgs-sim} builds the simulation algorithms on these results.

\section{Phased graph state simulation}\label{sec:pgs-sim}
In this section we use the relation between phased graph states and Gaussian elimination as presented in \cref{sec:pgs-ge} with techniques in numerical linear algebra to solve the phased graph state simulation problem as defined in \cref{def:pgs-sim}. Throughout this section and the next, we take $\tau\ge1$ (replace the width by $\max(1,\tau)$ when necessary), and a tree decomposition is assumed to have $O(n/\tau)$ bags of size $O(\tau)$, using the normalized decomposition of \cref{def:sparse-ldl-setting}. Writing $\|T\|_\omega^\omega:=\sum_{B\in T}|B|^\omega$ for a rooted full binary tree decomposition $T$, each operation below costs $O(|B|^\omega)$ per bag, hence $O(\|T\|_\omega^\omega)$ in total~\cite[Theorem 5.9]{solomonik2025fast}; the bounds below state the normalized width-$\tau$ case, with the batched terms charged at the maximum bag size~$\tau$, where $\|T\|_\omega^\omega=O(n\tau^{\omega-1})$.

\subsection{Sparse LDL decomposition}
As shown in \cref{lem:gj_cor}, the generalized Gauss--Jordan process is closely tied to LDL decompositions. Therefore, the sparse LDL algorithm~\cite{solomonik2025fast} can be applied here to solve the phased graph state simulation problem in running time depending on the treewidth of the graph. However, extra caution is needed because the sparse LDL algorithm represents the decomposition in an implicit form and we need to make sure all of the intermediate steps for phased graph state simulation can be performed efficiently.

\begin{definition}[Sparse LDL setting]\label{def:sparse-ldl-setting}
Let $A\in\J_n$ and let a tree decomposition of the graph corresponding to $A$ of width $\tau$ be normalized to a rooted full binary tree decomposition with $O(n/\tau)$ bags of size $O(\tau)$~\cite[Corollary 5.3]{solomonik2025fast}. Let $k=\rank(\omega_1(A))$, let $P$ be the permutation that the relabelling of \cref{sec:pgs-gj} and the sparse LDL algorithm use to put $A$ in $\W_n$, and let $\omega_1(P^TAP)=\omega_1(LDL^T)$ be the reduced LDL decomposition produced by the sparse LDL algorithm~\cite[Definition 5.5 (Peeled Implicit LDL Decomposition) and Theorem 5.9]{solomonik2025fast}, where $L\in\F_2^{n\times k}$ is lower trapezoidal with unit diagonal, $D\in\F_2^{k\times k}$ is block diagonal with $1\times1$ blocks $1$ and $2\times2$ blocks $\begin{bmatrix}0&1\\1&0\end{bmatrix}$, and, writing $L=\begin{bmatrix}L_1\\L_2\end{bmatrix}$ with $L_1\in\F_2^{k\times k}$ and $L_2\in\F_2^{(n-k)\times k}$, we set $\tilde L=\begin{bmatrix}L_1&0\\L_2&I\end{bmatrix}\in\F_2^{n\times n}$, the leading $k\times k$ block of $P^TAP$ being its full-rank block $A_{11}$. We use the elimination order induced by the post-ordering of~\cite[Definition 5.8 (Tree Decomposition Post-Ordering)]{solomonik2025fast}. The algorithm represents this decomposition implicitly, without forming $L$: $O(n/\tau)$ stored transformations, one per bag of the decomposition, each acting on the $O(\tau)$ rows and columns of its bag, and the rank-deficient rows are handled by the algorithm's peeling transformations, which act on the same $O(\tau)$ coordinates of a bag and its children ($k$ may be less than $n$)~\cite[Theorem 5.9 and Section 7 (Conclusion)]{solomonik2025fast}; inverses of $L_1$ and $\tilde L$ are taken over $\F_2$.
\end{definition}
Throughout this section and the next, input and output data are transported with the relabelling that puts $A$ in $\W_n$: an index set $S$ and data $y$ given in the original labelling are replaced by $P^TS$ and $P^Ty$, the support condition and the free coordinates are read in the elimination order, and samples or amplitudes are mapped back by $P$. The number of free coordinates $\ell=n-\abs{S}$ is unchanged, because $P$ is a permutation.

\begin{lemma}\label{lem:tinv}
Let $A\in\F_2^{n\times n}$ be symmetric and consider the setting of \cref{def:sparse-ldl-setting}.
Let $A_{11}$ be the maximal leading square full-rank block of $P^TAP$, and consider
\[(P^TAP)^g = \begin{bmatrix} A_{11}^{\#} & A_{11}^{\#}A_{12} \\ A_{21}A_{11}^{\#} & 0\end{bmatrix}\]
Let $X_{B}=P_B^TXP_B$ be the restriction of $X$ to the vertices of a bag $B$ of the tree decomposition (each column of $P_B$ is a distinct elementary vector) for any $X\in\F_2^{n\times n}$.
The set of matrices $\{(P(P^TAP)^gP^T)_B\}$, over the $O(n/\tau)$ bags $B$, may be computed in time $O(n\tau^{\omega-1})$.
\end{lemma}
\begin{proof}
We can compute these blocks by induction, starting from the root block of the decomposition.
For the root block $R$, it suffices to directly compute $A^g_R$ from the Schur complement (which may be formed from the lower right block of the triangular part of the $L$ factor and the corresponding part of $D$).
For any non-root block $B$, we assume that we have access to $A^g_Q$ where $Q$ is the parent block of $B$.
Compute $A^h_Q$ by zeroing out all rows/columns of $A^g_Q$ corresponding to peeled rows/columns (not part of the leading full rank block), and removing all rows/columns that are not in $B$.
Define $\bar{A}^h_Q$ by setting the diagonal entries corresponding to zero rows/columns, to $1$.
Now, let $K\subseteq B$ the set of vertices in $B$ not contained in $Q$ and let $M$ be the remaining vertices in $B$.
Assume $Y=A^g_B$ is ordered so that vertices in $K$ appear after vertices in $M$.
We may obtain the leading $|M|$-dimensional block of $Y$, $Y_{11}$ directly from $A^g_Q$.
To get $Y_{21}$ and $Y_{22}$, we use products with the off-diagonal block between $K$ and $M$, as follows.
In particular, for two adjacent blocks $B$ and $Q$, order the coordinates as $K$, $M$, and the remaining vertices, so that the initial $A$ has the block structure,
\[A=\begin{bmatrix}A_{11} & A_{12} & 0 \\ A_{21} & A_{22} & A_{23} \\ 0 & A_{32} & A_{33}\end{bmatrix}, \quad A_{11}\in\F_2^{|K|\times |K|}, A_{22}\in\F_2^{|M|\times |M|}.\]
Let $B=A^h_Q$ and $\bar{B}=\bar{A}^h_Q$, where we have (in this case), $A_Q=\begin{bmatrix}A_{22} & A_{23} \\ A_{32}  & A_{33}\end{bmatrix}$.
Define $A^h_{11}$ and $\bar{A}^h_{11}$ analogously to $A^h_Q$ and $\bar{A}^h_Q$.
Then, we have that
\[P(P^TAP)^gP^T=\begin{bmatrix}A^g_{11}-\bar{A}^h_{11}A_{12}BA_{21}\bar{A}^h_{11} & \bar{A}^h_{11}A_{12}\bar{B}_{11} & \ast \\ \bar{B}_{11}A_{21}\bar{A}^h_{11} & (A^g_Q)_{11}& (A^g_Q)_{12} \\ \ast & (A^g_Q)_{21} & (A^g_Q)_{22}\end{bmatrix},\]
where $\ast$ denotes blocks we are not interested in computing.
Clearly, the block matrix multiplications we are interested in computing have dimension smaller than the block size.
When there are more than two blocks, in particular when $B$ has descendants, the situation does not change, except that we should start with the Schur complement obtained after eliminating all descendants of $B$ instead of simply $A$.
Hence, it suffices to do $O(\tau^\omega)$ work for each of the $O(n/\tau)$ blocks, leading to the complexity in the lemma.
\end{proof}
An alternative algorithm (which also proves the lemma) computes an LDL decomposition for every block with that block as root; the desired block of the inverse is then read off at the root.
The $O(n/\tau)$ LDL decompositions can be computed in $O(\tau^{\omega})$ amortized time each by shifting the root along an Euler tour of the tree: consecutive tour vertices are adjacent, so each shift changes only the $O(1)$ blocks incident to the single edge traversed.

Leveraging the sparse LDL and partial inverse construction algorithms, we provide the following lemma to summarize the complexities of subroutines that are needed for phased graph state simulation based on the fast LDL algorithm described in~\cite{solomonik2025fast}.
\begin{lemma}\label{lem:gj-complexity}\label{lem:gj-data}
Given $A\in\W_n$ such that $\rank(\omega_1(A))=k$ and a tree decomposition of the graph corresponding to $A$ with width $\tau$, the fast LDL algorithm~\cite{solomonik2025fast} can compute an implicit representation of $L\in\F_2^{n\times k}$ and $D\in\F_2^{k\times k}$ in $O(n\tau^{\omega-1})$ time, such that the following conditions are all satisfied:
\begin{itemize}
\item $\omega_1(A)=\omega_1(LDL^T)$ is a reduced LDL decomposition over $\F_2$.
\item Let $L=\begin{bmatrix}L_1\\L_2\end{bmatrix}$ where $L_1\in\F_2^{k\times k}$. Then the matrix-vector product with $L_1$, $L_1^{-1}$ or $L_2L_1^{-1}$, and their transposes, over $\F_2$ could be performed in $O(n\tau)$ time.
\item Let $\tilde L=\begin{bmatrix}L_1&0\\L_2&I\end{bmatrix}$ for $L_1$, $L_2$ as defined above. Then for any $X\in\F_2^{n\times s}$ with $s\le\tau$, $\tilde LX$, $\tilde L^TX$, $\tilde L^{-1}X$, or $\tilde L^{-T}X$ over $\F_2$ could be computed in $O(n\tau^{\omega-1})$ time. For a single right-hand side, each of these products costs $O(n\tau)$ time.
\end{itemize}
Moreover, $v$ and $w$ as defined in \cref{thm:amp-formula}, as well as the scalar $\gamma(A_{11})$, can be computed along with the LDL decomposition without increasing the time complexity. The quadratic form of $B_{11}$ in that theorem can be evaluated on any $s\in\N$ given vectors in $O(n\tau^{\omega-1}+sn\tau^{\omega-2})$ time.
\end{lemma}
\begin{proof}
The fast LDL algorithm described in~\cite[Theorem~5.9 and Corollary~5.10]{solomonik2025fast} takes $O(n\tau^{\omega-1})$ time and represents $L=\omega_1\left(\prod_i^m Q_i\right)$ for $m=O(n/\tau)$ such that each $Q_i$ acts on a submatrix of size $O(\tau)$. The elimination permutation is absorbed as in \cref{def:sparse-ldl-setting}. Applying these factors or their transposes gives products with $L$ or $L^T$; padding and restriction give products with $L_1$, $L_2$, and their transposes. The factor $L_1$ of the leading nonsingular principal block is available explicitly~\cite[proof of Corollary~5.11]{solomonik2025fast}, so solves with $L_1$ and $L_1^T$ use forward and backward substitution on the same block structure. Products with $\tilde L$ and $\tilde L^T$ add the trailing identity block. For the inverse operations, use
\[
\tilde L^{-1}=\begin{bmatrix}L_1^{-1}&0\\-L_2L_1^{-1}&I\end{bmatrix}
\]
and its transpose. Thus these operations require only triangular solves and the preceding products, not inverses of the rectangular peeling transformations. Each operation costs $O(\tau^\omega)$ per bag for at most $\tau$ right-hand sides, or $O(\tau^2)$ for one, giving $O(n\tau^{\omega-1})$ or $O(n\tau)$ in total. Therefore, the three conditions in the lemma are satisfied.

For computing $v$ (recall that $v_i$ corresponds to the second bit of the diagonal in the Gauss--Jordan process when it is eliminated), it suffices to keep track of the second bit of Schur complement updates to the diagonal, without using it in the binary pivot decisions or multipliers. It is easy to check that computing this does not change the complexity of the LDL (no greater asymptotic number of block operations is required, even when considering fast matrix multiplication); only the diagonal accounting is performed modulo $4$.

However, obtaining the second part of $w$ (recall that it is defined to be $\omega_2(d(B_{22}))$ for $B_{22}$ as given in \cref{lem:gj_cor} that corresponds to nodes that are not eliminated in the Gauss--Jordan process) is trickier with the sparse LDL procedure in~\cite{solomonik2025fast}, because linearly-dependent rows are peeled off early in the factorization and we do not have the corresponding rows of $L$ explicitly.
To overcome this without an efficient algorithm for explicit construction of the rows of $L$, we may leverage \cref{lem:tinv}, which gives an explicit formula for computing the generalized inverse based on the Gauss--Jordan process for all diagonal blocks. For the phase data, \cref{eq:gjl3} gives $B_{22}=\Omega_{1,2}(A_{22}+A_{21}B_{11}A_{12})$, with $B_{11}$ represented by \cref{eq:gjl2}. To obtain $w$, it suffices to keep track of the second bit of the diagonal in these quadratic congruences (also when computing the Schur complement itself). Each reconstruction uses the same $O(\tau)$-sized blocks, so the time complexity is unchanged.

For the quadratic form, given $y\in\F_2^k$, solve $L_1z=y$ over $\F_2$. By \cref{eq:gjl2},
\begin{equation}\label{eq:exp-implicit}
y^TB_{11}y\equiv-z^TDz\pmod4,\qquad z=\omega_1(L_1^{-1}y).
\end{equation}
Since $D$ has blocks of size at most two, evaluating $z^TDz$ modulo $4$ takes $O(k)$ further operations. Processing $s$ vectors in batches of at most $\tau$ therefore costs $O(n\tau^{\omega-1}+sn\tau^{\omega-2})$, including a partially filled batch. Finally, $\gamma(A_{11})=e^{-i\pi n_1/4}$ requires only the number $n_1$ of single pivots and is recorded at no extra asymptotic cost.
\end{proof}
\begin{remark}\label{rem:implicit-required}
The algorithms in this section do not require an explicit reduced factor $L$: when the matrix is far from full rank the explicit lower-trapezoidal factor may be expensive to form~\cite[Section 7 (Conclusion)]{solomonik2025fast}, whereas the block transformations, the triangular solves and the peeling transformations above run in $O(n\tau^{\omega-1})$ time. The conditioning step in \cref{sec:weak-sim-pgs} uses an additional sparse LU factorization of a submatrix of the original input.
\end{remark}

\subsection{Weak simulation of phased graph states}\label{sec:weak-sim-pgs}
Since the nonzero amplitudes of a phased graph state have equal magnitude as given in \cref{thm:amp-formula}, we are interested in uniform sampling of
\begin{equation}
\set*{x\in\mathcal{X}_{S,y}}{x\oplus w\in\spn(\omega_1(A))}.
\end{equation}
We can assume, WLOG, that $w=0$, provided the prescribed values are also shifted: sample $x'$ with $x'\in\spn(\omega_1(A))$ and $x'_S=y_S\oplus w_S$, and return $x=x'\oplus w$.
Hence, we focus on sampling
\begin{equation}
\mathcal{Y}_{A,S,y}=\set*{x\in\mathcal{X}_{S,y}}{x\in\spn(\omega_1(A))}.
\end{equation}

Let $r=\rank(\omega_1(A))$ and use the elimination ordering, in which a nonsingular principal $r\times r$ block comes first. Let $J=\{1,\ldots,r\}$ denote its indices in this ordering and define the matrix of its columns by
\begin{equation}\label{eq:sparse-support-basis}
G=\omega_1\!\left(\begin{bmatrix}A_{11}\\A_{21}\end{bmatrix}\right)\in\F_2^{n\times r}.
\end{equation}
The columns of $G$ form a basis of $\spn(\omega_1(A))$. In this ordering $G=LDL_1^T$, so $G$ and the reduced factor $L$ differ only by an invertible change of coefficient coordinates. Unlike the explicit reduced factor $L$, however, $G$ is a submatrix of the original sparse input. We use this sparse basis to impose the fixed-coordinate constraints by a second sparse factorization.

For a full-column-rank matrix $G\in\F_2^{n\times r}$, write
\[
\mathcal{Y}_{G,S,y}=\set*{Gz}{z\in\F_2^r,\ G_Sz=y_S},
\]
where $G_S$ denotes the rows indexed by $S$. Then $\mathcal{Y}_{A,S,y}=\mathcal{Y}_{G,S,y}$ for the basis in \eqref{eq:sparse-support-basis}. All operations in the rest of this subsection are over $\F_2$.

\begin{lemma}[Affine sampling by elimination]\label{lem:affsamp}
Let $G\in\F_2^{n\times r}$ have full column rank, let $S\subseteq[n]$, and let $y\in\F_2^n$. Set $s=\rank(G_S)$ and $\nu=r-s$. Choose $I\subseteq S$ with $|I|=s$ and reorder the columns as $G=[G_p\ G_f]$ so that
\[
C_0=(G_p)_I\in\F_2^{s\times s}\quad\text{is nonsingular},\qquad
C_1=(G_f)_I\in\F_2^{s\times\nu}.
\]
Let $h=C_0^{-1}y_I$. The set $\mathcal{Y}_{G,S,y}$ is nonempty if and only if
\begin{equation}\label{eq:sampling-consistency}
(G_p)_Sh=y_S.
\end{equation}
When this condition holds,
\begin{equation}\label{eq:affine-rep}
\mathcal{Y}_{G,S,y}=x_0\oplus\spn(K),\qquad
x_0=G_ph,\qquad K=G_f\oplus G_pC_0^{-1}C_1.
\end{equation}
Moreover, $K_S=0$, $\rank(K)=\nu\le\ell=n-|S|$, and $x_0\oplus Kc$ is uniform on $\mathcal{Y}_{G,S,y}$ when $c$ is uniform on $\F_2^\nu$.
\end{lemma}
\begin{proof}
The rows indexed by $I$ span the row space of $G_S$. The solutions of the equations on these rows are exactly
\begin{equation}\label{eq:affine-param}
z=\begin{bmatrix}h\oplus C_0^{-1}C_1c\\c\end{bmatrix},\qquad c\in\F_2^\nu.
\end{equation}
Their differences are annihilated by $G_I$, and hence by $G_S$. Thus all these vectors satisfy $G_Sz=y_S$ if the vector with $c=0$ does, and otherwise none does. This proves \eqref{eq:sampling-consistency} and, after multiplication by $G$, \eqref{eq:affine-rep}.
The matrix
\[
Z=\begin{bmatrix}C_0^{-1}C_1\\I_\nu\end{bmatrix}
\]
is a basis of $\ker(G_S)$, so $K=GZ$ satisfies $K_S=0$. Since $G$ is injective and $Z$ has independent columns, $\rank(K)=\nu$. Only the $\ell$ rows outside $S$ can be nonzero, giving $\nu\le\ell$. Finally, \eqref{eq:affine-param} parameterizes the solution set bijectively, and multiplication by $G$ is injective, proving uniformity.
\end{proof}
The formulas also cover $s=0$ or $\nu=0$, with the usual empty-matrix conventions. Inverses such as $C_0^{-1}$ are applied by triangular solves, not formed explicitly. The columns of $K$ cancel the prescribed coordinates by combinations of columns of $G$; simply zeroing those coordinates would not, in general, preserve $\spn(G)$.

\begin{lemma}[Sparse conditioning factorization]\label{lem:sampling-complexity}
Given $A\in\J_n$ and a normalized width-$\tau$ tree decomposition as in \cref{def:sparse-ldl-setting}, the basis $G$ in \eqref{eq:sparse-support-basis}, the index sets and a factorization of $C_0$ in \cref{lem:affsamp}, and the consistency test \eqref{eq:sampling-consistency} can be obtained in $O(n\tau^{\omega-1})$ time. After this preprocessing, applying the map $c\mapsto x_0\oplus Kc$ to $q\ge1$ vectors costs $O(n\tau^{\omega-1}+qn\tau^{\omega-2})$ time. Alternatively, the $\ell\times\nu$ matrix $K_{\bar S}$, where $\bar S=[n]\setminus S$, can be formed in $O(n\tau^{\omega-1}+\nu n\tau^{\omega-2})$ time.
\end{lemma}
\begin{proof}
The sparse LDL algorithm identifies $J$ and the nonsingular principal block in $O(n\tau^{\omega-1})$ time~\cite[Corollary 5.10 and the proof of Corollary 5.11]{solomonik2025fast}. Extracting $G$ from the original input requires no dense factor construction.
Set $C=G_S$, the restriction of these columns to the prescribed rows, and consider its symmetric dilation
\begin{equation}\label{eq:conditioning-dilation}
\mathcal{D}_C=\begin{bmatrix}0&C\\C^T&0\end{bmatrix}.
\end{equation}
For each bag $B$ of the original tree decomposition, form a bag containing a row copy of every index in $B\cap S$ and a column copy of every index in $B\cap J$. Every nonzero of $C$ is covered, including a nonzero diagonal entry of the original matrix, and the bags containing each copy remain connected. These bags have at most twice the original size; in particular, a width-$\tau$ decomposition gives width at most $2\tau+1$. After normalization there are $O(n/\tau)$ bags of size $O(\tau)$, and the dilation has at most $2n$ vertices.

Apply the sparse LU construction of~\cite[Corollary 5.12]{solomonik2025fast} to $C$. As specified in its proof, the construction identifies an $s\times s$ nonsingular submatrix $C_0$ and supplies its explicit LU factorization even if $C$ is far from full rank. Its row and column indices give $I$, $G_p$, and $G_f$. This costs $O(n\tau^{\omega-1})$; no explicit full rectangular LU factorization of $C$ is required.

Triangular solves with the nonsingular pivot block use the same $O(\tau)$-sized frontal blocks as the factorization~\cite[Theorem 5.9 and the proof of Corollary 5.12]{solomonik2025fast}. For a block of at most $\tau$ right-hand sides, forward and backward substitution cost $O(\tau^\omega)$ per bag and $O(n\tau^{\omega-1})$ overall. Products with the sparse input submatrices $G_p$, $G_f$, and $C_1$ have the same bound: assign each nonzero to one containing bag and multiply the resulting $O(\tau)\times O(\tau)$ blocks by the corresponding rows of the right-hand sides. Processing $q$ columns in $\lceil q/\tau\rceil$ batches therefore costs
\[
O\bigl(\lceil q/\tau\rceil n\tau^{\omega-1}\bigr)
=O(n\tau^{\omega-1}+qn\tau^{\omega-2}).
\]
One solve for $h$ and one sparse product suffice for the consistency test and $x_0$. For each batch of samples, use
\begin{equation}\label{eq:sample-pivot-solve}
C_0z_p=y_I\oplus C_1c,\qquad x=G_pz_p\oplus G_fc.
\end{equation}
Alternatively, solve $C_0R=C_1$ for its $\nu$ right-hand sides and form $K_{\bar S}=(G_f)_{\bar S}\oplus(G_p)_{\bar S}R$. This gives the stated basis-construction bound. When $\nu=0$, there is no direction matrix to form.
\end{proof}
The additional factorization depends on $S$ but not on $y_S$; it is reused for all samples and for other prescribed values on the same set. It also gives the affine-support dimension $\nu=r-\rank(G_S)$ needed in \cref{lem:circ-as-pgs}, within the same preprocessing bound.

\begin{algorithm}[t]
\caption{$x=$ sample-$\mathcal{Y}$($G,S,y$)}\label{alg:pgs-sample}
\begin{algorithmic}[1]
\Require the sparse basis $G$ of \eqref{eq:sparse-support-basis}, $S\subseteq[n]$, and $y\in\F_2^n$.
\Ensure a uniform sample from $\mathcal{Y}_{G,S,y}$, or \textbf{empty} if the set is empty.
\State Use \cref{lem:sampling-complexity} to choose $I$, partition $G=[G_p\ G_f]$, and factor $C_0=(G_p)_I$; set $C_1=(G_f)_I$.
\State Solve $C_0h=y_I$. If $(G_p)_Sh\ne y_S$, \Return \textbf{empty}.
\State Sample $c\in\F_2^\nu$ uniformly at random.
\State Solve $C_0z_p=y_I\oplus C_1c$.
\State \Return $G_pz_p\oplus G_fc$.
\end{algorithmic}
\end{algorithm}
The first two steps of \cref{alg:pgs-sample} are preprocessing. Repeating the remaining steps with independent choices of $c$ gives independent samples. Equivalently, after forming $K_{\bar S}$, return the free coordinates $(x_0)_{\bar S}\oplus K_{\bar S}c$.

The complexity of the sampling algorithm described above is given in the following theorem.
\begin{theorem}[Complexity of weak phased graph state simulation]\label{thm:weak-sim-pgs}
Given any $y\in\F_2^n$, $S\subseteq[n]$, $A\in\J_n$ and a tree decomposition of the graph corresponding to $A$ with width $\tau$, uniformly sampling $k\in\N$ elements from
\[\set*{x\in\mathcal{X}_{S,y}}{x\oplus w\in\spn(\omega_1(A))},\]
where $w$ is defined as in \cref{thm:amp-formula}, or reporting the set is empty, can be done in
\[
O\bigl(n\tau^{\omega-1}+\min(kn\tau^{\omega-2},\ell n\tau^{\omega-2}+k\ell^{\omega-1})\bigr)
\]
time for $\ell=n-|S|$, with each sample represented by its $\ell$ free coordinates.
\end{theorem}
\begin{proof}
Assume, WLOG, $A\in\W_n$, transporting $S$ and $y$ with the elimination permutation as described in \cref{def:sparse-ldl-setting}. By \cref{lem:gj-complexity}, compute the sparse LDL and $w$ in $O(n\tau^{\omega-1})$ time. Form $G$ from its original pivot columns and replace $y$ by $\tilde y=y\oplus w$. The desired set is
\[
w\oplus\mathcal{Y}_{G,S,\tilde y}.
\]
By \cref{lem:affsamp}, the consistency test is exact and, when it succeeds, \cref{alg:pgs-sample} followed by the shift $w$ gives a uniform sample.

Applying the sampler directly to $k$ random coefficient vectors costs $O(n\tau^{\omega-1}+kn\tau^{\omega-2})$ by \cref{lem:sampling-complexity}. Alternatively, form $K_{\bar S}$ in $O(n\tau^{\omega-1}+\ell n\tau^{\omega-2})$ time, since $\nu\le\ell$. For $k\ge\ell\ge1$, multiplying this $\ell\times\nu$ matrix by the random coefficient vectors in batches of $\ell$ costs $O(k\ell^{\omega-1})$. For $k<\ell$, the direct route already costs at most $O(n\tau^{\omega-1}+\ell n\tau^{\omega-2})$, so it also satisfies the second bound. Taking the better route proves the stated minimum. If $\ell=0$, only consistency needs to be checked and there are no free output bits.
\end{proof}

\subsection{Strong simulation of phased graph states}\label{sec:strong-sim-pgs}
The strong simulation of phased graph state simulation boils down to the evaluation of the amplitudes formula in \cref{thm:amp-formula}, which could be evaluated at the same cost as the sparse LDL algorithm of~\cite{solomonik2025fast}.

\begin{theorem}[Complexity of strong phased graph state simulation]\label{thm:strong-sim-pgs}
Given $k$ binary strings $x^{(1)},\dots,x^{(k)}\in\F_2^n$, $A\in\J_n$ and a tree decomposition of the graph corresponding to $A$ with width $\tau$, the amplitudes formula \cref{eq:amp-formula} in \cref{thm:amp-formula} can be evaluated for all $x=x^{(1)},\dots,x^{(k)}$ in $O(n\tau^{\omega-1}+kn\tau^{\omega-2})$ time.
\end{theorem}
\begin{proof}
Assume WLOG that $A\in\W_n$ and let $\rho=\rank(\omega_1(A))$; the parameter $k$ of \cref{eq:amp-formula} is $\rho$ here. By \cref{lem:gj-complexity} we can compute an implicit LDL decomposition of $\omega_1(A)$ in $O(n\tau^{\omega-1})$ time, and by \cref{lem:gj-data} the data $v$, $w$, $\gamma$ and an implicit representation of $\mat{B}_{11}$ in \cref{thm:amp-formula} at no extra asymptotic cost. Then we check the complexity of evaluating $\cref{eq:amp-formula}$, which is restated below (with $x_1$ denoting the leading $\rho$ coordinates of $x$),
\begin{equation}
\bra{\mat{x}}\mat{H}^{\otimes n}\ket{\mat{A}}
=\begin{cases}
2^{-\rho/2}\gamma(A_{11})
i^{-(\mat{x}_1+\mat{v})^T\mat{B}_{11}(\mat{x}_1+\mat{v})} & \mat{x}\oplus\mat{w} \in\spn(\omega_1(\mat{A})) \\
0 & \text{otherwise}.
\end{cases}\label{eq:amp-formula-2}
\end{equation}

Then $x\oplus w\in\spn(\omega_1(A))$ could be checked by solving the systems $Lz=x\oplus w$ for $x=x^{(1)},\dots,x^{(k)}$. This can be done in $O(n\tau^{\omega-1}+kn\tau^{\omega-2})$ time using $\lceil k/\tau\rceil$ batches of fast matrix multiplications by \cref{lem:gj-complexity}.

For an input on the support, let $z$ be the coefficient vector obtained from $Lz=x\oplus w$; equivalently, $L_1z=x_1\oplus v$. The same solve used to test support therefore supplies the phase:
\[
-(x_1+v)^TB_{11}(x_1+v)\equiv z^TDz\pmod4.
\]
By \cref{lem:gj-data}, evaluating these block-diagonal quadratic forms costs only $O(kn)$ additional operations. Thus both support testing and phase evaluation have the stated total cost.
\end{proof}

In applications such as circuit simulation, the $x$'s for which to compute amplitudes often share common parts. In this case, it can be more efficient to precompute the part of the amplitudes formula corresponding the shared parts of $x$ to avoid repetitive work in a similar way as in weak simulation. This is detailed in the following theorem.
\begin{theorem}\label{thm:strong-sim-pgs-fixed-bits}
Given $k$ binary strings $x^{(1)},\dots,x^{(k)}\in \mathcal{X}_{S,y}$ for some $S\subseteq[n]$ and $y\in\F_2^n$ (only the restriction $y_S$ is used), $A\in\J_n$ and a tree decomposition of the graph corresponding to $A$ with width $\tau$, the amplitudes formula \cref{eq:amp-formula} in \cref{thm:amp-formula} can be evaluated for all $x=x^{(1)},\dots,x^{(k)}$ in $O(n\tau^{\omega-1}+\ell n\tau^{\omega-2}+k\ell^{\omega-1})$ time for $\ell=n-\abs{S}$.
\end{theorem}
\begin{proof}
If $k<\ell$, apply \cref{thm:strong-sim-pgs} directly; its cost is bounded by $O(n\tau^{\omega-1}+\ell n\tau^{\omega-2})$. Assume henceforth that $k\ge\ell\ge1$.
To check the support condition of \cref{eq:amp-formula}, compute the sparse LDL and the phase data as in \cref{lem:gj-complexity}, and apply \cref{lem:sampling-complexity} with the shifted prescribed values $\tilde y_S=y_S\oplus w_S$. If the consistency test fails, return zero for all amplitudes. Otherwise, \cref{lem:affsamp} gives a feasible point $x_*=w\oplus x_0$ and a basis $K_{\bar S}\in\F_2^{\ell\times\nu}$ in $O(n\tau^{\omega-1}+\ell n\tau^{\omega-2})$ time. Thus
\begin{equation}\label{eq:fixed-support-test}
\mel{x}{H^{\otimes n}}{A}\ne0
\quad\Longleftrightarrow\quad
(x\oplus x_*)_{\bar S}\in\spn(K_{\bar S}).
\end{equation}
Factor $K_{\bar S}$ by rank-revealing Gaussian elimination in $O(\ell^\omega)$ time. For example, choosing $\nu$ independent rows gives a nonsingular submatrix $R$; solve with $R$ to obtain candidate coefficients from those entries of each test vector and check all remaining entries. Batched solves and products test all $k$ vectors in $O(k\ell^{\omega-1})$ time. The factorization cost is absorbed by this term because $k\ge\ell$. If $\nu=0$, the test is simply $x=x_*$ on the free coordinates.

It remains to restrict the quadratic form of $B_{11}$ to the varying leading coordinates. Let $E$ select these $q\le\ell$ coordinates, and write $x_1\oplus v=a\oplus Ec$ with $a$ fixed. Compute $Z=\omega_1(L_1^{-1}E)$ and $z_0=\omega_1(L_1^{-1}a)$ using \cref{lem:gj-complexity}. The off-diagonal entries of $E^TB_{11}E$ are obtained modulo $2$ by applying $D$, solving with $L_1^T$, and restricting the result; its diagonal entries are $-Z_{:,j}^TDZ_{:,j}\pmod4$. The constant term is $-z_0^TDz_0$, and the linear term is $-2z_0^TDZc$. Thus the entire restricted form is constructed in $O(n\tau^{\omega-1}+\ell n\tau^{\omega-2})$ time without forming $B_{11}$. Evaluating the resulting $q$-variable quadratic form for all $k$ inputs costs $O(k\ell^{\omega-1})$ by dense batching. This gives the claimed bound. If $\ell=0$, all input strings coincide and a single amplitude evaluation suffices; its value is shared by all queries.
\end{proof}

Section \ref{sec:apps-sim} applies the simulation algorithms of this section to graph states, Clifford circuits and Clifford+T circuits.

\section{Applications in simulation}
\label{sec:apps-sim}
In this section, we apply our results to several simulation problems. In \cref{sec:gs-sim}, we show that our algorithm  solves the graph state simulation problem and matches the complexity of the state of art~\cite{gosset2024fast}. Then we apply our algorithm to Clifford circuit simulation in \cref{sec:clifford-sim} and Clifford+T circuit simulation in \cref{sec:clifford-t-sim} and compare the complexities with known techniques.

\subsection{Graph state simulation}\label{sec:gs-sim}
The graph state simulation problem~\cite{gosset2024fast} is of independent interest due to the wide application of graph states, which is defined as follows.
\begin{definition}[Graph state simulation~\cite{gosset2024fast,kerzner2021clifford}]\label{def:gs-sim}
Let $G=(V,E)$ be a simple graph and each vertex $i\in V$ is associated with a unitary $U_i\in\set{H, H\check{Z}, I}$. Let $U_\text{base}=\bigotimes_{i\in V}U_i$. The strong simulation of $G$ takes any $x\in\F_2^{\abs{V}}$ as input, and outputs
\[
\mel{x}{U_\text{base}}{G}.
\]
The weak simulation takes $S\subseteq V$ and $y\in\F_2^S$ as input, and samples $x\in\mathcal{X}_{S,y}$ from the distribution,
\[
p(x)\propto\abs{\mel{x}{U_\text{base}}{G}}^2.
\]
\end{definition}
Note that applying $H$, $H\check{Z}$, or $I$ gate to a state and then measuring it in computational basis can be viewed as measuring the state in Pauli X, Y or Z basis respectively.

The phased graph state simulation problem as defined in \cref{def:pgs-sim} is equivalent to \cref{def:gs-sim}, as the reduction in \cref{thm:gs-sim-complexity} shows, because all vertices measured in Z basis can be trimmed in linear time using the following property of graph states~\cite{raussendorf2003measurement}.

\begin{lemma}\label{lem:z-measure}
Let $G=(V,E)$ be a graph, $v\in V$ and $b\in\F_2$. Then
\begin{equation*}
\actson{\bra{b}}{v}
\ket{G}=\frac{1}{\sqrt2}\prod_{u\in N(v)}\left(\actson{\mat{Z}}{u}\right)^b\ket{G\setminus\set{v}}.
\end{equation*}
\end{lemma}
\begin{proof}
Assume WLOG that $v$ is the first vertex in $G$. Then by definition,
\begin{align*}
    \ket{G}&=\prod_{ij\in E}\actson{\mat{U}_{CZ}}{i,j}\ket{+}^{\otimes\abs{V}}\\
    &=\prod_{u\in N(v)}\actson{\mat{U}_{CZ}}{u,v}
    \prod_{ij\in E,i\neq v,j\neq v}\actson{\mat{U}_{CZ}}{i,j}
    \ket{+}^{\otimes\abs{V}}\\
    &=\prod_{u\in N(v)}\actson{\mat{U}_{CZ}}{u,v}\ket{+}\otimes\ket{G\setminus\set{v}},
\end{align*}
since $\actson{\mat{U}_{CZ}}{i,j}$ commutes with each other.
Note that $\ket{+}=\frac{1}{\sqrt2}(\ket{0}+\ket{1})$, and $\mat{U}_{CZ}$ does nothing if the first qubit is $\ket{0}$ and applies $\mat{Z}$ to the second qubit if the first qubit is $\ket{1}$. So, the equation above can be written as
\begin{align*}
    \ket{G}&=\frac{1}{\sqrt2}\ket{0}\otimes\ket{G\setminus\set{v}}+
    \frac{1}{\sqrt2}\prod_{u\in N(v)}\actson{\mat{Z}}{u}\ket{1}\otimes\ket{G\setminus\set{v}}.
\end{align*}
Applying $\actson{\bra{b}}{v}$ to this equation completes the proof.
\end{proof}

Then we have the following theorem that matches the $O(n\tau^{\omega-1})$ time complexity of single-sample weak graph state simulation in~\cite{gosset2024fast} and improves the $O(n\tau^2)$ time complexity of strong graph state simulation (with phase) in \cite{kerzner2021clifford}; for $k$ amplitudes our $O(n\tau^{\omega-1}+kn\tau^{\omega-2})$ bound is better than the $O(kn\tau^2)$ of repeating that method for all $k\ge1$, through the leading term at $k=1$ and the incremental $n\tau^{\omega-2}$ term for large $k$.
\begin{theorem}\label{thm:gs-sim-complexity}
Let $G$ be a graph with $n$ vertices and a tree decomposition of $G$ with width $\tau$ is given.
Then the strong/weak simulation of $\ket{G}$ for $k\in\N$ amplitudes/samples could be solved in $O(n\tau^{\omega-1}+kn\tau^{\omega-2})$ time.
\end{theorem}
\begin{proof}
By \cref{lem:z-measure}, each vertex $v$ measured in $\mat{Z}$ basis can be preprocessed efficiently, producing the phase and the bit flips of \eqref{eq:ztrim} below. In detail, suppose $G$ has $n$ vertices and $m$ edges, and $t$ vertices are measured in $Z$ basis. Permute the vertices so that these $t$ vertices come first, and partition
\[
\mat A=\begin{bmatrix}\mat A_{11}&\mat A_{12}\\\mat A_{21}&\mat A_{22}\end{bmatrix},
\qquad \mat x=(\mat x_1^T,\mat x_2^T)^T,
\qquad \mat A_{11}\in\F_2^{t\times t}.
\]
Then the reduction, constructed in $O(n+m)$ time, gives
\begin{align}
\mel{x}{U_\text{base}}{G}=2^{-t/2}(-1)^{\frac12\mat{x}_1^T\mat{A}_{11}\mat{x}_1}\mel{\mat{x}_2\oplus\mat{A}_{21}\mat{x}_1}{H^{\otimes(n-t)}\check{Z}^{\mat{d}}}{G'},\label{eq:ztrim}
\end{align}
where $\mat{A}$ is the adjacency matrix of $G$, $\mat{d}\in\Z_4^{n-t}$ is indexed by the remaining vertices with $d_u=1$ if $u$ is measured in the $\mat{Y}$ basis and $d_u=0$ if it is measured in the $\mat{X}$ basis, and $G'$ has adjacency matrix $\mat A_{22}$. This reduces amplitude evaluation to strong simulation of $\check Z^d\ket{G'}$. The sparse products and quadratic forms in \eqref{eq:ztrim} are evaluated in batches of at most $\tau$ inputs, with $O(\tau^\omega)$ work per bag, as in \cref{lem:sampling-complexity}. Together with \cref{thm:strong-sim-pgs}, this gives the stated strong-simulation bound.

For weak simulation, let $w'\oplus\spn(F)$ be the support of $H^{\otimes(n-t)}\check Z^d\ket{G'}$, where the columns of $F$ are independent columns of $\omega_1(A_{22}+D(d))$, as in \cref{eq:sparse-support-basis}. Equation \eqref{eq:ztrim} shows that the full output is uniform on the affine space with offset $\widehat w=(0^T,w'^T)^T$ and generator
\[
\widehat F=\begin{bmatrix}I_t&0\\A_{21}&F\end{bmatrix}.
\]
This generator has full column rank. Apply \cref{lem:affsamp} to $\widehat F$ with the prescribed outcomes shifted by $\widehat w$. Its nonzeros are original graph edges or diagonal entries, so the row-copy and column-copy construction in \cref{lem:sampling-complexity} gives a width-$O(\tau)$ decomposition of the symmetric dilation of any row restriction. The same sparse factorization and batched sampling argument therefore gives $O(n\tau^{\omega-1}+kn\tau^{\omega-2})$ time, including detection of an inconsistent conditioning event.

The reduction preserves the exact amplitudes, including their global phase: the scalar $\gamma$ of \cref{thm:amp-formula} is accumulated by the same elimination that produces the LDL factor, at no extra asymptotic cost (\cref{lem:gj-data}), so the strong simulation obtained this way is exact, including the global phase.
\end{proof}

The planar decomposition of~\cite{gosset2024fast} has $\|T_G\|_\omega^\omega=\tilde O(n^{\omega/2})$ with $n=\abs{V(G)}$~\cite[Lemma 31]{gosset2024fast}, and the $(d+4)$-coarse-graining of the graph-like diagram~\cite[Theorem 8 and its proof]{gosset2024fast} gives $\|T_A\|_\omega^\omega=\tilde O(n^{\omega/2}d^\omega)$. Adding a leaf copy of the bag at each unary node~\cite[Corollary 5.3]{solomonik2025fast} makes their nice decomposition~\cite[Definition 12]{gosset2024fast} a rooted full binary tree decomposition, changing $\|T\|_\omega^\omega$ by at most a constant factor, so it satisfies the hypothesis of~\cite[Theorem 5.9]{solomonik2025fast}. Since our cost is $O(\|T\|_\omega^\omega)$ on any rooted full binary decomposition~\cite[Theorem 5.9]{solomonik2025fast}, single-sample weak simulation matches their planar $\tilde O(n^{\omega/2}d^\omega)$ bound~\cite[Theorem 7]{gosset2024fast}.

\subsection{Clifford circuit simulation}\label{sec:clifford-sim}
In this section, we give the time complexities of our Clifford circuit simulation algorithm and compare it with related works.
Combining \cref{lem:circ-as-pgs} with the phased graph state simulation algorithms in \cref{sec:pgs-sim} gives the following theorem.
\begin{theorem}\label{thm:clifford-complexity}
Let $\mat{U}$ be a Clifford circuit with $n$ qubits and $m$ gates, each acting on at most 2 qubits and $m=\Omega(n)$, and let a tree decomposition of its graph-like diagram (\cref{lem:circ-as-pgs}) of width $\tau$ be given. Then strong/weak simulation of $\mat{U}$ for $k\in\N$ amplitudes/samples can be solved in
$O(m\tau^{\omega-1}+\min(km\tau^{\omega-2}, mn\tau^{\omega-2}+kn^{\omega-1}))$ time.
\end{theorem}
\begin{proof}
By \cref{lem:circ-as-pgs}, the circuit simulation problem can be reduced to a phased graph state simulation problem with $N=O(m+n)$ vertices. Further, a tree decomposition of the phased graph state of width at most $\tau$ can be obtained in linear time.
Also, note that the number of unfixed output bits of this phased graph state, i.e. $\ell$ as defined in \cref{thm:weak-sim-pgs}, is exactly $n$ because it corresponds to the output qubits of the circuit. Then \cref{thm:strong-sim-pgs} and \cref{thm:weak-sim-pgs} gives the desired complexities stated in the theorem.
The reduction preserves the exact amplitudes: \cref{lem:gj-data} records the scalar $\gamma$ of \cref{thm:amp-formula}, and \cref{lem:sampling-complexity} computes the fibre's affine-support dimension $\nu$ needed for the factor $\lambda$ in \cref{lem:circ-as-pgs}, within the same preprocessing bound.

\end{proof}

For any circuit as defined in \cref{thm:clifford-complexity}, a tree decomposition of width $O(n)$ can always be obtained by slicing the circuit in time: each bag contains the $n$ qubits at that point of the circuit together with the gates of the corresponding layer, every edge of the diagram is covered by a bag, and consecutive bags share the qubits, so the bags form a tree decomposition of width $O(n)$ (this holds for arbitrary two-qubit gates, including long-range ones). Then we have the following corollary.
\begin{corollary}\label{cor:general-clifford-sim-complexity}
Let $\mat{U}$ be a Clifford circuit with $n$ qubits and $m$ gates, each of which acts on at most 2 qubits, for $m=\Omega(n)$. Then strong/weak simulation of $\mat{U}$ for $k\in\N$ amplitudes/samples could be solved in $O(mn^{\omega-1}+kn^{\omega-1})$ time.
\end{corollary}
This improves on the quadratic form expansion method~\cite{de2022fast} in the regime recorded in \cref{tab:complexity-cmp}, namely when the coset sum in its amplitude evaluation is trivial, thanks to the use of fast matrix multiplication. Compared with the tableau method~\cite{aaronson2004improved,gosset2024fast}, whose cost is $O(mn+kn^{\omega})$, and Pauli-frame propagation as implemented in the Stim library~\cite{gidney2021stim}, whose estimated cost is $O(n^{\omega}+km)$, our bound is asymptotically smaller when $m\gg n^{\omega-1}$ and $k\gg m$. The two methods are therefore complementary: ours is preferable when many amplitudes or samples are required, while for a single amplitude the tableau method has the better gate-count dependence ($mn$ versus $mn^{\omega-1}$).
\cref{tab:complexity-cmp} compares the time complexity for strong and weak simulation in detail.

\begin{table}[!ht]
\centering
\small
\begin{tabular}{|c|c|c|}
    \hline
     Methods & Strong simulation & Weak simulation \\
     \hline
     Tableau \cite{aaronson2004improved,gosset2024fast} & $O(mn+kn^\omega)$ & $O(mn+kn^\omega)$ \\
     \hline
     Tableau with Pauli frame~\cite{gidney2021stim} & - & $O(n^\omega+km)$ \\
     \hline
     Quadratic form expansions~\cite{de2022fast} & $O(mn^2+kn^{\omega-1})$ & $O(mn^2+kn^2)$ \\
     \hline
     This work & $O(mn^{\omega-1}+kn^{\omega-1})$ & $O(mn^{\omega-1}+kn^{\omega-1})$ \\
     \hline
\end{tabular}
\caption{Time complexity of strong and weak simulation of Clifford circuits, with
$m$ gates, $n$ qubits, and $k$ amplitudes/samples. The Pauli-frame entry is our
own estimate: the cited paper reports per-gate update costs but no bound
of this form, and its simulator samples measurement outcomes rather than
computing amplitudes. The quadratic-form-expansion entries are taken in the
full-rank case $\rank(\mat A)=r$, as in~\cite{de2022fast}.
As~\cite{de2022fast} focuses mainly on weak simulation, the strong-simulation entry there is our own evaluation of \cref{eq:qfe-amplitudes}.
Sub-cubic fast matrix
multiplication is assumed where possible. This work improves the dependence on the number of amplitudes or
samples, where it achieves $kn^{\omega-1}$ against the tableau method's
$kn^{\omega}$, and the corresponding gate-count terms are $mn^{\omega-1}$ and
$mn$. The method is therefore preferable in the many-amplitude or many-sample
regime described in the text.}
\label{tab:complexity-cmp}
\end{table}

Applying the graph state simulation algorithms of \cite{gosset2024fast} and \cite{kerzner2021clifford} to the graph of our circuit-to-graph reduction in \cref{sec:circ-as-pgs}, with the width-$O(n)$ tree decomposition, gives $O(mn^2)$ for single-amplitude strong simulation and $O(mn^{\omega-1})$ for weak simulation, matching our weak bound, while our algorithm improves the strong-simulation bound.

\subsection{Clifford+T circuit simulation}
This section expands each T gate in terms of stabilizer states and then bounds the resulting simulation cost; \cref{thm:clifford-T} states the cost, and the discussion after it delimits the regime in which the shared factorisation of the Clifford block is retained.
\label{sec:clifford-t-sim}
General quantum circuits can be represented by Clifford gates and $\mat{T}$ gates, where $\mat{T}=\mat{R}_z(\pi/4)$.
Our techniques can also be used to simulate Clifford+T circuits by applying the stabilizer rank decomposition of non-Clifford states~\cite{bravyi2019simulation}. Such decompositions are defined as follows,
\[
\ket{T}^{\otimes t}=\left(
\vcenter{\hbox{\begin{tikzpicture}
    \node[znode] (v) at (0,0) {${\scriptstyle \pi/4}$};
    \coordinate (v0) at (0.5,0);
    \draw (v) -- (v0);
\end{tikzpicture}}}
\right)^{\otimes t}
=\sum_{j=1}^r \alpha_j\ket{\phi_j},
\]
where $\alpha_j\in\C$ are coefficients and $\ket{\phi_j}$ are Clifford states. The smallest possible $r$ is called the stabilizer rank of $\ket{T}^{\otimes t}$ and it is known that $r\le 2^{0.3963t}$~\cite{qassim2021improved}.

In particular, consider a quantum circuit $\mat{U}$ on $n$ qubits with $m$ Clifford gates and $t$ T gates.
By \cref{lem:circ-as-pgs}, we can represent $\mel{\mat{b}}{\mat{U}}{0^n}$ as a graph-like ZX diagram, and the phases of its nodes are multiples of $\pi/4$. For each node whose phase is an odd multiple of $\pi/4$, we apply the following equality:
\begin{equation}\label{eq:clifford-T-dummy}
\vcenter{\hbox{\begin{tikzpicture}
    \node[znode] (v) at (0,0) {${\scriptstyle \pi/4}$};
    \node (t) at (-0.6,0.1) {$\vdots$};
    \coordinate (v0) at (-0.75,0.5);
    \coordinate (v1) at (-0.75,-0.5);
    \draw (v) to[out=120,in=0] (v0) (v) to[out=240,in=0] (v1);
\end{tikzpicture}}}=
\vcenter{\hbox{\begin{tikzpicture}
    \node[znode] (u) at (2,0) {${\scriptstyle \pi/4}$};
    \node[hnode] (h1) at (1.5,0) {};
    \node[znode] (w) at (1,0) {};
    \node[hnode] (h2) at (0.5,0) {};
    \node[znode] (v) at (0,0) {};
    \node (t) at (-0.6,0.1) {$\vdots$};
    \coordinate (v0) at (-0.75,0.5);
    \coordinate (v1) at (-0.75,-0.5);
    \draw (v) to[out=120,in=0] (v0) (v) to[out=240,in=0] (v1);
    \draw (u) -- (h1) (h1) -- (w) (w) -- (h2) (h2) -- (v);
\end{tikzpicture}}}\quad,\qquad
\vcenter{\hbox{\begin{tikzpicture}
    \node[znode] (v) at (0,0) {${\scriptstyle 3\pi/4}$};
    \node (t) at (-0.6,0.1) {$\vdots$};
    \coordinate (v0) at (-0.75,0.5);
    \coordinate (v1) at (-0.75,-0.5);
    \draw (v) to[out=120,in=0] (v0) (v) to[out=240,in=0] (v1);
\end{tikzpicture}}}=
\vcenter{\hbox{\begin{tikzpicture}
    \node[znode] (u) at (2,0) {${\scriptstyle \pi/4}$};
    \node[hnode] (h1) at (1.5,0) {};
    \node[znode] (w) at (1,0) {};
    \node[hnode] (h2) at (0.5,0) {};
    \node[znode] (v) at (0,0) {${\scriptstyle \pi/2}$};
    \node (t) at (-0.6,0.1) {$\vdots$};
    \coordinate (v0) at (-0.75,0.5);
    \coordinate (v1) at (-0.75,-0.5);
    \draw (v) to[out=120,in=0] (v0) (v) to[out=240,in=0] (v1);
    \draw (u) -- (h1) (h1) -- (w) (w) -- (h2) (h2) -- (v);
\end{tikzpicture}}}\quad.
\end{equation}
In this way, the resulting ZX diagram can be divided into two parts: one contains $O(n+m)$ nodes with phases that are multiples of $\pi/2$, and the other contains the new nodes introduced by \cref{eq:clifford-T-dummy}, as depicted below.
\begin{equation}\label{eq:clifford-T-amplitude-diagram}
\mel{\mat{b}}{\mat{U}}{0^n}=
\vcenter{\hbox{\begin{tikzpicture}
    \draw (0,0) rectangle (1,2);
    \draw[dashed] (1.65,0) rectangle (3.65,2);
    \coordinate (v1) at (1,0.3) {};
    \coordinate (v2) at (1,1.7) {};
    \node[hnode] (h11) at (1.4,0.3) {};
    \node[hnode] (h12) at (2.6,0.3) {};
    \node[hnode] (h21) at (1.4,1.7) {};
    \node[hnode] (h22) at (2.6,1.7) {};
    \node[znode] (d1) at (2.0,0.3) {};
    \node[znode] (d2) at (2.0,1.7) {};
    \node (dots1) at (2.0,1.1) {$\vdots$};
    \node[znode] (z1) at (3.25,0.3) {${\scriptstyle \pi/4}$};
    \node[znode] (z2) at (3.25,1.7) {${\scriptstyle \pi/4}$};
    \node (dots2) at (3.25,1.1) {$\vdots$};
    \node (label1) at (0.5,1) {$\kappa$};
    \draw (v1) -- (h11) (h11) -- (d1) (d1) -- (h12) (h12) -- (z1);
    \draw (v2) -- (h21) (h21) -- (d2) (d2) -- (h22) (h22) -- (z2);
\end{tikzpicture}}}
\end{equation}
We note that the first part of the diagram, labeled by $\kappa$, inherits the graph structure of $\mat{U}$, while the second part the diagram represents $\ket{T}^{\otimes t}$ up to local Clifford gates.
In this way, non-Clifford circuits can be reduced to weighted sum of Clifford circuits.
\begin{equation}\label{eq:stabilizer-decomposition-zx}
\cref{eq:clifford-T-amplitude-diagram}
=\sum_{j=1}^r\quad
\alpha_j\left(\quad
\vcenter{\hbox{\begin{tikzpicture}
    \draw (0,0) rectangle (1,2);
    \draw (2,0) rectangle (3,2);
    \coordinate (v1) at (1,0.3) {};
    \coordinate (v2) at (1,1.7) {};
    \node[hnode] (h1) at (1.5,0.3) {};
    \node[hnode] (h2) at (1.5,1.7) {};
    \node (dots1) at (1.5,1.1) {$\vdots$};
    \coordinate (z1) at (2,0.3);
    \coordinate (z2) at (2,1.7);
    \node (label1) at (0.5,1) {$\kappa$};
    \node (label2) at (2.5,1) {$\phi_j$};
    \draw (v1) -- (h1) (h1) -- (z1);
    \draw (v2) -- (h2) (h2) -- (z2);
\end{tikzpicture}}}
\quad\right)
\end{equation}
Here, $\phi_j$ represents the ZX diagram representing the Clifford terms given by the stabilizer decomposition of $\ket{T}^{\otimes t}$.

Ignoring the edges between $\kappa$ and $\phi_j$, let the matrix representation of $\kappa$ and $\phi_j$ be $\mel{\mat{y}^{(\kappa)}}{H^{\otimes n}}{\mat{A}^{(\kappa)}}$ and
$\mel{\mat{y}^{(\phi_j)}}{H^{\otimes n}}{\mat{A}^{(\phi_j)}}$ respectively.
Then the matrix representation of each term in \cref{eq:stabilizer-decomposition-zx} can be written as $\mel{\mat{x}^{(j)}}{H^{\otimes n}}{\mat{A}^{(j)}}$, where
\[
    \mat{A}^{(j)}=\mqty[
        \mat{A}^{(\kappa)} & \mat{B}^{(j)}\\
        (\mat{B}^{(j)})^T & \mat{A}^{(\phi_j)}
    ],\quad
    \mat{x}^{(j)}=\mqty[
        \mat{y}^{(\kappa)}\\
        \mat{y}^{(\phi_j)}
    ],
\]
and $\mat{B}^{(j)}$ encodes $t'\le t$ edges between $\kappa$ and $\phi_j$, where $t'$ is the number of surviving nodes of odd-$\pi/4$ phase after the rewriting: distinct T gates can fuse into a single node, and their phases add, so two T gates can contribute a Clifford phase and no edge at all. Here $\mat{A}^{(\kappa)}$ is $\kappa\times\kappa$ with $\kappa=O(n+m)$, $\mat{A}^{(\phi_j)}$ is $O(t)\times O(t)$, and $\mat{B}^{(j)}$ is $\kappa\times O(t)$ with exactly one nonzero per surviving node, hence $\nnz(\mat{B}^{(j)})=t'$: every column of $\mat{B}^{(j)}$ is either $0$ or a unit vector, and the crossing set is the same for every term $j$. (The estimates below are stated with $t$; they remain valid upper bounds and are tight when no phases fuse.)
Let $\rho=\rank(\omega_1(\mat{A}^{(\kappa)}))$. We order the coordinates so that a nonsingular principal $\rho\times\rho$ block $\mat{A}^{(\kappa)}_{11}\in\W_\rho$ comes first, and use the same order for every term, so that its first $\rho$ pivots are shared. The remaining $\eta=\kappa-\rho$ coordinates are eliminated per term rather than shared. For $\eta\ge2$ the Gauss--Jordan process stalls at coordinate $\rho+1$ for every coupling; for $\eta=1$ the next pivot may straddle into the trailing block (it does so for some couplings and stalls for others). Peeling or an equivalent per-term elimination is therefore what completes each term (the remark following \cref{lem:gj-data}). A rank-deficient shared block already occurs for simple circuits: if the Clifford part is a graph state on an odd number of vertices, its adjacency matrix is alternating and therefore singular over $\F_2$, so $\eta\ge1$. If a normalization by dummy nodes is preferred, at least $\eta$ nodes are necessary, since appending $\mu$ nodes adds $\mu$ rows and columns while raising the rank over $\F_2$ by at most $2\mu$; the peeling route gives simulation without the sharing of the first $\rho$ pivots,
\begin{equation}
\vcenter{\hbox{\begin{tikzpicture}
    \node[znode] (v) at (0,0) {${\scriptstyle \beta}$};
    \node (t) at (-0.6,0.1) {$\vdots$};
    \coordinate (v0) at (-0.75,0.5);
    \coordinate (v1) at (-0.75,-0.5);
    \coordinate (v2) at (0.5,0);
    \draw (v) to[out=120,in=0] (v0) (v) to[out=240,in=0] (v1) (v) -- (v2);
\end{tikzpicture}}}=
\vcenter{\hbox{\begin{tikzpicture}
    \node[znode] (u) at (2,0) {${\scriptstyle\pi/2 }$};
    \node[hnode] (h1) at (1.5,0) {};
    \node[znode] (w) at (1,0) {};
    \node[hnode] (h2) at (0.5,0) {};
    \node[znode] (v) at (0,0) {${\scriptstyle\alpha }$};
    \node (t) at (-0.6,0.1) {$\vdots$};
    \coordinate (v0) at (-0.75,0.5);
    \coordinate (v1) at (-0.75,-0.5);
    \coordinate (v2) at (2.5,0);
    \draw (v) to[out=120,in=0] (v0) (v) to[out=240,in=0] (v1);
    \draw (u) -- (h1) (h1) -- (w) (w) -- (h2) (h2) -- (v) (u) -- (v2);
\end{tikzpicture}}},\;\where\;\beta=\alpha+\pi/2.
\end{equation}

Therefore, \cref{eq:stabilizer-decomposition-zx} could be evaluated by solving strong simulation problem for each $\ket{\mat{A}^{(j)}}$. Note that each $\mat{A}^{(\phi_j)}$ has $O(t)$ rows/columns because $\ket{\phi_j}$ can be represented as a graph state with local Clifford gates~\cite{van2004graphical}. Since $\mat{A}^{(\kappa)}$ is shared among all $r$ terms, its LDL decomposition is computed once. For each $j$, let $R_j$ consist of the first $\rho$ rows of $\mat{B}^{(j)}$, the coupling to the shared nonsingular block, and let $F=\Gamma_\rho(\mat{A}^{(\kappa)}_{11})$ be its Gauss--Jordan output. The phase-aware Schur complement is
\[
\widetilde{\mat{A}}^{(\phi_j)}=\Omega_{1,2}\bigl(\mat{A}^{(\phi_j)}+R_j^TFR_j\bigr).
\]
This is \eqref{eq:gjl3}; its binary part is the usual Schur complement. To form it, solve with the shared $L_1$ for the $t'$ nonzero columns of $R_j$. The second diagonal bits follow from \eqref{eq:exp-implicit}, and the off-diagonal entries modulo $2$ follow by the corresponding transpose solve and restriction. By \cref{lem:gj-complexity} this costs $O(t\rho\tau)$ for separate right-hand sides or $O(\rho\tau^{\omega-1}+t\rho\tau^{\omega-2})$ in blocks of $\tau$ columns. In this construction every nonzero column of $R_j$ is a unit vector (one per surviving node of odd-$\pi/4$ phase), so forming the restricted entries and diagonals adds no larger term. Eliminating the $\eta$ peeled coordinates and factorizing the $O(t)$-sized trailing block together cost $O((\eta+t)^\omega)$. Each term therefore costs $O(\min(t\rho\tau,\rho\tau^{\omega-1}+t\rho\tau^{\omega-2})+(\eta+t)^\omega)$, which is the bound stated in \cref{thm:clifford-T}.
The global scalar $\gamma(\mat{A}^{(j)}_{11})$ of \cref{thm:amp-formula} splits accordingly: although each Gauss--Jordan update acts on all remaining graph coordinates, its restriction to the shared block evolves exactly as in the elimination of $\mat{A}^{(\kappa)}_{11}$ alone while $i\le\rho$ and its pivot type is independent of the term; hence the corresponding $1\times1$ pivots and their second bits are shared, and $\gamma$ factorizes into a shared factor and a per-term factor from eliminating the remaining $\eta+O(t)$ coordinates together, so the exact phase of every term is available without changing the asymptotic complexity. So the overall time complexity is as stated in \cref{thm:clifford-T} below.
This is summarized in the following theorem.
\begin{theorem}\label{thm:clifford-T}
Let $\mat{U}$ be a quantum circuit on $n$ qubits represented by $m$ Clifford gates and $t$ T gates, let $r$ be the stabilizer rank of $\ket{T}^{\otimes t}$ (more generally, let a decomposition with $r$ terms be supplied, of which $t'\le t$ carry odd-$\pi/4$ phases), let $\mat{A}^{(\kappa)}$ be the $\kappa\times\kappa$ matrix of the phased graph state representing the Clifford part of $\mat{U}$ (\cref{lem:circ-as-pgs}) with a tree decomposition of width $\tau$ (in the normalization of \cref{sec:pgs-sim}); let $\rho$ and $\eta=\kappa-\rho$ be as above.
Then the amplitude $\mel{\mat{b}}{\mat{U}}{0^n}$ for arbitrary binary string $\mat{b}$ can be evaluated in
$O\big(T_{LDL}(\mat{A}^{(\kappa)}_{11}) + r \cdot \big(\min(t\,\rho\,\tau,\rho\tau^{\omega-1}+t\,\rho\tau^{\omega-2}) + (\eta+t)^\omega\big)\big)$ time, where $T_{LDL}(\mat{A}^{(\kappa)}_{11})=O(\rho\tau^{\omega-1})$ is the time complexity of computing the LDL decomposition of the shared leading block by \cref{lem:gj-complexity}, the two alternatives correspond to solving the $t'$ right-hand sides one at a time or in blocks of $\tau$ columns, and the $(\eta+t)^\omega$ term covers the per-term elimination of the $\eta$ peeled coordinates together with the trailing block.
\end{theorem}

This algorithm achieves better performance than naively simulating $r$ Clifford circuits.
In comparison, the previous work~\cite{bravyi2016improved} approximates $\abs{\mel{\mat{b}}{\mat{U}}{0^n}}^2$ with relative error $\epsilon$ and failure probability $p_f$ in $O(\text{poly}(n,m)+rt^3\epsilon^{-2}\log(p_f^{-1}))$ time, where the $2^{\Theta(t)}$ factor is the stabilizer-rank bound used by both approaches and the dependence on the failure probability is logarithmic. Our algorithm utilizes fast matrix multiplication, runs deterministically, and also computes the exact amplitude including its global phase: for each of the $r$ terms the eighth root of unity $\gamma$ of \cref{thm:amp-formula} is accumulated from the shared elimination, as explained above, and the real factor $\lambda$ of \cref{lem:circ-as-pgs} is likewise recovered, distinct from $\gamma$.
Further, for $n$-qubit planar circuits of depth $d$ and $t$ T gates as defined in \cite{gosset2024fast,kerzner2021clifford}\footnote{Planar circuits are circuits where two-qubit gates are applied on edges of a planar graph formed by qubits.}, the treewidth is bounded by $d\sqrt{n}$.
Hence, our algorithm runs in time $O\big(n^{(\omega+1)/2}d^{\omega} + r\big(\min(t\,n^{3/2}d^{2},\,n^{(\omega+1)/2}d^{\omega}+t\,n^{\omega/2}d^{\omega-1})+(\eta+t)^\omega\big)\big)$,
whose solve contribution per term is $O(t\,n^{3/2}d^{2})$ when $t=O(1)$, since the minimum then selects the first branch; the saving from sharing the factorization of $\mat{A}^{(\kappa)}$ is therefore largest when $t$ is small relative to $\tau$.
When the rank deficiency of the shared block is small, so that $(\eta+t)^\omega$ is dominated by the other per-term terms, this per-term bound improves on \cite[Theorem 33]{kerzner2021clifford}. For a Clifford part with large deficiency the $(\eta+t)^\omega$ term can dominate, and since $\eta$ can be $\Theta(\kappa)$ independently of $t$, the improvement is then not guaranteed.
For very shallow circuits the one-time term can dominate; in the regime described above this improves on the time complexity $O(r(n^{3/2}t^6d^3))$ achieved in \cite[Theorem 33]{kerzner2021clifford}.
Finally, we note that this strong simulation algorithm can also be used for weak simulation using the gate-by-gate sampling method~\cite{bravyi2022simulate}. \Cref{sec:apps-other} turns to further applications of the amplitude formula.

\section{Other applications}
\label{sec:apps-other}
In this section, we discuss other applications of our results. In \cref{sec:lc-eq}, we show that our result allows an alternative, linear-algebraic characterization of locally Clifford (LC) equivalent graph states, which leads to a linear bound of the diameter of each graph state orbit~\cite{adcock2020mapping}. In \cref{sec:gs-learn}, we present an efficient protocol to learn graph states whose adjacency matrices are low rank.

\subsection{Local equivalence of graph states}\label{sec:lc-eq}

Let $\mathbb{Y}_2$ be the set of the 2-by-2 Clifford unitaries.
Two quantum states $\ket{\psi}$ and $\ket{\phi}$ are locally Clifford equivalent (LC-equivalent) iff
\[\exists U_1,\ldots, U_n\in \mathbb{Y}_2, \text{ such that } \ket{\psi} = \bigg(\bigotimes_{i=1}^n U_i\bigg)\ket{\phi}.\]

\begin{theorem}
\label{thm:lc-eq}
Two $n$-qubit graph states are LC-equivalent iff, for  some consistent ordering of vertices, their adjacency matrices $A$ and $B$ satisfy
\[\exists u, v\in\F_2^n, k\in\{0,\ldots, n\}, \text{ such that }
A+ D(u) = \omega_1(\Gamma_k(B+D(v))),\]
where $k$ is admissible, i.e. the leading $k\times k$ block of $B+D(v)$ is nonsingular and admits an unpivoted LDL; the vertex order is chosen so that $B+D(v)\in\W_n$ as well.
\end{theorem}
\begin{proof}
For brevity, we use the notation $\ket{\psi} \propto \ket{\phi}$ to mean $\exists \alpha \in\mathbb{C}$, $|\alpha|=1$ such that $\ket{\psi}=\alpha\ket\phi$.

Any one-qubit Clifford unitary may be represented, for some $a,b,c\in\{0,1,2,3\}$, as
\[U \propto \check{Z}^a\check{X}^b\hat{Z}^c,\]
which is sometimes referred to as the Euler decomposition of the gate.
We provide a proof that the Clifford Euler decomposition always exists in \cref{lem:one-qubit-clifford} in \cref{sec:one-qubit-clifford}.
We further decompose $\check{X}^b = X^{\omega_2(b)}(H\check{Z}H)^{\omega_1(b)}$. Using $H\check{Z}H\propto\hat{Z}H\hat{Z}$ and $X\check{Z}\propto\hat{Z}X$, we obtain the canonical form $a',b'\in\F_2$, $c',d'\in\{0,1,2,3\}$,
\[U \propto \check{Z}^{c'}X^{a'}H^{b'}\hat{Z}^{d'}.\]

Hence, for a local Clifford transformation given by one-qubit Clifford unitaries $U_1,\ldots,U_n$, we may find $\hat{u},\bar{v}\in\{0,1,2,3\}^n$, $p,q\in\F_2^n$, such that
\begin{align}
\bigg(\bigotimes_{i=1}^n U_i\bigg)\ket{B} \propto \check{Z}^{\hat{u}} X^qH^p\ket{B+D(\bar{v})}.
\end{align}
Let $S=\{i\in[n]:p_i=1\}$ and $s=|S|$. Apply the same permutation to the two graphs so that $S=\{1,\ldots,s\}$, and write
\[
M=B+D(\bar v)=\begin{bmatrix}M_{11}&M_{12}\\M_{21}&M_{22}\end{bmatrix},
\qquad M_{11}\in\J_s,
\]
with diagonal entries reduced modulo $4$.
Since $\ket{A}$ has all $2^n$ amplitudes of magnitude $2^{-n/2}$ and $\check Z^{\hat u}X^q$ only permutes amplitudes and multiplies them by phases, $H^p\ket{M}\propto X^q\check Z^{-\hat u}\ket{A}$ has no zero amplitudes.
For the amplitudes whose trailing $n-s$ bits are zero, \cref{lem:phased-graph-state-amplitudes} gives, for every $a\in\F_2^s$,
\[
\mel{a,0^{n-s}}{H^{\otimes s}\otimes I^{\otimes(n-s)}}{M}
=2^{-(n-s)/2}\mel{a}{H^{\otimes s}}{M_{11}}.
\]
By \cref{thm:amp-formula}, after a symmetric permutation of the first $s$ coordinates, the right-hand side has a zero amplitude if $\omega_1(M_{11})$ is singular. Hence $\omega_1(M_{11})$ is nonsingular.
Order the first $s$ vertices so that this block admits an unpivoted LDL, and order the remaining vertices using an LDL decomposition of its Schur complement. With this order $M\in\W_n$ and its leading $s\times s$ block is nonsingular, so $s$ is admissible.
Then, by \cref{lem:gj-pgs}, for some $x\in\F_2^n$ with only the first $s$ entries potentially nonzero,
\[
(H^{\otimes s}\otimes I^{\otimes(n-s)})\ket{M}
\propto X^x\ket{\Gamma_s(M)}.
\]

For any phased graph state $\ket{R}$ and $j\in\{1,\ldots,n\}$, we also have
\begin{align}
X^{\langle j\rangle}\ket{R}
&=2^{-n/2}\sum_{t\in\F_2^n}i^{-(t-e_j)^TR(t-e_j)}\ket{t}\notag\\
&\propto2^{-n/2}\sum_{t\in\F_2^n}i^{-t^T(R+2D(r_j))t}\ket{t}
=\ket{R+2D(r_j)},\label{eq:applyX}
\end{align}
where $r_j$ is the $j$th column of $R$ and the diagonal addition is modulo $4$.
Hence there exists $\bar u\in\{0,1,2,3\}^n$ such that
\[
\ket{\Gamma_s(M)}\propto X^{q\oplus x}\check Z^{-\hat u}\ket{A}
\propto\ket{A+D(\bar u)}.
\]
The amplitudes at $0$ fix the proportionality scalar, and those at $e_j$ and $e_j+e_l$ determine the diagonal modulo $4$ and the off-diagonal entries modulo $2$. Thus $A+D(\bar u)=\Gamma_s(M)$ in $\J_n$.
Taking first bits gives the relation in the theorem with $u=\omega_1(\bar u)$ and $v=\omega_1(\bar v)$, since \eqref{eq:gjl1} shows that the first bits of $\Gamma_s(M)$ depend only on $\omega_1(M)$.

Conversely, if $A+D(u)=\omega_1(\Gamma_k(B+D(v)))$ for an admissible $k$, then $\ket{A+D(u)}$ differs from $\ket{\Gamma_k(B+D(v))}$ only by diagonal Pauli gates. By \cref{lem:gj-pgs}, the latter is obtained from $\ket{B+D(v)}$ by Hadamard and Pauli gates. Since adding $D(u)$ or $D(v)$ applies only local phase gates, $\ket{A}$ and $\ket{B}$ are LC-equivalent.
\end{proof}
The vertex complementations generate a group acting on the set of $n$-qubit graph states, whose orbits are the classes of LC-equivalent graph states, and the diameter of an orbit is the maximum, over pairs of graph states in that orbit, of the minimum number of vertex complementations needed to transform one into the other. It has been heuristically observed that the diameter of a graph state orbit is strongly correlated to entanglement measure (i.e. Schmidt measure) of graph states~\cite{adcock2020mapping}.
It has been proved that a general upper bound of diameter for all graph state orbits is $\floor{3n/2}$~\cite[Proposition 6]{claudet2025local}. The same result can be recovered as a consequence of \cref{thm:lc-eq}, as shown in the theorem below.
Note that we consider graph states with labeled vertices, which means isomorphic graphs are considered different.
\begin{theorem}
The LC-equivalence diameter of any $n$-qubit graph state orbit with labeled vertices is at most $\floor{3n/2}$.
\end{theorem}
\begin{proof}
Consider any two LC-equivalent graph states with adjacency matrices $A$ and $B$. By \cref{thm:lc-eq}, there exist some $u,v\in\F_2^n$ and $k\in\set{0,...,n}$ such that in some consistent ordering of vertices,
\begin{align}
\ket{A+D(u)}=\ket{\omega_1(\Gamma_k(B+D(v)))}.
\end{align}
The off-diagonals of $A+D(u)$ and $\omega_1(\Gamma_k(B+D(v)))$ must match, so we have
\[
A=O(\Gamma_k(B+D(v))).
\]
Further, by definition of Gauss--Jordan process, the RHS can be viewed as the resulting graph after applying a sequence of vertex and edge complementations on the graph represented by $B$, with no pivot vertex reused. If there are $a$ vertex and $b$ edge complementations, then $a+2b=k\le n$. Each edge complementation can be represented by three vertex complementations (\cref{thm:edge-complementation}), giving at most $a+3b\le\floor{3k/2}\le\floor{3n/2}$ vertex complementations.
\end{proof}

\subsection{Graph state learning}\label{sec:gs-learn}
Quantum state learning is of great interest in order to understand complexity-theoretic properties of quantum states and experimentally characterize real-world quantum states~\cite{anshu2024survey}.
In the setting of this problem, we are provided with a number of identical copies of some unknown quantum states $\ket{\psi}$, on which we can apply arbitrary circuits/measurements,
and the goal is to obtain a mathematical representation of $\ket{\psi}$ using as few copies as possible. In this section, we consider the problem of learning graph states and provide an algorithm to learn the graph state efficiently if its adjacency matrix is low rank.

We take advantage of the fact that the measurement results of an $n$-qubit stabilizer state forms an affine subspace of $\F_2^n$, so learning the subspace could reduce the original state to a stabilizer state with fewer qubits~\cite[Theorem 3]{anshu2024survey}. Although the measurement results of a graph state form the full space of $\F_2^n$ (i.e. the amplitudes vector of any graph state has no zeros), our results for strong simulation of phased graph states allows similar ideas to be applied on graph states. This is detailed in the following lemma.
\begin{lemma}\label{lem:learn-low-rank-gs-reduction}
Given an unknown $n$-qubit graph state $\ket{G}$, then one can compute a nonnegative integer $r$ and a circuit $U$ composed of CNOT and X gates using $O((r+1)\log\frac{1}{\delta})$ one-copy measurements of $\ket{G}$ such that the following is satisfied with probability at least $1-\delta$,
\begin{itemize}
\item $r$ is exactly the rank of adjacency matrix of $G$;
\item $UH^{\otimes n}\ket{G}=\gamma'\ket{\phi}\otimes\ket{0^{n-r}}$ for some $r$-qubit stabilizer state $\ket{\phi}$ and some global phase $\gamma'\in\C$ with $\abs{\gamma'}=1$.
\end{itemize}
\end{lemma}
\begin{proof}
Applying \cref{thm:amp-formula} to the phased graph state produced by \cref{lem:circ-as-pgs}, there exist $Q\in\J_r$, $A\in\F_2^{n\times r}$, $b\in\F_2^n$ and a global phase $\gamma'\in\C$ with $\abs{\gamma'}=1$ such that
\begin{align}
H^{\otimes n}\ket{G}=\gamma'\,2^{-r/2}\sum_{x\in\F_2^r} i^{x^TQx}\ket{\omega_1(Ax+b)}.\label{eq:learn-low-rank-gs-1}
\end{align}
Measuring \cref{eq:learn-low-rank-gs-1} in computational basis outputs $\omega_1(Ax+b)$ for uniformly random $x\in\F_2^r$, so $k+1$ measurements yield $k+1$ vectors uniform on the coset $\mat{b}+\spn(A)$.

Consider the algorithm that iteratively samples vectors $v^{(1)},\dots,v^{(k+1)}$ from the coset $\mat{b}+\spn(\mat{A})$ and uses the $k$ differences $v^{(i)}\oplus v^{(1)}$ for $i\ge2$, which are independent and uniform on $\spn(\mat{A})$. The offset $\mat{b}$ is never needed: the circuit the proof builds depends on it only through $\omega_1(\mat{A}_2\mat{A}_1^{\#}(y+\mat{b}_1)+\mat{b}_2)$, which is invariant under adding a vector of $\spn(\mat{A})$ to $\mat{b}$. It constructs $\mat{A}^{(k)}=\begin{bmatrix}
v^{(2)}\oplus v^{(1)} & v^{(3)}\oplus v^{(1)} & ... & v^{(k+1)}\oplus v^{(1)}
\end{bmatrix}$
until $\rank(A^{(k)})=\rank(A^{(k+s)})$ for $s=\ceil{\log_2\frac{2}{\delta}}$, where $r$ is the rank of $A$.
Each phase that does not stop the algorithm ends within its first $s$ differences, because the rank increases inside the window; hence the number of samples is at most $rs+s+1=O((r+1)\log\frac{1}{\delta})$.
It remains to bound the failure probability. The rank increases by one exactly when a difference falls outside the span of the previous differences, so the algorithm passes through the phases $j=0,\dots,r-1$ in order, and in phase $j$ a difference falls inside the span with probability $p_j=2^{j-r}$, independently of the past.
It can stop in phase $j$ only if the first $s$ differences of that phase all fall inside the current span, an event of probability $p_j^s=2^{-(r-j)s}$. Letting $F$ denote the event that the algorithm stops with $\rank(A^{(k)})<r$, we therefore have
\begin{align}
p(F)\le\sum_{j=0}^{r-1}2^{-(r-j)s}=\sum_{t=1}^{r}2^{-ts}\le\frac{2^{-s}}{1-2^{-s}}\le 2^{1-s}\le\delta,
\end{align}
where the last inequality holds for $s\ge\log_2\frac{2}{\delta}$; the contributions of the phases decay geometrically, so no factor of $r$ is needed.

Next, we construct the circuit $U$.
Assume WLOG that $A=\begin{bmatrix}A_1\\A_2\end{bmatrix}$ for invertible $A_1\in\F_2^{r\times r}$. Then
\begin{align}
\cref{eq:learn-low-rank-gs-1}
&=\gamma'\,2^{-r/2}\sum_{x\in\F_2^r} i^{x^TQx}\ket{\omega_1(A_1x+b_1)}\otimes\ket{\omega_1(A_2x+b_2)}\\
&=\gamma'\,2^{-r/2}\sum_{x\in\F_2^r} i^{x^TQx}U\ket{\omega_1(A_1x+b_1)}\otimes\ket{0^{n-r}},\label{eq:learn-low-rank-gs-2}
\end{align}
where $U$ is a circuit that computes $\omega_1(A_2x+b_2)$ in the last $n-r$ qubits using $\omega_1(A_1x+b_1)$ from the first $r$ qubits. In particular, for any $y\in\F_2^r$, $z\in\F_2^{n-r}$,
\begin{align}
U\ket{y}\otimes\ket{z}=\ket{y}\otimes\ket{z\oplus\omega_1(A_2A_1^\#(y+b_1)+b_2)}.
\end{align}
Note that $U$ can be expressed using only CNOT and X gates and $U=U^{-1}$.
Also, $U$ commutes with $i^{x^TQx}$, so \cref{eq:learn-low-rank-gs-2} gives $UH^{\otimes n}\ket{G}=\gamma'\ket{\phi}\otimes\ket{0^{n-r}}$, where
\begin{align}
\ket{\phi}=2^{-r/2}\sum_{x\in\F_2^r} i^{x^TQx}\ket{\omega_1(A_1x+b_1)}.
\end{align}
Note that $\ket\phi$ must be a stabilizer state because $H$ and $U$ are Clifford gates.
\end{proof}

Then we have the following theorem to learn graph states with low-rank adjacency matrices efficiently.
\begin{theorem}
\label{thm:learn-low-rank-gs}
An unknown graph state whose adjacency matrix is of rank $r$ can be learned using $O(r^2+(r+1)\log\frac{1}{\delta})$ one-copy measurements, or $O((r+1)\log\frac{1}{\delta})$ one-copy measurements together with $O(r)$ two-copy measurements, with probability at least $1-\delta$.
\end{theorem}
\begin{proof}
Apply \cref{lem:learn-low-rank-gs-reduction} with failure probability $\delta/2$. If $r=0$, no further measurements are needed; assume henceforth that $r\ge1$. The resulting $r$-qubit state $\ket\phi$ has full support by \eqref{eq:learn-low-rank-gs-2}, so it is a phased graph state. For each qubit, measure the other qubits in the $Z$ basis and that qubit in either the $X$ or $Y$ basis. In one of these two bases its outcome is an affine function of the other outcomes; in the other basis it is an independent uniform bit. With $h=r+\lceil\log_2(8r/\delta)\rceil$ samples per qubit and basis, failure to determine the affine function in the correct basis, or consistency with an affine function in the wrong basis, has total probability at most $2r\,2^{r-h}\le\delta/4$. The recovered functions and basis choices determine the phased graph state. This uses $O(r^2+r\log(1/\delta))$ one-copy measurements.

For the second alternative, if $\delta\ge2^{-r}$, use Bell sampling~\cite{montanaro2017learning} with $3r+2$ differences. Its failure probability is at most $2^{r-(3r+2)}\le\delta/4$, and it uses $O(r)$ two-copy measurements and $O(r)$ one-copy measurements to determine stabilizer signs. If $\delta<2^{-r}$, the preceding one-copy procedure already costs $O(r\log(1/\delta))$, so no two-copy measurements are required.
The one-copy count is of the same order as the sample-optimal tomography of stabilizer states, which uses $O(r^2)$ one-copy measurements~\cite{haah2017sample}.
Combining either procedure with the reduction and summing their failure probabilities proves the theorem.
\end{proof}

\section{Conclusion}
In this work, we present reductions from several quantum circuit simulation problems, including graph state simulation, Clifford circuit simulation, and Clifford+T circuit simulation, to the standard linear algebra problem known as LDL decomposition. In particular, we represent the problem as a modified adjacency matrix of some graph that inherits the structure of the original problem, such that the LDL decomposition of the matrix gives the simulation result. This relation is shown via an explicit formula for amplitudes of Hadamard-transformed phased graph state in \cref{thm:amp-formula}.
This result extends the application of standard linear algebra algorithms to the field of quantum circuit simulation. Based on this idea, we apply the sparse LDL algorithm of~\cite{solomonik2025fast} and derive algorithms for various simulation problems. We analyze the complexities of our algorithms and show improvements over the state of the art in certain cases.
Moreover, we show that our amplitudes formula has wider applications beyond simulation. In particular, we give a new characterization of LC-equivalent graph state based on the Gauss--Jordan process and discuss its implications for bounding the diameter of each graph state orbit. Also, we design a protocol to efficiently learn graph states with low-rank adjacency matrices based on the amplitudes formula.
We believe that our work provides a new perspective on Clifford circuit simulation and opens up new possibilities for applying existing linear algebra techniques to quantum circuit simulation.

\section*{Acknowledgments}
We are grateful to Ryan Mann for helpful comments on an earlier version of the manuscript.
Part of this work was conducted while one of the authors (E.S.) was visiting the Simons Institute for the Theory of Computing.
This work was supported by the National Science Foundation QLCI–HQAN program under Award No. 2016136.

The authors used ChatGPT (OpenAI) and DeepSeek V4.1 Flash to assist in checking and correcting the derivations of the amplitude formula, conditional-sampling algorithm, and local Clifford equivalence characterization, and in generating code for finite-case verification. The authors assume responsibility for all content.

\appendix
\section{Reducing stabilizer states to graph states}
\label{app:stab-to-graph}
It is widely known that stabilizer states can always be represented as a graph state with local Clifford gates~\cite{van2004graphical}. In this section, we list some useful operations on tableaux in \cref{alg:tableau-operations}. Based on these operations, we provide \cref{alg:stab-to-graph}, which converts a stabilizer state in tableau representation into an XY-basis graph state up to bit flips. It takes as input the lower $n$ rows of a tableau representation of a stabilizer state, and outputs a graph and three layers of Clifford gates, namely $\mat{S}^\dagger$, $\mat{H}$ and $\mat{Z}$.

\begin{algorithm}[h]
\caption{Basic tableau operations}
\label{alg:tableau-operations}
\begin{algorithmic}
\Procedure{row\_add}{$\mat{\tilde{Z}},\mat{\tilde{X}},\mat{r},i,j$}\Comment{add j'th row to i'th row}
    \For{k=1,\dots,n}
        \State $\tilde{z}_{ik}:=(\tilde{z}_{ik}+\tilde{z}_{jk})\bmod2$, $\tilde{x}_{ik}:=(\tilde{x}_{ik}+\tilde{x}_{jk})\bmod2$
    \EndFor
    \State $r_i:=\left(r_i+r_j+\floor{\left(\sum_{k=1}^{n}f(g(\tilde{z}_{ik},\tilde{x}_{ik}),g(\tilde{z}_{jk},\tilde{x}_{jk})\right)/2}\right)\bmod2$, where
    \begin{align*}
        g(z,x)&=x+3z-2zx\\
        f(p_1,p_2)
        &=\begin{cases}
        0 & p_1=0 \text{ or } p_2=0 \text{ or } p_1=p_2\\
        2((p_2-p_1) \bmod 3)-1& \otherwise
        \end{cases}
    \end{align*}
    \Comment{$g(z,x)$ represents $\mat{I},\mat{X},\mat{Y},\mat{Z}$ as $0,1,2,3$ respectively}
\EndProcedure

\Procedure{apply\_H}{$\mat{\tilde{Z}},\mat{\tilde{X}},\mat{r},i$}

    \State Swap the $i$'th column of $\mat{\tilde{Z}}$ and the $i$'th column of $\mat{\tilde{X}}$
    \For{j=1,\dots,n}
        \State $r_j:=(r_j+\tilde{z}_{ji}\tilde{x}_{ji})\bmod2$
    \EndFor
\EndProcedure

\Procedure{apply\_S}{$\mat{\tilde{Z}},\mat{\tilde{X}},\mat{r},i$}
    \For{j=1,\dots,n}
        \State $r_j:=(r_j+\tilde{z}_{ji}\tilde{x}_{ji})\bmod2$
        \State $\tilde{z}_{ji}:=(\tilde{z}_{ji}+\tilde{x}_{ji})\bmod2$
    \EndFor
\EndProcedure

\Procedure{apply\_CZ}{$\mat{\tilde{Z}},\mat{\tilde{X}},\mat{r},a,b$}
    \For{j=1,\dots,n}
        \State $r_j:=(r_j+\tilde{x}_{ja}\tilde{x}_{jb}(\tilde{z}_{ja}+\tilde{z}_{jb}))\bmod2$
        \State $\tilde{z}_{ja}:=(\tilde{z}_{ja}+\tilde{x}_{jb})\bmod2$, $\tilde{z}_{jb}:=(\tilde{z}_{jb}+\tilde{x}_{ja})\bmod2$
    \EndFor
\EndProcedure
\end{algorithmic}
\end{algorithm}

\begin{algorithm}[h]
\caption{$[G,\mat{a},\mat{b}, \mat{c}]=$ stabilizers\_to\_graph($\mat{\tilde{Z}},\mat{\tilde{X}},\mat{r}$)}\label{alg:stab-to-graph}
\begin{algorithmic}
\Require $\mat{\tilde{Z}},\mat{\tilde{X}}\in\binaryset^{n\times n}$, $\mat{r}\in\binaryset^n$, $(\mat{\tilde{Z}},\mat{\tilde{X}},\mat{r})$ represents stabilizers of $\ket{\psi}$
\Ensure $\mat{a},\mat{b},\mat{c}\in\binaryset^n$ such that, up to a global phase not determined by the signed tableau, $\ket{\psi}
=\left(\bigotimes_{j=1}^n \mat{H}^{a_j}(\mat{S}^\dagger)^{b_j}\mat{Z}^{c_j}\right)\ket{G}
$

\State Initialize $\mat{a},\mat{b}:=0$, $h:=0$, and $\mat{p}:=(1,\dots,n)$

\For{$i=1,\dots,n$}\Comment{forward elimination}
    \If{$\mat{\tilde{X}[i,i:]}$ are all zeros}
        \State Find the index $k$ of the first nonzero in $\mat{\tilde{Z}}[i,i:]$
        \State \Call{apply\_H}{$\mat{\tilde{Z}},\mat{\tilde{X}},\mat{r},i+k-1$}
        \State $h:=1$
    \Else
        \State Find the index $k$ of the first nonzero in $\mat{\tilde{X}}[i,i:]$
    \EndIf
    \State Swap the $i$'th and $(i+k-1)$'th entries of $\mat{a},\mat{b},\mat{p}$
    \State Swap the $i$'th and $(i+k-1)$'th columns of $\mat{\tilde{Z}},\mat{\tilde{X}}$
    \If{$h=1$}\Comment{Hadamard now in column $i$}
        \State $a_i:=a_i+1$, $h:=0$
    \EndIf
    \For{$j=i+1,\dots,n$}
        \If{$\tilde{x}_{ji}\neq0$}
            \State \Call{row\_add}{$\mat{\tilde{Z}},\mat{\tilde{X}},\mat{r},j,i$}
        \EndIf
    \EndFor
\EndFor

\For{$i=n,\dots,1$}\Comment{backward substitution}
    \For{$j=1,\dots,i-1$}
        \If{$\tilde{x}_{ji}\neq0$}
            \State \Call{row\_add}{$\mat{\tilde{Z}},\mat{\tilde{X}},\mat{r},j,i$}
        \EndIf
    \EndFor
\EndFor

\For{$i=1,\dots,n$}\Comment{zero out diagonals of $\mat{\tilde{Z}}$}
    \If{$\tilde{z}_{ii}\neq0$}
        \State \Call{apply\_S}{$\mat{\tilde{Z}},\mat{\tilde{X}},\mat{r},i$}
        \State $b_i:=b_i+1$
    \EndIf
\EndFor

\State Compute the inverse permutation $\mat{p}^{-1}$ of $\mat{p}$\Comment{read as $\text{new}[i]=\text{old}[\mat{p}^{-1}[i]]$}
\State Permute $\mat{a},\mat{b},\mat{r}$ and both rows and columns of $\mat{\tilde{Z}}$ with $\mat{p}^{-1}$
\State Let $\mat{c}:=\mat{r}$ and $G$ be the graph represented by adjacency matrix $\mat{\tilde{Z}}$
\State \Return $G$, $\mat{a},\mat{b},\mat{c}$
\end{algorithmic}
\end{algorithm}

\section{A canonical form of one-qubit Clifford gate}\label{sec:one-qubit-clifford}
One-qubit Clifford gates can be generated by $\set{\check{Z},\check{X}}$. Note that $\check{Z},\check{X}$ can be viewed as $-90$ degree rotations around z- or x-axis in Bloch sphere. Therefore, any such rotation can be represented in form $\check{Z}^a\check{X}^b\check{Z}^c$ for $a,b,c\in\Z_4$ using Euler angles of the rotation as presented in the following lemma.
\begin{lemma}\label{lem:one-qubit-clifford}
Any one-qubit Clifford gate can be represented, up to a global phase, as $\check{Z}^a\check{X}^b\check{Z}^c$ for some $a,b,c\in\Z_4$; equivalently, the $64$ words $\check{Z}^a\check{X}^b\check{Z}^c$ realize all $24$ projective classes of one-qubit Clifford gates.
\end{lemma}
\begin{proof}
Consider an arbitrary one-qubit Clifford gate $U=\prod_i(\check{X}^{x_i}\check{Z}^{z_i})$ for $x_i,z_i\in\set{0,1,2,3}$, which is a string of powers of $\check{Z}$ or $\check{X}$. For any sub-string of form $\check{X}^{a}\check{Z}^{b}\check{X}^{c}$, one can apply one of the equalities below to replace it with $\check{Z}^{a'}\check{X}^{b'}\check{Z}^{c'}$.
\begin{align}
\check{X}\check{Z}^2\check{X}&=-i\check{Z}^2
\label{eq:xzx-zxz-1}\\
\check{X}^2\check{Z}^b\check{X}^2&=(-i)^b\check{Z}^{b-2},\;b\in\set{1,3}\label{eq:xzx-zxz-2}\\
\check{X}^a\check{Z}^b\check{X}^c&=\begin{cases}
-\check{Z}^c\check{X}^b\check{Z}^a,&(a,b,c)\in\set{(1,3,1),(3,1,3)},\\
\check{Z}^c\check{X}^b\check{Z}^a,&\text{otherwise},
\end{cases}\quad a,b,c\in\set{1,3}\label{eq:xzx-zxz-3}
\end{align}
Each identity above is applied below without tracking the scalar prefactor: 
$(-i)$ in \cref{eq:xzx-zxz-1}, $(-i)^b$ in \cref{eq:xzx-zxz-2}, and $\pm1$ in \cref{eq:xzx-zxz-3}. Each such scalar is a fourth root of unity, so the canonical form is obtained only up to a global phase.
More specifically,
\begin{itemize}
\item if any of $a,b,c$ is 0, then it is trivial;
\item if $b=2$, applying \cref{eq:xzx-zxz-1} for $a$ times gives $\check{Z}^{b}\check{X}^{c-a}$;
\item if $b\in\set{1,3}$ and either $a$ or $c$ is even, applying \cref{eq:xzx-zxz-2} for $a/2$ times (if $a$ is even) or $c/2$ times (if $c$ is even) gives $\check{Z}^{b-a}\check{X}^{c-a}$ or $\check{X}^{a-c}\check{Z}^{b-c}$ respectively;
\item if $a,b,c\in\set{1,3}$, applying \cref{eq:xzx-zxz-3} directly gives the desired form.
\end{itemize}
This process reduces the number of nonconsecutive $\check{X}$ in the string by one and we can repeat it until there is only one $\check{X}$ in the string, which has the form $\check{Z}^a\check{X}^b\check{Z}^c$ for $a,b,c\in\Z_4$.
\end{proof}

\end{document}